\documentclass[a4paper, 12pt]{article}

\usepackage[sort&compress]{natbib}
\bibpunct{(}{)}{;}{a}{}{,} 

\usepackage{amsthm, amsmath, amssymb, mathrsfs, multirow, url, subfigure}
\usepackage{graphicx} 
\usepackage{ifthen} 
\usepackage{amsfonts}
\usepackage[usenames]{color}
\usepackage{fullpage}
\usepackage{tikz}

\theoremstyle{plain} 
\newtheorem{thm}{Theorem}
\newtheorem{cor}{Corollary}

\theoremstyle{definition}

\newtheorem{ex}{Example}

\theoremstyle{remark}

\newcommand{\prob}{\mathsf{P}} 
\newcommand{\E}{\mathsf{E}}

\newcommand{\bin}{{\sf Bin}}
\newcommand{\unif}{{\sf Unif}}
\newcommand{\nm}{{\sf N}}

\newcommand{\gam}{{\sf Gamma}}

\newcommand{\chisq}{{\sf ChiSq}}
\newcommand{\mult}{{\sf Mult}}

\newcommand{\RR}{\mathbb{R}}

\newcommand{\FF}{\mathbb{F}}
\newcommand{\LL}{\mathbb{L}}

\newcommand{\YY}{\mathbb{Y}}
\newcommand{\ZZ}{\mathbb{Z}}

\newcommand{\TT}{\mathbb{T}}

\newcommand{\iid}{\overset{\text{\tiny iid}}{\,\sim\,}}

\newcommand{\prior}{\mathsf{Q}}
\newcommand{\credal}{\mathscr{Q}}
\newcommand{\cred}{\mathscr{C}}

\newcommand{\lPi}{\underline{\Pi}}
\newcommand{\uPi}{\overline{\Pi}}

\newcommand{\uprior}{\overline{\mathsf{Q}}}

\newcommand{\lprob}{\underline{\mathsf{P}}}
\newcommand{\uprob}{\overline{\mathsf{P}}}

\title{Valid and efficient imprecise-probabilistic inference with partial priors, III. Marginalization}
\author{Ryan Martin\footnote{Department of Statistics, North Carolina State University, {\tt rgmarti3@ncsu.edu}}}
\date{\today}

\begin{document}

\maketitle 

\begin{abstract}
As Basu (1977) writes, ``Eliminating nuisance parameters from a model is universally recognized as a major problem of statistics,'' but after more than 50 years since Basu wrote these words, the two mainstream schools of thought in statistics have yet to solve the problem.  
Fortunately, the two mainstream frameworks aren't the only options.  This series of papers rigorously develops a new and very general inferential model (IM) framework for imprecise-probabilistic statistical inference that is provably valid and efficient, while simultaneously accommodating incomplete or partial prior information about the relevant unknowns when it's available.  The present paper, Part~III in the series, tackles the marginal inference problem.  Part~II showed that, for parametric models, the likelihood function naturally plays a central role and, here, when nuisance parameters are present, the same principles suggest that the profile likelihood is the key player.  When the likelihood factors nicely, so that the interest and nuisance parameters are perfectly separated, the valid and efficient profile-based marginal IM solution is immediate.  But even when the likelihood doesn't factor nicely, the same profile-based solution remains valid and leads to efficiency gains.  This is demonstrated in several examples, including the famous Behrens--Fisher and gamma mean problems, where I claim the proposed IM solution is the best solution available.  Remarkably, the same profiling-based construction offers validity guarantees in the prediction and non-parametric inference problems.  Finally, I show how a broader view of this new IM construction can handle non-parametric inference on risk minimizers and makes a connection between non-parametric IMs and conformal prediction.  

\smallskip

\emph{Keywords and phrases:} Bayesian; frequentist; inferential model; non-parametric; nuisance parameter; possibility theory; prediction; profile likelihood. 
\end{abstract}

\vfill 

\pagebreak

\tableofcontents

\pagebreak

\section{Introduction}
\label{S:intro}

In \citet{martin.partial2}, henceforth Part~II, I developed a new framework for valid and efficient (imprecise-probabilistic) statistical inference that incorporates general incomplete or partial prior specifications.  In particular, it simultaneously covers the classical frequentist case of vacuous prior information and the classical Bayesian case of complete prior specification.  Like Bayes, it provides a sort of ``probabilistic'' uncertainty quantification---that is, to each relevant hypothesis a data-dependent numerical score is assigned that represents its support/plausibility---but, unlike Bayes, this isn't done with probability theory and Bayes's formula.  Moreover, the new framework comes equipped with both reliability and coherence-like properties, so it achieves both the classical frequentist and Bayesian objectives.  The present paper is a follow-up to Part~II that addresses the common yet non-trivial situation in which the quantity of interest is just one feature of the full unknown.  In other words, this paper focuses on cases where nuisance parameters are present and need to be eliminated for valid and efficient {\em marginal} inference. 

What's unique about this framework overall is that inference is {\em imprecise-probabilistic} (actually, {\em possibilistic}---more on this below) in the sense that, in light of data and other relevant information, uncertainty about the unknowns is quantified in terms of an imprecise probability, or a lower and upper probability pair.  This ``imprecision'' isn't a shortcoming of the approach, or an undesirable feature that warrants an apology, it's necessary for the kinds of reliability that statisticians expect their methods to satisfy.  A detailed justification of this claim is given in \citet{martin.nonadditive, imchar}, but here let me make a couple high-level points.  When we teach the classical frequentist methods that are reliable by definition, we're careful not to express these in terms of probability statements about the unknown parameter.  For example, we emphasize that 
\begin{itemize}
\item the p-value is not the conditional probability, given data, that the true parameter value meets the criteria specified by the null hypothesis, and 
\vspace{-2mm}
\item the confidence level is not the conditional probability, given data, that the true parameter value is contained in the stated interval. 
\end{itemize} 
\citet{fisher1973} himself commented on this:
\begin{quote}
{\em [A p-value] is more primitive, or elemental than, and does not justify, any exact probability statement about the proposition} (ibid, p.~46)
\end{quote}
\begin{quote}
{\em ... It is clear, however, that no exact probability statements can be based on [confidence sets]} (ibid, p.~74).
\end{quote}
Note that Fisher says no {\em exact}---or {\em precise}---{\em probability statements} can be made based on these classical procedures, but he leaves open the prospect that they offer certain {\em inexact} or {\em imprecise probability statements}.  The proposed imprecise-probabilistic framework is simply embracing Fisher's points and trying to squeeze all that we possibly can out of these classical ideas.  Of course, not just any kind of imprecise probability would be compatible with p-values/confidence sets and the properties they're intended to satisfy, but it turns out that the consonant \citep[e.g.,][Ch.~10]{shafer1976} or possibilistic \citep[e.g.,][]{dubois2006, dubois.prade.book} brand of imprecision is both simple and ``right'' for this task; see \citet{imchar} and Section~3.2 of Part~II.  And as shown in Part~II, the examples below, and elsewhere \citep[e.g.,][]{imbasics, imcond, immarg, imbook}, this imprecision {\em need not} come with any sacrifice in efficiency, but some care is needed.  The goal of this paper is to flesh out what I mean by ``care'' when the task is marginal inference. The key point, again, is that a certain kind and degree of imprecision is necessary to guarantee statisticians' desired reliability and, by the false confidence theorem \citep{balch.martin.ferson.2017, martin.nonadditive}, attempts to push the model + data pair beyond its imprecise limits, as the default-prior Bayes and fiducial solutions do, create risks for severe unreliability; further discussion and references on this can be found in \citet{martin.basu}.  

As mentioned above, the focus here is on valid and efficient marginal inference.  \citet{basu1977} writes ``Eliminating nuisance parameters from a model is universally recognized as a major problem of statistics'' and yet it remains unsolved.  That is, reliable elimination of nuisance parameters is a challenging problem that requires significant care; see, e.g., the frequentist impossibility results in \citet{gleser.hwang.1987} and \citet{dufour1997} and the unreliable behavior of Bayesian solutions in, e.g., \citet{fraser2011} and \citet{fraser.etal.2016}.  A benefit of the general, imprecise-probabilistic framework developed in this series is that it offers a reliable-but-naive strategy (Section~\ref{SS:naive}) for eliminating nuisance parameters.  The aforementioned solution is ``reliable'' in the sense that the corresponding marginal inference would be valid (i.e., confidence sets attain the nominal coverage probability), but ``naive'' in the sense that it's not tailored to any specific feature and, therefore, inference would tend to be inefficient for any particular feature. So the goal here is to avoid sacrificing on  validity or efficiency---I want to tailor the IM solution to a particular feature of interest so that the corresponding inference is both valid and efficient.  

An early IM solution to the marginal inference problem was put forward in \citet{immarg}, but this was limited by, for one thing, its reliance on an expression of the statistical model in terms of a (simple, relatively easily manipulated) data-generating equation.  The new framework developed in Part~II is likelihood-driven, rather than data-generating equation-driven, and, therefore, can readily be applied to a wider range of problems.  Moreover, being likelihood-driven makes it possible to take advantage of certain structure, i.e., factorizations, in the likelihood function, which can aid in the elimination of nuisance parameters.  For example, in certain ``ideal'' cases (Section~\ref{SS:factor}), the likelihood factors in such a way that one term depends only on the interest parameter and the other only on the nuisance parameter.  In such cases, as I show in Section~\ref{SS:more} below, it's relatively easy to construct a valid and efficient marginal IM.  The way that I suggest to take advantage of this special structure is via a {\em profiling} step, where the likelihood function is maximized over the nuisance parameter with the interest parameter fixed; when the likelihood function factors in an ``ideal'' way, certain key terms involving the nuisance parameter cancel out in the relative profile likelihood and efficient marginalization can be achieved.  This is illustrated in several important, classical examples.  Most importantly, the same profiling strategy leads to valid and efficient marginal inference in most (but not all---see Section~\ref{SS:caution}) cases that are ``less-than-ideal'' in one way or another.  I show, in Section~\ref{SS:big.examples} that the profiling-based IM construction leads to valid and efficient marginal inference in two challenging and practically relevant examples, namely, the Behrens--Fisher and gamma mean problems.  In fact, to my knowledge, these are the best available solutions that are exactly---not just asymptotically approximately---valid.  

An important, albeit extreme case of marginal inference is prediction, where the model parameter itself is a nuisance parameter and only a future observation is of interest.  Section~\ref{S:prediction} tackles the case of a parametric statistical model and the goal is valid and efficient inference on/prediction of some feature of future observables.  There are at least three different ways prediction can be carried out in this framework, and I consider each of these in turn.  Interestingly enough, the same profiling strategy described above can be used in this context as well, and tends to produce the most efficient predictive inference.  These ideas are illustrated in a number of non-trivial examples, including predicting the largest of $k$ many future gamma observations.  

The focus so far has been on problems that come equipped with a parametric model amenable to a likelihood-based analysis.  But there are important, even classical problems that don't fit this mold, e.g, inference on the mean of an otherwise unspecified distribution.  Towards an IM solution to this problem, but staying relatively close to the theory developed so far, I proceed in Section~\ref{S:semi} by treating the model itself as the parameter, forming a sort of ``empirical likelihood'' as developed by \citet{owen.book} and others, then applying the same profiling strategy in hopes of eliminating all but the quantity of interest.  In this higher-complexity context, computation of the IM's lower and upper probabilities, etc., is more challenging, so I offer some suggestions on how to approximate these (e.g., using bootstrap) and a few illustrations. 

Finally, in Section~\ref{S:beyond}, I take a different perspective on the non-parametric problem mentioned above, one that avoids both thinking about and constructing/using a likelihood.  In parametric models, it's the likelihood function that plays the (very important) role of mathematically linking the observable data and the quantity of interest, but presumably there'd be other more direct ways to make this link in non-parametric cases.  Further investigation along these lines is needed, but I consider two relevant problems, namely, inference on risk minimizers and (non-parametric) prediction of future observations.  One highlight of this investigation is that I'm able to show how the now-widely-used {\em conformal prediction} methodology \citep[e.g.,][]{vovk.shafer.book1} can actually be derived by applying some of the key principles in Part~II to this broader notion of an inferential model.  

The paper concludes with a brief summary and a discussion of some open problems and directions for future investigation.

\section{Background}

\subsection{Recap of Part~II}

Part~II of the series put forward a general IM construction that can accommodate partial prior information about the model parameter $\Theta$, if any, and returns a necessity--possibility measure pair as output to be used for uncertainty quantification.  Here I give a relatively quick recap of this construction and the relevant properties.  

First a bit of notation.  The statistical model is a family of probability distributions for the observable data $Y$, which I'll write as $\{\prob_{Y|\theta}: \theta \in \TT\}$.  Note that the subscript indicates which quantity is random/uncertain, with dependence on, in this case, a parameter $\theta$ being marked by the vertical bar.  I'll assume that, for each $\theta$, $\prob_{Y|\theta}$ admits a density/mass function, denoted by $y \mapsto p_{Y|\theta}(y)$.  Moreover, I'll write upper-case $\Theta$ for the uncertain value of the model parameter that's to be inferred.  Write $\text{probs}(\TT)$ for the set of all probability measures defined on the Borel $\sigma$-algebra of $\TT$.  

The approach developed in Part~II allows for various kinds of prior information about $\Theta$ to be incorporated.  This includes the traditional Bayesian case where a single prior distribution that completely and precisely quantifies the {\em a priori} uncertainty about $\Theta$, as well as cases where the available prior information is incomplete or imprecise to some degree, even vacuous.  Mathematically, this can be described in general by a lower and upper probability pair $(\lprob_\Theta, \uprob_\Theta)$ that quantifies the available prior information about $\Theta$.  The lower and upper probabilities are related via the duality
\begin{equation}
\label{eq:dual}
\uprob_\Theta(A) = 1 - \lprob_\Theta(A^c), \quad A \subseteq \TT. 
\end{equation}
I'll assume throughout that this {\em a priori} assessment is coherent in the sense of, e.g., \citet[][Ch.~2.5]{walley1991}, \citet[][Sec.~2.2.1]{miranda.cooman.chapter}, and \citet[][Def.~4.10]{lower.previsions.book}.  The technical definition of coherence isn't relevant here, but there are a few key consequences that are worth mentioning.  
\begin{itemize}
\item The lower probability $A \mapsto \lprob_\Theta(A)$ is 2-monotone which, in particular, implies that it's also super-additive, i.e., 
\[ \lprob_\Theta(A \cup B) \geq \lprob_\Theta(A) + \lprob_\Theta(B), \quad A \cap B = \varnothing. \]
Since $\lprob_\Theta(A \cup A^c) = 1$, it follows from \eqref{eq:dual} and super-additivity that
\[ \lprob_\Theta(A) \leq \uprob_\Theta(A), \quad A \subseteq \TT. \]
This explains the lower/upper terminology and the under/over-bar notation.
\item De~Finetti's school treats probabilities as fully subjective, and their real-world meaning and interpretation is teased out through a sort of game where you and I are able buy and sell gambles to one another.  In the present imprecise case, this is roughly as follows: \$$\lprob_\Theta(A)$ is the most that I'd be willing to pay you for a gamble that pays me \$$1(\Theta \in A)$ and, similarly, \$$\uprob_\Theta(A)$ is the least that I'd be willing to accept from you in exchange for the gamble that pays you \$$1(\Theta \in A)$.  That is, $\lprob_\Theta$ and $\uprob_\Theta$ bound my buying and selling prices, respectively.  If the imprecise probability that drives my pricing scheme is coherent, then I cannot be made a sure loser, i.e., there is no finite sequences of transactions for which my net winnings is sure to be negative.  This no-sure-loss property is relatively weak, but if I fail to avoid sure-loss, that that's a clear sign my probability assessments are flawed.  
\item The imprecise prior assessment is equivalent to a set of precise probability assessments.  That is, the upper prior probability $\uprob_\Theta$ determines a (closed and convex) set of compatible probabilities, called a {\em credal set}, given by 
\[ \cred(\uprob_\Theta) = \{ \prob_\Theta \in \text{probs}(\TT): \prob_\Theta(\cdot) \leq \uprob_\Theta(\cdot) \}, \]
and that set, in turn, determines the lower and upper probabilities as its corresponding lower and upper envelopes:
\[ \lprob_\Theta(A) = \inf_{\prob_\Theta \in \cred(\uprob_\Theta)} \prob_\Theta(A) \quad \text{and} \quad \uprob_\Theta(A) = \sup_{\prob_\Theta \in \cred(\uprob_\Theta)} \prob_\Theta(A), \quad A \subseteq \TT. \]
\end{itemize} 
These same properties hold for all the coherent imprecise probabilities discussed below, there's nothing mathematically special about the imprecise {\em prior} probabilities. 

The most common situation in the statistics literature is where no prior information is assumed, i.e., all that can be said {\em a priori} is that the ``prior probability of $\Theta \in A$'' is between 0 and 1.  This can't be modeled with ordinary probability, but it's easy to handle with imprecise probability: the corresponding lower and upper prior probabilities would be $\lprob_\Theta(A) = 0$ for all $A \neq \TT$ and $\uprob_\Theta(A) = 1$ for all $A \neq \varnothing$.  It's also easy to see that the credal set $\cred(\uprob_\Theta)$ corresponds to the set of all probability distributions.  This observation offers an interesting take-away message: the classical ``no prior'' case is more accurately described as ``every prior'' in the sense that the lack of prior information available actually means that one can't rule out any prior distributions; see Part~I \citep{martin.partial}.  This so-called vacuous prior case will be my primary focus in this paper, mostly for the sake of comparison with existing solutions in the key examples.  

The partial prior and the statistical model together determine an imprecise joint distribution $(\lprob_{Y,\Theta}, \uprob_{Y,\Theta})$ for the pair $(Y,\Theta)$.  The upper joint distribution $\uprob_{Y,\Theta}$ is 
\[ \uprob_{Y,\Theta}(Y \in B, \, \Theta \in A) = \sup_{\prob_\theta \in \cred(\uprob_\Theta)} \int_A \prob_{Y|\theta}(B) \, \prob_\Theta(d\theta), \quad A \subseteq \TT, \quad B \subseteq \YY. \]
The right-hand side above is a Choquet integral \citep[e.g.,][App.~C]{lower.previsions.book}---which is familiar in certain statistical contexts \citep[e.g.,][]{huber1973.capacity}---and there may be simplified expressions depending on the mathematical form of the partial prior; see Equation \eqref{eq:choquet.theta} below and Section~6.1 in Part~II.  In any case, the upper joint distribution, which depends on {\em exactly} what the data analyst knows or is willing to assume about the application at hand, is what drives the construction of an IM for quantification of uncertainty about $\Theta$ given the observed $Y=y$.  Following some lengthy justification in terms of what I called ``outer consonant approximations,'' I arrived at the following IM construction, given $Y=y$: first, define the (plausibility) contour function 
\begin{equation}
\label{eq:im.contour}
\pi_y(\theta) = \uprob_{Y,\Theta}\{ R_q(Y,\Theta) \leq R_q(y,\theta)\}, \quad \theta \in \TT, 
\end{equation}
where $R_q$ is a sort of normalized joint density 
\[ R_q(y,\theta) = \frac{p_{Y|\theta}(y) \, q_\Theta(\theta)}{\sup_{\vartheta \in \TT} \{p_{Y|\vartheta}(y) \, q_\Theta(\vartheta)\}}, \quad \theta \in \TT, \]
with $q_\Theta(\theta) := \uprob_\Theta(\{\theta\})$, a relevant summary of the partial prior.  If $q_\Theta$ satisfies $\sup_\theta q_\Theta(\theta) = 1$, which is within the user's control, then the right-hand side of \eqref{eq:im.contour} can be evaluated as 
\begin{equation}
\label{eq:choquet.theta}
\pi_y(\theta) = \int_0^1 \sup_{\vartheta: q_\Theta(\vartheta) > s} \prob_{Y|\vartheta}\{ R_q(Y,\vartheta) \leq R_q(y,\theta) \} \, ds, \quad \theta \in \TT, 
\end{equation}
with, for example, the inner $\prob_{Y|\vartheta}$-probability evaluated via Monte Carlo.  Since the contour \eqref{eq:im.contour} clearly satisfies $\sup_\theta \pi_y(\theta) = 1$ for each $y$, this contour can be used to directly define the IM's upper probability via consonance 
\[ \uPi_y(A) = \sup_{\theta \in A} \pi_y(\theta), \quad A \subseteq \TT, \]
and the corresponding lower probability via the general duality \eqref{eq:dual}.  This is a generalization of the suggestion in \citet{plausfn, gim}, and a detailed justification for this is presented in Part~II.  The IM output is a coherent imprecise probability, so those properties described above for $(\lprob_\Theta, \uprob_\Theta)$ also hold for $(\lPi_y, \uPi_y)$.  Moreover, the $\uPi_y$ term in the IM's output satisfies the properties of a {\em possibility measure}, so I will often refer to this as a {\em possibilistic IM} and inference drawn from it {\em possibilistic inference}.  While it might not look it at first glance, this solution is actually quite straightforward: the data and model are combined with available prior information via the rule \eqref{eq:im.contour}.  The key take-away, however, is that the IM output $(\lPi_y, \uPi_y)$---which depends on the data $y$, the posited model, and the available prior information---is special because it's completely determined by the contour function \eqref{eq:im.contour}.  Indeed, like how a Bayesian's posterior density determines everything, the contour function \eqref{eq:im.contour} determines the IM solution; the only difference is that I optimize the contour function whereas the Bayesian integrates the density function.  See Section~3.2 of Part~II for further details concerning this special (consonance) structure.  Connections can also be made, in the vacuous prior case, between the IM output's credal set $\cred(\uPi_y)$ and Fisher's fiducial solution \citep[e.g.,][]{martin.isipta2023}. 

The value in/quality of the IM solution lies in the properties that it satisfies.  One property is related to, but more demanding than, the coherence property described above.  This concerns the act of updating probabilities/prices based on new information, in our case, the data $y$.  The idea is my probabilities/prices should be such that you can't make me a sure-loser by leveraging some inadequacy in how I update my prior assessments.  As I explain in Section~3.3 of Part~I and Section~5.2.2 of Part~II, and won't repeat in details here, the IM solution described above comes equipped with protection against this type of updating sure-loss as well.  Some believe that, by de~Finetti's theory, only Bayesian solutions are coherent, but the above result largely debunks this folklore.  

More directly relevant to the discussion in this paper are the IM solution's statistical properties.  One basic property, called (strong) {\em validity} (Definition~3 in Part~I) has lots of practically relevant consequences.  

\begin{thm}[Part~II]
\label{thm:valid}
The IM with contour as in \eqref{eq:im.contour} is (strongly) valid in the sense that  
\begin{equation}
\label{eq:valid}
\uprob_{Y,\Theta}\{ \pi_Y(\Theta) \leq \alpha \} \leq \alpha, \quad \alpha \in [0,1].
\end{equation}
\end{thm}

This closely resembles the familiar property satisfied by p-values in the context of statistical significance testing, but generally is different.  In the case of vacuous prior information, the ``prior'' admits $q_\Theta(\theta) \equiv 1$ and validity boils down to 
\[ \sup_{\theta \in \TT} \prob_{Y|\theta}\{ \pi_Y(\theta) \leq \alpha \} \leq \alpha, \quad \alpha \in [0,1]. \]
This looks even closer to the stochastically-no-smaller-than-uniform property satisfied by p-values.  The key point is that validity ensures the IM output is suitably calibrated, so that inferences based on the magnitudes of the IM's lower and upper probabilities are reliable.  As a consequence of this kind of high-level reliability, one can establish (Corollary~1 in Part~II) more mathematically specific results for IM-driven statistical procedures.  In particular, the set estimator 
\[ C_\alpha(Y) = \{\theta \in \TT: \pi_Y(\theta) > \alpha\}, \quad \alpha \in [0,1], \]
which can be interpreted as ``the set of all sufficiently plausible values'' of $\Theta$ satisfies 
\[ \uprob_{Y,\Theta}\{ C_\alpha(Y) \not\ni \Theta \} \leq \alpha. \]
That is, the set estimator $C_\alpha$ is a $100(1-\alpha)$\% confidence set, but in a partial prior-dependent sense through the evaluation via $\uprob_{Y,\Theta}$.  In the vacuous prior case, this reduces to the usual coverage probability guarantees:
\[ \sup_{\theta \in \TT} \prob_{Y|\theta}\{ C_\alpha(Y) \not\ni \theta \} \leq \alpha. \]
Take-away: this IM construction accommodates very general forms of partial prior information in a way that's coherent in various senses, and does so without sacrificing on the reliability properties that are essential to the logic of scientific inference.

\subsection{Interest and nuisance parameters}

It's rare that the quantity of interest exactly corresponds to the unknown parameters of the posited statistical model for the observable data $Y$.  A more realistic situation is one where the quantity of interest is some functional or feature of the full model parameter.  That is, if the model is $\{\prob_\theta: \theta \in \TT\}$, with $\Theta$ the uncertain value, then interest would often be in one or more features $\Phi = f(\Theta)$ of $\Theta$, where $f: \TT \to \FF$ is a known mapping.  One of the most common example of this is where the data are assumed to be normally distributed, where both the mean and variance are uncertain, but the goal is inference on the mean only.  In such cases, it's often possible to decompose the full parameter $\Theta$ as a pair $(\Phi, \Lambda)$, where the quantity of interest $\Phi$ is the {\em interest parameter}, taking values in $\FF$, and $\Lambda$ is the {\em nuisance parameter}, taking values in $\LL$.  The residual feature $\Lambda$ is only relevant for reconstructing the full parameter $\Theta$ from $\Phi$.  Details of the IM construction in this setting will be relevant for the construction of IMs in prediction problems in Section~\ref{S:prediction} and in modern non- and semi-parametric inference problems in Section~\ref{S:semi} below.  


It's important to emphasize that, despite what the notation suggests, $\Phi$ need not be a sub-vector of the vector model parameter $\Theta$.  In general, the notation $\Theta=(\Phi,\Lambda)$ is meant to indicate that the model parameter $\Theta$ determines and is determined by the pair $(\Phi, \Lambda)$.  For example, if $\Theta$ is a vector and $\Phi = \|\Theta\|$ is its length, then $\Lambda = \Theta / \|\Theta\|$ is the unit vector pointing in the direction of $\Theta$.  All I'm assuming is that $\Theta$ and $(\Phi,\Lambda)$ are in one-to-one correspondence.  This is important because it'll often be the case that the quantity of interest $\Phi$ isn't specific to any particular statistical model, e.g., $\Phi$ might be a quantile or the coefficients that determine a linear conditional quantile function.  So, if I impose a statistical model that's parametrized by $\Theta$, then my quantity of interest becomes a general feature of the model parameter $\Theta$, not necessarily a sub-component thereof.  Furthermore, this structure also suggests that genuine partial prior information about $\Phi$ would often be available in applications, whereas prior information about the nuisance parameter $\Lambda$ would be vacuous.  This kind of ``partial prior factorization'' will be useful in what follows.

\subsection{Classical marginalization}
\label{SS:factor}

There are a number of ways to carry out marginalization, depending on the statistical paradigm one is working in.  One of the selling points of the Bayesian paradigm is that marginalization is at least conceptually straightforward---it's just an application of probability theory.  I'll explain in the next subsection that there's a counterpart to this for possibilistic IMs, that's similarly straightforward, but it tends to be inefficient; the purpose of this paper is to explain how to do this more efficiently.  Here I'll present some classical ideas about how, in certain cases, the likelihood function factors in a way that makes elimination of the nuisance parameters fairly convenient.  My presentation will be based largely on the survey presented in \citet{basu1977, basu1978}, which is based on \citet{neyman1935}, \citet{olshevsky1940}, \citet{fraser1956}, \citet{sandved1966}, and \citet{bn.thesis}.  

Certain models and interest-nuisance parameter decompositions allow for a convenient factorization of the likelihood function that suggests a strategy for eliminating the nuisance parameter.  Recall that the likelihood function for the full parameter $\Theta=(\Phi,\Lambda)$, given $Y=y$, is $\theta \mapsto p_{Y|\theta}(y)$, which I'll write as $(\phi,\lambda) \mapsto p_{Y|\phi,\lambda}(y)$ to emphasize the interest-nuisance parameter decomposition.  In what follows, I'll slightly abuse notation by using ``$p$'' to represent all the (marginal and conditional) densities.  

\begin{description}
\item[Ideal factorization.] {\em Complete separation of $\phi$ and $\lambda$.}

The idea here is that the likelihood factors as a function depending on $\phi$ (and data) times a function of $\lambda$ (and data).  \citet[][Ch.~7]{royall.book} refers to this as parameter orthogonality; see, also, \citet{anscombe1964}.  There are a number of different ways in which parameter orthogonality might manifest, and below I describe a few.  Let $(U,V)$ denote a generic partition of the data $Y$, so that $y$ is equivalent to the pair $\{U(y), V(y)\}$, and consider the following factorizations:
\begin{align*}
p_{Y|\phi,\lambda}(y) & = p_{U|\phi}(u) \, p_{V|u,\lambda}(v) \\ 
p_{Y|\phi,\lambda}(y) & = p_{U|\lambda}(u) \, p_{V|u,\phi}(v), \quad (u,v) = \{U(y),V(y)\}. 
\end{align*}
These represent factorizations of the joint distribution of $Y$ in terms of marginal and conditional distributions of the features $U(Y)$ and $V(Y)$.  What differentiates the two is whether the interest parameter $\phi$ goes with the marginal or the conditional, and I'll consider both cases below in turn.  As \citet{basu1977} explains, the first case above is one where $U$ is P-sufficient for $\Phi$---``P'' for ``partial''---and is S-ancillary for $\Lambda$---``$S$'' for ``Sandved.''  The point is that $U=U(Y)$ is exhaustive concerning $\Phi$ since the conditional distribution of $V$, given $U=u$, doesn't depend on $\phi$, which aligns closely with the classical definition of sufficiency.  Similarly, the marginal distribution of $U$ doesn't depend on $\phi$, so it's ancillary in a certain sense.  In this case, if inference on $\Phi$ is the goal, then one can safely ignore the $\lambda$-dependent term and work with the marginal likelihood, $\phi \mapsto p_{U|\phi}(u)$.  In the second case above, $U$ is S-ancillary for $\Phi$ and P-sufficient for $\Lambda$, so one can safely ignore the $\lambda$-dependent term and work with the conditional likelihood, $\phi \mapsto p_{V|u,\phi}(v)$. 
\item[Less-than-ideal factorization.] {\em Incomplete separation of $\phi$ and $\lambda$.}

Here, consider the factorizations 
\begin{align*}
p_{Y|\phi,\lambda}(y) & = p_{U|\phi}(u) \, p_{V|u,\phi,\lambda}(v) \\
p_{Y|\phi,\lambda}(y) & = p_{U|\phi,\lambda}(u) \, p_{V|u,\phi}(v), \quad (u,v) = \{U(y), V(y)\}.
\end{align*}
Note the incomplete separation: unlike above, here there is no factorization into a function of $\phi$ (and data) times a function of $\lambda$ (and data).  There is a partial split, however.  In the first case, $U$ is what Basu would call $\Phi$-oriented whereas, in the second case, Basu would say that $U$ is specific-sufficient for $\Phi$, i.e., that $U$ is sufficient for $\Phi$ if $\Lambda=\lambda$ was taken as known.  Like above, if inference on $\Phi$ is the goal, then one could choose to work with the marginal or conditional likelihood in the two cases, respectively.  But this is not an obvious step like above because, here, ignoring the other factor implies some loss of information about $\Phi$
\end{description}
There is, of course, no guarantee that every problem would fit into one of these two categories.  Fortunately, the relatively simple strategy suggested by the ideal factorization works quite well---in the sense of improving efficiency in marginal inference---even outside the ideal factorization case.  

A notion that will prove to be quite useful in what follows is {\em profiling}.  If $\theta \mapsto p_\theta(y) = p_{Y|\phi,\lambda}(y)$ is the likelihood function for the pair $\theta=(\phi,\lambda)$ based on data $Y=y$, then the {\em profile likelihood function} for the interest parameter $\phi$ is determined by maximizing over the nuisance parameter $\lambda$ for fixed $\phi$; that is, the profile likelihood is 
\[ \phi \mapsto \sup_{\lambda \in \LL} p_{Y|\phi,\lambda}(y), \quad \phi \in \FF. \]
Of course, the profile is not a genuine likelihood in the sense that it doesn't correspond to a density in $y$ that's being treated as a function of $\phi$.  But it does capture the property that a likelihood function is supposed to have, namely, that it provides a meaningful ranking of the $\phi$ values in terms of how well they explain the data $y$; it does so in a very optimistic way, i.e., assigning to $\phi$ the rank corresponding to $(\phi, \hat\lambda_y(\phi))$ with $\hat\lambda_y(\phi)$ the ``best'' companion to $\phi$ for the given $y$, the so-called conditional maximum likelihood estimator.  Below I'll show that the profile likelihood is a powerful tool to help guide efficient marginal inference on $\Phi$ in a wide range of problems.  

Note, however, that despite the wide range of problems in which profiling will lead to efficient marginal inference, there are cases in which profiling can be quite inefficient.  Fortunately, those problematic cases have a feature in common that we can easily spot and then modify our approach accordingly.

\section{Marginal possibilistic IMs}
\label{S:nuisance}

\subsection{Naive solution}
\label{SS:naive}

Since the IM framework returns a data-dependent, imprecise probability $(\lPi_y, \uPi_y)$, one has the option to carry out marginalization simply using the available imprecise probability calculus.  This is exactly how the Bayesian framework proceed: get a posterior distribution for the full parameter $\Theta$, then marginalizing using the ordinary probability calculus/integration.  Of course, the imprecision baked into the IM output changes the technical details (of the probability calculus), but not the intuition.  

In the imprecise probability literature, the notion of taking an uncertainty quantification about some quantity, say $\Theta$, and mapping it to an uncertainty quantification about a different quantity, say $\Phi$, is often referred to as {\em extension}.  It's arguably somewhat of a misnomer to refer to marginalization as ``extension''---since, in this case, $\Phi=f(\Theta)$ is actually contained in $\Theta$---but, nevertheless, the natural/naive approach to marginalization is to apply the available extension techniques.  Since the IM output is a necessity--possibility measure pair, the appropriate extension is based on the so-called {\em extension principle} of \citet{zadeh1975a, zadeh1978}; see, also, \citet[][Sec.~3.2.3]{hose2022thesis}.  In particular, the extension principle applied to the present marginalization task produces a marginal possibilistic IM that's determined by the contour function 
\begin{equation}
\label{eq:naive.contour}
\pi_y^f(\phi) = \sup_{\theta: f(\theta)=\phi} \pi_y(\theta), \quad \phi \in \FF. 
\end{equation}
Then the corresponding upper probability $\uPi_y^f$ for quantifying uncertainty about $\Phi$ is
\[ \uPi_y^f(B) = \sup_{\phi \in B} \pi_y^f(\phi), \quad B \subseteq \FF. \]
The lower probability is defined via \eqref{eq:dual}.  To see that this is completely consistent with the possibilistic IM for $\Theta$, we can elaborate the right-hand side above as 
\[ \uPi_y^f(B) = \sup_{\phi \in B} \pi_y^f(\phi) = \sup_{\theta: f(\theta) \in B} \pi_y(\theta) = \uPi_y\{ f^{-1}(B) \}, \quad B \subseteq \FF. \]
That is, the marginal IM for $\Phi$ is obtained by suitably mapping the original IM for $\Theta$ via the function $f$ almost exactly as in the familiar probability calculus: the key point is 
\[ \phi \in B \iff f(\theta) \in B \iff \theta \in f^{-1}(B). \]
So then the (imprecise) probability of the left-most assertion should equal the (imprecise) probability of the right-most assertion.  The only difference here compared to the more familiar probability calculus is that the appropriate operation is optimization of the possibility contour rather than integration of the probability density.  

The possibility calculus is particularly well-suited for preserving the original IM's statistical properties through the marginalization process.  Next is a result that demonstrates the above-defined marginal IM for $\Phi=f(\Theta)$ remains (strongly) valid.  Consequently, the marginal set estimator 
\[ C_\alpha^f(Y) = \{ \phi: \pi_Y^f(\phi) > \alpha \}, \quad \alpha \in [0,1], \]
is a nominal $100(1-\alpha)$\% confidence set in the sense described above.  

\begin{cor}
The marginal IM for $\Phi=f(\Theta)$ derived from the IM for $\Theta$ via the extension principle is (strongly) valid in the sense that 
\[ \uprob_{Y,\Theta}\bigl\{ \pi_Y^f\bigl( f(\Theta) \bigr) \leq \alpha \bigr\} \leq \alpha, \quad \alpha \in [0,1]. \]
\end{cor}

\begin{proof}
Follows immediately from Theorem~\ref{thm:valid} and the fact that $\pi_Y^f(f(\Theta)) \leq \pi_Y(\Theta)$, which follows from the definition of $\pi_y^f$ as a supremum in \eqref{eq:naive.contour}. 
\end{proof}

It may help to consider the special case of vacuous prior information about $\Theta$, where the above result can be compared to more familiar results in the (non-Bayesian) statistics literature.  Suppose, I'm in possession of a $100(1-\alpha)$\% confidence set for $\Theta$, say $C_\alpha(Y)$, and I want a corresponding confidence set for $\Phi=f(\Theta)$---how should I proceed?  Naturally, I'd just map $C_\alpha(Y)$ to a subset of $\FF$ via the mapping $f$, i.e., 
\[ C_\alpha^f(Y) = f\{ C_\alpha(Y) \} = \{ \phi: \text{$\phi = f(\theta)$ for some $\theta \in C_\alpha(Y)$}\}. \]
In light of the duality between confidence sets and tests of significance, there exists a p-value function, say, $\theta \mapsto \varpi_y(\theta)$ such that 
\[ C_\alpha(Y) = \{\theta: \varpi_Y(\theta) > \alpha\} \]
and, therefore, 
\begin{align*}
C_\alpha^f(Y) & = \{\phi: \text{$\phi=f(\theta)$ for some $\theta$ with $\varpi_Y(\theta) > \alpha$}\} \\
& = \Bigl\{ \phi: \sup_{\theta: f(\theta)=\phi} \varpi_Y(\theta) > \alpha \Bigr\}. 
\end{align*}
So, the notion of marginalization via optimization is completely natural and done by statisticians without second thought.  The reason is that it perfectly aligns with the preservation of desirable statistical properties, such as coverage probability guarantees.  

If a valid marginal IM for $\Phi = f(\Theta)$ is available for any $f$, then what's left to do?  Isn't the problem of eliminating nuisance parameters settled?  The only concern with the above (naive) IM solution is that it generally will fall short in the sense of efficiency compared to what's possible for any specific $f$.  To better understand this notion of efficiency, let's consider a simple example involving iid data from a normal model where $\Theta$ denotes the uncertain mean and standard deviation.  If $n=10$ and the observed sample mean and standard deviation are 0 and 1, respectively, then the IM's joint contour function for the pair, based on a vacuous prior is displayed in Panel~(a) of Figure~\ref{fig:normal.marg.compare}; this is the same plot shown in Figure~11 of Part~II.  Panel~(b) displays the naive marginal IM's contour function (red) for $\Phi = \text{mean}$ based on the extension principle and it looks as one would expect.  It also displays the more efficient marginal IM's contour (black) as described below.  Note that the two curves are symmetric around the same point, the sample mean, but the latter vanishes much more rapidly.  It's this faster decay that makes the latter marginal IM solution more efficient than the former.  The primary reason for this lack of efficiency is as follows: the solution based on the extension principle must produce a valid marginal IM for any choice of feature mapping $f$, so it necessarily can't be tailored toward efficient marginal inference concerning any specific feature.  

\begin{figure}[t]
\begin{center}
\subfigure[Joint contour for $\Theta=(\text{mean, sd})$]{\scalebox{0.6}{\includegraphics{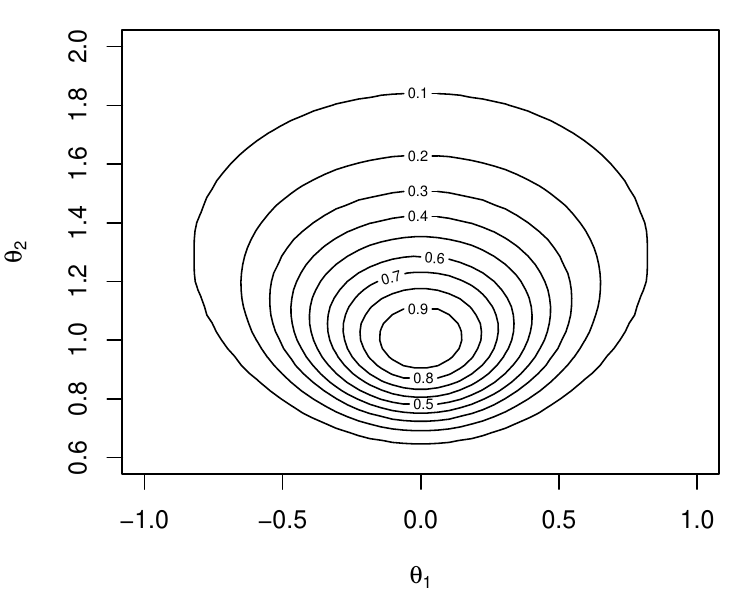}}}
\subfigure[Marginal contours for $\Phi=\text{mean}$]{\scalebox{0.6}{\includegraphics{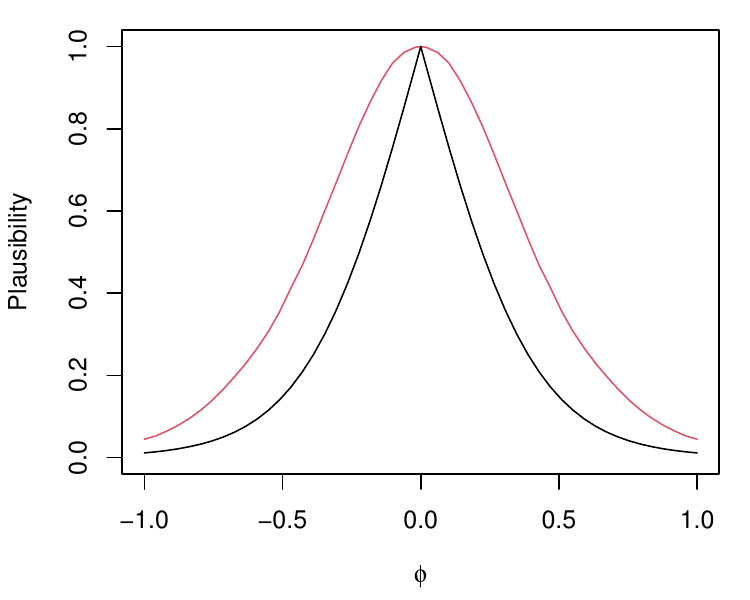}}}
\end{center}
\caption{Panel~(a) shows the joint contour function for $\Theta$, the mean and standard deviation of the normal model; same as Figure~11 in Part~II.  Panel~(b) shows the marginal contour for $\Phi=\text{mean}$ based on the naive extension principle (red) and the more efficient marginalization strategy (black) described in Section~\ref{SS:more}.}
\label{fig:normal.marg.compare}
\end{figure}

To me, the benefit of the naive marginal IM solution is that it's simple, i.e., it doesn't require any model-specific considerations, just solving an optimization problem.  This can be useful in exploratory situations where the inferential target isn't or hasn't yet been determined, or perhaps if the original IM solution has already been carried out and all I have access to is the resulting contour function.  But in light of its inefficiencies, it's necessary to push for more efficient marginal IMs.

\subsection{More efficient solutions}
\label{SS:more}

Before jumping into some new details, it may help to revisit marginalization in the Bayesian context.  While it's true that Bayesian marginalization is at least conceptually straightforward---application of the ordinary probability calculus---this doesn't necessarily ``work'' from a statistical point of view.  The situations that I have in mind here are those where little or no prior information is available about $\Theta$, so the Bayesian is required to proceed with the choice of a default prior distribution.  In such cases, what justifies the default-prior Bayes solution is that the derived procedures, e.g., credible sets, have good statistical properties.  It's known, however, that a suitable default prior for $\Theta$ might lead to a marginal posterior for $\Phi=f(\Theta)$ that has poor statistical properties; a classical example of this is \citet{stein1959}, but see, also, \citet{fraser2011}.  To prevent such cases, the Bayesian is forced to choose a default prior for $\Theta$ that's designed for the specific choice of $f$ so that the corresponding marginal posterior for $\Phi=f(\Theta)$ has the aforementioned desired statistical properties; see, for example, \citet{jeffreys1946}, \citet{tibshirani1989}, \citet{berger.bernardo.1992a}, \citet{berger.liseo.wolpert.1999}, \citet{bergerbernardosun2009}, \citet{liu.etal.reference.2014}, and others.  The point is that {\em reliable} marginalization often doesn't follow simply from the framework's marginalization calculus; instead, it requires careful consideration of the specific inferential target.  What I present below is in this vein, but far less nebulous and convoluted than defining and constructing statistically-suitable default priors. 

As discussed in Part~II, efficiency gains are a consequence of reducing the complexity of the Choquet integral---the upper probability with respect to ``$\uprob_{Y,\Theta}$''---that defines the IM's contour.  My {\em Principle of Minimum Complexity} says that, for the sake of efficiency, drop the complexity as much as possible, and this is typically achieved by reducing the dimension of the variables being integrated over.  Towards achieving this dimension/complexity reduction, one can leverage special structure in the model's likelihood function, in particular, the various forms of factorization discussed in Section~\ref{SS:factor} above.  

To start, I'll ignore any available partial prior information and just focus on the model/data; the partial prior will come back into the picture shortly.  The profile relative likelihood \citep[e.g.,][]{kalb.sprott.1970, murphy.vaart.2000, maclaren.profile} offers a natural data-dependent plausibility order exclusively on the interest parameter space $\FF$; that is, if $R(y,\phi) > R(y,\phi')$, where 
\[ R(y,\phi) = \frac{\sup_{\lambda \in \LL} p_{Y|\phi,\lambda}(y)}{\sup_{\varphi \in \FF, \lambda \in \LL} p_{Y|\varphi,\lambda}(y)}, \quad \phi \in \FF, \]
then the value $\phi$ is understood as being ``more compatible'' with data $y$ than the value $\phi'$.  The use of the relative profile likelihood also aligns with the principles in Part~II that justified the likelihood-based IM construction.  The key point is that, by removing the direct dependence on $\lambda$ in the plausibility ordering, an opportunity is created for the $\Lambda$ dimension in the Choquet integral calculation that defines the IM's contour to collapse, leading to improved efficiency.  Removing direct dependence on the nuisance parameters to create an opportunity for dimension reduction is a recurring theme in this paper. 

This begs the question: how can this opportunity for dimension reduction and/or efficiency gain be realized?  When the likelihood function factors as in Section~\ref{SS:factor}, the effective dimension drops because terms in the relative likelihood cancel out.  First, in the ``Ideal factorization'' case, suppose the data decomposes as $y \mapsto \{U(y), V(y)\}$ and $U$ is P-sufficient for $\Phi$.  Then the relative profile likelihood easily simplifies to 
\[ R(y,\phi) = \frac{p_{U|\phi}(u)}{\sup_{\varphi \in \FF} p_{U|\varphi}(u)}, \quad \phi \in \FF, \quad u = U(y), \]
and note that the right-hand side depends on data $y$ only through the value $u$ of the statistic $U(y)$; I'll abuse notation and write this as $R(u,\phi)$.  Since $U(y)$ is lower-dimensional than $y$ itself, we have effectively achieved a dimension reduction.  Moreover, if there's no interest in $\Lambda$, then there's no utility in fleshing out partial prior information for $\Lambda$---it's safe to focus on {\em a priori} uncertainty quantification about $\Phi$.  The necessary function is $q_\Phi(\phi) = \uprior(\{\phi\} \times \LL)$, and this can be combined directly with the (simplified) relative profile likelihood to get 
\[ R_q(u, \phi) = \frac{R(u,\phi) \, q_\Phi(\phi)}{\sup_{\varphi \in \FF} \{ R(u,\varphi) \, q_\Phi(\varphi)\}}, \quad \phi \in \FF, \quad u = U(y). \]
The critical observation is that, in the calculation of the IM's contour function, the (Choquet) integration over the $V$ and $\Lambda$ dimensions disappears:
\begin{align*}
\pi_y(\phi) & = \prob_{Y,\Theta}\{ R_q(U(Y), \Phi(\Theta)) \leq R_q(U(y), \phi) \} \\
& = \uprob_{U,\Phi}\{ R_q(U, \Phi) \leq R(U(y),\phi) \}, \quad \phi \in \FF. 
\end{align*}
Note that the far left-hand side, $\pi_y(\phi)$, depends on $y$ only through the value $u$ of $U(y)$, so, with my usual abuse of notation, I'll denote this instead $\pi_{U(y)}(\phi)$.  The collapsing in the dimension happens because the relative profile likelihood only depends on the random variable $U(Y)$, and its distribution depends only on the uncertain $\Phi$, not on $\Lambda$.  More formally, since $Y=\{U(Y),V(Y)\}$ and $\Theta=(\Phi,\Lambda)$, 
\begin{align*}
\pi_{U(y)}(\phi) & = \uprob_{Y,\Theta}\{ R_q(U(Y), \Phi(\Theta)) \leq R_q(U(y), \phi) \} \\
& = \uprob_{U,V,\Phi,\Lambda}\{ R_q(U(Y), \Phi) \leq R_q(U(y), \phi) \} \\
& = \sup_{\prior_{\Phi,\Lambda}} \int \prob_{U,V|\varphi,\lambda}\{ R_q(U,\varphi) \leq R_q(u,\phi) \} \, \prior_{\Phi,\Lambda}(d\varphi, d\lambda) \\
& = \sup_{\prior_\Phi} \int \prob_{U|\varphi}\{ R_q(U,\varphi) \leq R_q(u,\phi) \} \, \prior_\Phi(d\varphi) \\
& = \uprob_{U,\Phi}\{ R_q(U,\Phi) \leq R_q(u,\phi) \}.
\end{align*}
The suprema above are over all the joint and marginal priors in the respective credal sets.  I'll discuss below how, at least in some cases, further dimension reduction is possible.  A couple more points deserve note.  First, (strong) validity still holds---the reason is that the same collapsing ``$\uprob_{Y,\Theta} \searrow \uprob_{U,\Phi}$'' occurs when (upper) probabilities concerning $(Y,\Theta) \mapsto \pi_{U(Y)}(\Phi)$, so it's as if the problem originated with the marginal model for $U(Y)$ depending only on the uncertain $\Phi$ and I constructed the IM solution from there.  Second, the above reduction happens automatically without any direct intervention from the data analyst.  That is, even if one doesn't recognize that the $(V,\Lambda)$ dimensions can be collapsed, they get collapsed anyway and, consequently, the results obtained end up the same either way.  So, the efficiency gain that occurs from recognizing this particular opportunity for dimension reduction is computational, not statistical.  

Sticking with the ``Ideal factorization'' case, with a decomposition $y \mapsto \{U(y),V(y)\}$, now suppose that $U$ is S-ancillary for $\Phi$.  Then the relative profile likelihood reduces to 
\[ R(y,\phi) = \frac{p_{V|u,\phi}(v)}{\sup_{\varphi \in \FF} p_{V|u,\varphi}(v)}, \quad \phi \in \FF, \quad (u,v) = \{U(y),V(y)\}, \]
and note that the right-hand side (basically) only depends on the data $y$ through the value $v$ of the statistic $V(y)$.  I say ``basically'' because, since the dependence is through a conditional likelihood, there's an opportunity to take the value $u$ of $U(y)$ as {\em fixed}, thereby making the effective dimension that of $v$.  Towards this, let me write the left-hand side above as $R(v,\phi \mid u)$.  With the partial prior for $\Phi$ described by $q_\Phi$, define 
\[ R_q(v, \phi \mid u) := \frac{R(v,\phi \mid u) \, q_\Phi(\phi)}{\sup_{\varphi \in \FF} \{ R(v,\varphi \mid u) \, q_\Phi(\varphi)\}}, \quad \phi \in \FF. \]
The appearance of the word ``ancillary'' and the notation I've introduced suggests a strategy wherein the observed value $u$ of $U(y)$ is conditioned on: 
\begin{align}
\pi_y(\phi) & = \pi_{v|u}(\phi) \notag \\
& = \uprob_{V,\Phi|u}\{ R_q(V, \Phi \mid u) \leq R_q(v, \phi \mid u) \} \notag \\
& = \sup_{\prior_\Phi} \int \prob_{V|u,\varphi}\{ R_q(V, \varphi \mid u) \leq R_q(v, \phi \mid u) \} \, \prior_\Phi(d\varphi), \quad \phi \in \FF. \label{eq:contour.conditional} 
\end{align}
Similar to the P-sufficient case above, there is a reduction in dimension because the (Choquet) integration over the $U$ and $\Lambda$ spaces collapses.  But note that here in the S-ancillary case, the user would have to intervene---to carry out the ``condition on $U=u$'' step manually---it doesn't happen automatically like in the P-sufficient case above.  There are benefits to be enjoyed as a result of this careful intervention, however.  First, just like in Section~6.1 of Part~II, this conditioning preserved validity.  Second, there is a computational efficiency gain resulting from the dimension reduction.  Finally, since the observed value of $U(Y)$ often can be interpreted as a sort of ``informativeness index,'' by conditioning on this value, the inference is adaptive in the sense that improved efficiency is achieved if it's warranted by the data in hand.  

Beyond the ideal factorization confines, it's far less obvious how to proceed.  In the less-than-ideal factorization cases described above, and even more generally, the profiling strategy can still be carried out.  Consider first the $\Phi$-oriented case, where the relative profile likelihood function is 
\[ R_q(y,\phi) = \frac{p_{U|\phi}(u) \, \sup_{\lambda \in \LL} \{ p_{V|u,\phi,\lambda}(v) \, q_{\Phi,\Lambda}(\phi,\lambda)\}}{\sup_{\varphi \in \FF} \sup_{\lambda \in \LL} \{ p_{U|\varphi}(u) \, p_{V|u,\varphi,\lambda}(v) \, q_{\Phi,\Lambda}(\varphi,\lambda)\}}, \quad \phi \in \FF, \]
where $(u,v)$ is the observed value of $\{U(y), V(y)\}$.  Unfortunately, none of the terms cancel, but what really matters is that the right-hand side doesn't depend directly on the nuisance parameter value.  So the profiling strategy is still viable, but it's clearly not the only option.  In fact, one might have good reason (e.g., Section~\ref{SS:caution} below) to simply ignore the term that involves both $(\phi,\lambda)$, and work with a different kind of plausibility order, based on the relative marginal likelihood:
\[ R_q(y,\phi) = \frac{p_{U|\phi}(u) \, q_\Phi(\phi)}{\sup_{\varphi \in \FF} p_{U|\varphi}(u) \, q_\Phi(\varphi)}, \quad \phi \in \FF. \]
This boils down to ignoring some relevant information about $\Phi$, but in exchange for added simplicity and perhaps greater efficiency.  Indeed, working just with the marginal distribution of the statistic $U$ completely eliminates $V$ and $\Lambda$, significantly simplifying the IM construction.  The same points apply to the specific-sufficient case under the less-than-ideal factorization umbrella.  There, the relative profile likelihood is
\[ R_q(y,\phi) = \frac{p_{V|u,\phi}(v) \, \sup_{\lambda \in \LL} \{ p_{U|\phi,\lambda}(u) \, q_{\Phi,\Lambda}(\phi,\lambda)\}}{\sup_{\varphi \in \FF} \sup_{\lambda \in \LL} \{ p_{V|u,\varphi}(v) \, p_{U|\varphi,\lambda}(u) \, q_{\Phi,\Lambda}(\varphi,\lambda)\}}, \quad \phi \in \FF, \]
where $(u,v)$ is the observed value of $\{U(y), V(y)\}$.  Again, there are no simplifications that can be made here, but the profiling strategy can still be applied.  But the user might---for good reason---opt to ignore the information about $\Phi$ in the marginal distribution of $U$ and use simply the relative conditional likelihood 
\[ R_q(y,\phi) = \frac{p_{V|u,\phi}(v) \, q_\Phi(\phi)}{\sup_{\varphi \in \FF} p_{V|u,\varphi}(u) \, q_\Phi(\varphi)}, \quad \phi \in \FF. \]
As above, such a decision offers substantial simplification compared to profiling, along with potential efficiency gains in certain cases.  

More generally, there will be no obvious factorizations that can be made to the likelihood, neither ideal nor less-than-ideal.  One can, of course, still use the profiling strategy described above, and strong validity still holds.  In my experience, which I'll share in the examples below, profiling is quite reliable and often leads to an optimal/most efficient solution; but there are cases in which better efficiency can be achieved by working with the relative marginal likelihood instead.  So, while my go-to nuisance parameter elimination strategy is profiling, it's important to keep in mind that efficient marginal inference is an extremely challenging problem and one can't expect a single strategy to work best uniformly over all problems.  Therefore, I'm open to finding creative alternative---perhaps not likelihood-based---plausibility orderings on a case-by-case basis; such creativity is probably necessary for efficient inference in non-parametric problems (Section~\ref{S:beyond}).

\subsection{Further efficiency gains}

The above discussion focused on how to collapse the dimensions in the Choquet integration directly related to the nuisance parameter---both in the $\Lambda$ dimension and in that of the statistic carrying information relevant to $\Lambda$.  Depending on the form of the partial prior information, there may be extra opportunities to reduce the dimension further.  This, again, relies on my {\em Principle of Minimum Complexity}, and I discuss this idea at length in Part~II; so I'll only give a brief explanation here. 

In the fully vacuous prior case, validity is equivalent to 
\[ \sup_{\lambda \in \LL} \prob_{Y|\phi,\lambda}\{ \pi_Y(\phi) \leq \alpha \} \leq \alpha, \quad \phi \in \FF, \quad \alpha \in [0,1]. \]
So it's enough to take the contour function as 
\begin{equation}
\label{eq:contour.vacuous}
\pi_y(\phi) = \sup_{\lambda \in \LL} \prob_{Y|\phi,\lambda}\{ R(Y,\phi) \leq R(y,\phi)\}, \quad \phi \in \FF. 
\end{equation}
The point is that the value of $\phi$ is taken as fixed, so there's no (Choquet) integration over $\phi$ needed.  In the ideal factorization cases above, e.g., the P-sufficient case, the plausibility ordering was expressed exclusively in terms of the statistic $U$ whose sampling distribution doesn't depend on the nuisance parameter.  In such a case, the $\Lambda$ dimension collapses too, so the above display reduces to 
\[ \pi_y(\phi) = \prob_{U|\phi}\{ R(U,\phi) \leq R(U(y), \phi) \}, \quad \phi \in \FF, \]
i.e., all that's required is (ordinary) integration over the $U$ space.  Similar extra reduction is possible in the S-ancillary case as well.  

In the examples follow, I'll focus primarily on the vacuous prior case and, therefore, I'll use the contour function in \eqref{eq:contour.vacuous} or, whenever possible, its no-$\lambda$ version.  The opposite extreme case, when there's a complete (precise) prior distribution for $(\Phi,\Lambda)$, there is a similarly extreme dimension reduction that can be achieved.  Roughly, one can collapse the integration over the entire $Y$-space, and the IM contour is obtained by (ordinary) integration with respect to the usual conditional distribution of $(\Phi,\Lambda)$, given $Y=y$.  Applications where complete prior information is available are quite rare, so these details are really only of theoretical interest.  For the general partial prior case, which is far more common in application, I'm not aware of a one-size-fits-all dimension reduction strategy and, if such a strategy exists and, if so, what does it look like are important open questions.  In the examples below where I incorporate partial prior information, I use the basic Choquet integral formula with no extra dimension reduction tricks.

\subsection{First examples}
\label{SS:examples.nuisance}

\begin{ex}[Multinomial]
As a first and relatively simple example, consider a random sample of size $n$ from a population having three distinct categories; let $Y = (Y_1, Y_2, Y_3)$ denote the corresponding vector of counts, with $Y_k$ the number of category-$k$ observations, for $k=1,2,3$, so that $Y_1 + Y_2 + Y_3 = n$.  Let $\prob_{Y|\theta}$ denote a multinomial model where $\theta=(\theta_1, \theta_2, \theta_3)$ is such that $\theta_k \geq 0$ is the category-$k$ probability and $\theta_1 + \theta_2 + \theta_3 = 1$.  Write $\Theta$ for the uncertain value, with components $\Theta_k$, $k=1,2,3$. 

Suppose interest is in $\Phi = \Theta_1$.  Thanks to the probability vector constraints, $\Theta$ is equivalent to $(\Phi, \Lambda)$ where $\Lambda = \Theta_2 / (1 - \Theta_1)$ is the conditional probability of category 2 given categories 2 or 3.  Similarly, $Y$ is equivalent to the pair $(U,V)$, where $U=Y_1$ and $V=Y_2$.  Then it's easy to check that $U$ is P-sufficient for $\Phi$, so this example falls under the ``ideal factorization'' umbrella and the marginal IM construction is simple and leads to efficient inference on $\Phi$.  Indeed, the marginal distribution of $U$, given $\Phi=\phi$, is binomial with parameters $n$ and $\phi$, and so the marginal IM for $\Phi$ is, after this reduction, exactly like that given in Example~1 of Part~II.  
\end{ex}

\begin{ex}[Two binomial counts]
\label{ex:log.odds}
Let $Y=(Y_1,Y_2)$ denote two independent binomial counts, with $(Y_1 \mid \Theta_1=\theta_1) \sim \bin(n_1, \theta_1)$ and $(Y_2 \mid \Theta_2=\theta_2) \sim \bin(n_2,\theta_2)$, where $n=(n_1,n_2)$ is known but $\Theta=(\Theta_1,\Theta_2)$ is unknown.  This setup is common in, e.g., clinical trails, where $Y_1$ and $Y_2$ correspond to the number of events observed in the control and treatment groups, respectively.  A relevant feature of $\Theta$ is the log odds ratio
\[ \Phi = \log\Bigl( \frac{\Theta_2}{1-\Theta_2} \div \frac{\Theta_1}{1-\Theta_1} \Bigr). \]
It's well known that the conditional distribution of $V=Y_2$, given $U=Y_1 + Y_2$ and $\Phi=\phi$, has a mass function 
\[ p_{V|u,\phi}(v) \propto \binom{n_2}{v} \binom{n_1}{u - v} \, e^{\theta v}, \quad v=\max(u-n_1, 0),\ldots,\min(n_2,u), \]
where $\phi$ is the log-odds ratio corresponding to $\theta$.  The key observation is that this conditional distribution only depends on $\phi$, not on any other features of $\theta$.  So this example is one where $U$ is S-ancillary and so the marginal IM solution proceeds by conditioning on the observed value $u=$ of $U=Y_1+Y_2$.  That is, the marginal IM's contour function is $\pi_y(\phi) = \pi_{v|u}(\phi)$ as in \eqref{eq:contour.conditional}.   

For illustration, I consider two mortality data sets presented in Table~1 of \citet{normand1999}, namely, Trials~1 and 6.  Two plausibility contours for $\Phi$ are shown in each of the two panels in Figure~\ref{fig:or}: one for a vacuous prior and one for a partial prior.  The partial prior for $\Phi$ I'm considering here is the so-called {\em Markov prior}, like in Example~2 of Part~II, that represents the entire collection precise prior distributions with $\E|\Phi| \leq 1$; the contour is given by 
\[ q_\Phi(\phi) = 1 \wedge |\phi|^{-1}, \quad \phi \in \RR. \]
The estimated log-odds ratios are 0.83 and 0.99 for Trial~1 and 6, respectively, but both cases are challenging thanks to the number of events being relatively small.  Trial~6 is an overall larger study, so the plausibility contour is much more tightly concentrated than that for Trial~1.  For Trial~1, the data is quite compatible with the partial prior information so, as expected, there's a non-trivial gain in efficiency when comparing the partial-prior IM to the vacuous-prior IM.  For Trial~6, on the other hand, the data and prior are still compatible, but now the data is much more informative, so the difference in efficiency isn't as wide.  The same data was analyzed in \citet{hannig.xie.2012} and \citet{gim} but the results look quite different compared to those presented here in Figure~\ref{fig:or}.  This is because the former reference makes certain adjustments to effectively eliminate the discreteness of the data, while the latter does something very similar to the vacuous-prior IM solution here, just is less efficient due to its failure to make full use of the relative conditional likelihood. 
\end{ex}

\begin{figure}[t]
\begin{center}
\subfigure[Trial~1: $Y=(1,2)$, $n=(43, 39)$]{\scalebox{0.55}{\includegraphics{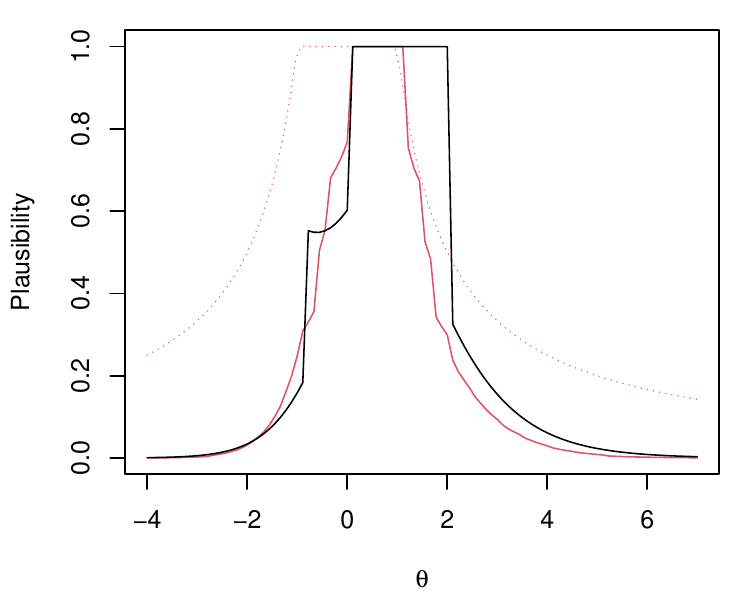}}}
\subfigure[Trial~6: $Y=(4,11)$, $n=(146, 154)$]{\scalebox{0.55}{\includegraphics{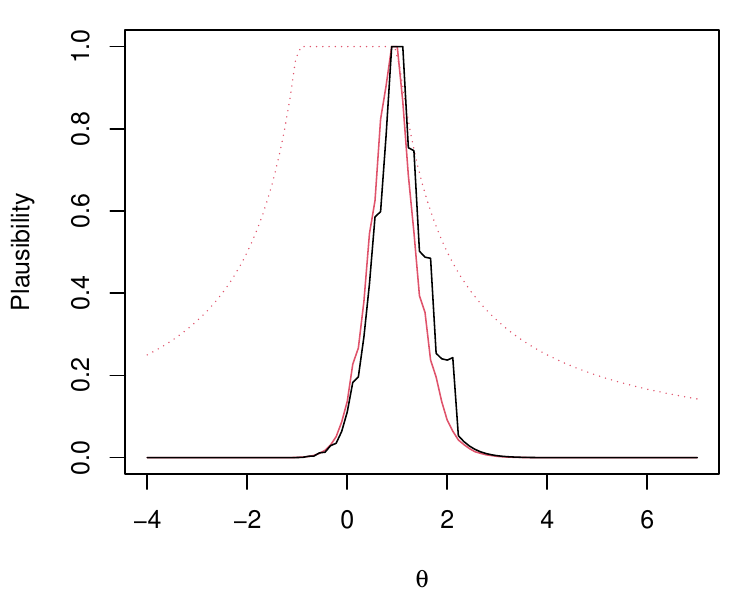}}}
\end{center}
\caption{Plausibility contours for the log odds ratio in two mortality data sets (Trial~1 and Trial~6) presented in Table~1 of \citet{normand1999}: vacuous prior (black), partial prior (red), and prior contour (red dotted).}
\label{fig:or}
\end{figure}

The reader might be surprised to learn that the often-simple normal model, considered next, is a less-than-ideal factorization case.  Fortunately, despite being less-than-ideal, there are still some nice features that make this example rather straightforward.  

\begin{ex}[Normal]
\label{ex:normal.1}
Let $Y=(Y_1,\ldots,Y_n)$ be an iid sample from a normal distribution with uncertain mean $\Theta_1$ and standard deviation $\Theta_2$.  For brevity in this illustration, I'll suppose the prior information about $\Theta=(\Theta_1,\Theta_2)$ is vacuous.  To start, suppose the mean $\Phi=\Theta_1$ is the interest parameter and $\Lambda=\Theta_2$ is the nuisance parameter.  As hinted at above, this example doesn't fall even under the ``less-than-ideal factorization'' umbrella.  In particular, there is no $\Phi$-oriented statistic---e.g., the distribution of the sample mean depends on both $\Theta_1$ and $\Theta_2$---and, therefore, no P-sufficient statistic.  But there's nothing stopping me from proceeding with the profiling strategy described above.  The relative profile likelihood is easy to get in this case, and it's given by 
\[ R(y,\phi) = \Bigl( \frac{\hat\lambda_y^2}{\hat\lambda_y^2(\phi)} \Bigr)^{n/2}, \]
where $\hat\lambda_y^2 = n^{-1} \sum_{i=1}^n (y_i - \bar y)^2$ is the maximum likelihood estimator of $\Lambda^2$ and $\hat\lambda_y^2(\phi) = n^{-1} \sum_{i=1}^n (y_i - \phi)^2$.  Some simple algebra reveals
\begin{align*}
R(Y,\phi) \leq R(y,\phi) & \iff \frac{\hat\lambda_Y^2}{\hat\lambda_Y^2(\phi)} \leq \frac{\hat\lambda_y^2}{\hat\lambda_y^2(\phi)} \\
& \iff \frac{\hat\lambda_Y^2(\phi)}{\hat\lambda_Y^2} \geq \frac{\hat\lambda_y^2(\phi)}{\hat\lambda_y^2} \\
& \iff \frac{n(\bar Y - \phi)^2}{\hat\lambda_Y^2} \geq \frac{n(\bar y - \phi)^2}{\hat\lambda_y^2}. 
\end{align*}
Under $\prob_{Y|\phi,\lambda}$, the random variable on the left-hand side is a pivot---distribution independent of both $\phi$ and $\lambda$---and, in particular, is distributed as ${\sf F}(1,n-1)$ or, equivalently, as the square of $T \sim {\sf t}(n-1)$.  Then the marginal IM contour is given by 
\begin{align*}
\pi_y(\phi) & = \prob\Bigl\{ T^2 \geq \frac{n(\bar y - \phi)^2}{\hat\lambda_y^2} \Bigr\} \\
& = 1 - \Bigl| 2 F_{n-1}\Bigl( \frac{n(\bar y - \phi)^2}{\hat\lambda_y^2} \Bigr) - 1 \Bigr|, \quad \phi \in \RR. 
\end{align*}
A plot of this curve was shown in Figure~\ref{fig:normal.marg.compare} to demonstrate the efficiency gains that are possible when the IM construction is tailored to the interest parameter.  This function is just the usual p-value for the Student-t test, and the marginal IM plausibility intervals for $\Phi$ are exactly the familiar textbook Student-t confidence intervals.  

Next, suppose that the standard deviation $\Phi = \Theta_2$ is of interest and the mean $\Lambda=\Theta_1$ is the nuisance parameter.  In this case, there exists $\Phi$-oriented statistics---including the maximum likelihood estimator---but \citet[][p.~280]{basu1977} argues that none is ``maximally'' $\Phi$-oriented.  This suggests two possible strategies to construct a marginal IM: one is to work with the marginal likelihood based on the maximum likelihood estimator, which is tied directly to a chi-square distribution, and another is to stick with the general profiling strategy and work out the details.  It turns out, however, that the two solutions are almost the same in this case.  Write $\hat\phi_y^2 = n^{-1} \sum_{i=1}^n (y_i - \bar y)^2$ for the maximum likelihood estimator of the variance $\Phi^2$.  Then the relative profile likelihood is easily shown to be 
\[ R^\text{pr}(y,\phi) \propto \Bigl( \frac{n\hat\phi_y^2}{\phi^2} \Bigr)^{n/2} \exp\Bigl(-\frac{n}{2} \frac{\hat\phi_y^2}{\phi^2} \Bigr), \quad \phi > 0. \]
Clearly, $R(Y,\phi)$ depends on $(Y,\phi)$ only through $n\hat\phi_Y^2/\phi^2$, which is a pivot (chi-square) when $Y$ is normal with variance $\phi$, so computation of the marginal IM contour is straightforward.  Similarly, the relative marginal likelihood is easily shown to be
\[ R^\text{ma}(y,\phi) \propto \Bigl( \frac{n\hat\phi_y^2}{\phi^2} \Bigr)^{(n-1)/2} \exp\Bigl(-\frac{n}{2} \frac{\hat\phi_y^2}{\phi^2} \Bigr), \quad \phi > 0. \]
The only difference between the two relative likelihoods is the power on the polynomial term, which is a negligible difference even for moderate $n$.  In terms of the corresponding marginal IM contours, this difference mainly only affects the location of its peak: for the profile likelihood the peak is at the maximum likelihood estimator and, for the marginal likelihood, the peak is at the sample variance.  Figure~\ref{fig:normal.marg.var}(a) shows the same joint contour for the normal mean and standard deviation---based on $n=10$---presented in Figure~\ref{fig:normal.marg.compare}(a) above; then Panel~(b) shows three different marginal IM contours for the standard deviation $\Phi$: one based on the naive marginalization via the extension principle and the two relative likelihood-based strategies above.  Note the efficiency gain in the two direct marginal IM constructions compared to that based on naive marginalization post-construction.  My preferred solution is that based on the profiling strategy, since the maximum likelihood estimator ought to be the most plausible.  

Of course, the above analysis carries over to many of the commonly used normal models, like in linear regression and analysis of variance, and the suitably constructed marginal IMs will reproduce the standard textbook results.  For example, the marginal IM for a subset of the regression coefficients in an ordinary linear model would correspond to the p-value function based on Hotelling's $T^2$ statistic and the F-distribution.  
\end{ex}

\begin{figure}[t]
\begin{center}
\subfigure[Joint contour for $\Theta=(\text{mean, sd})$]{\scalebox{0.6}{\includegraphics{twopar_normal}}}
\subfigure[Marginal contours for $\Phi=\text{sd}$]{\scalebox{0.6}{\includegraphics{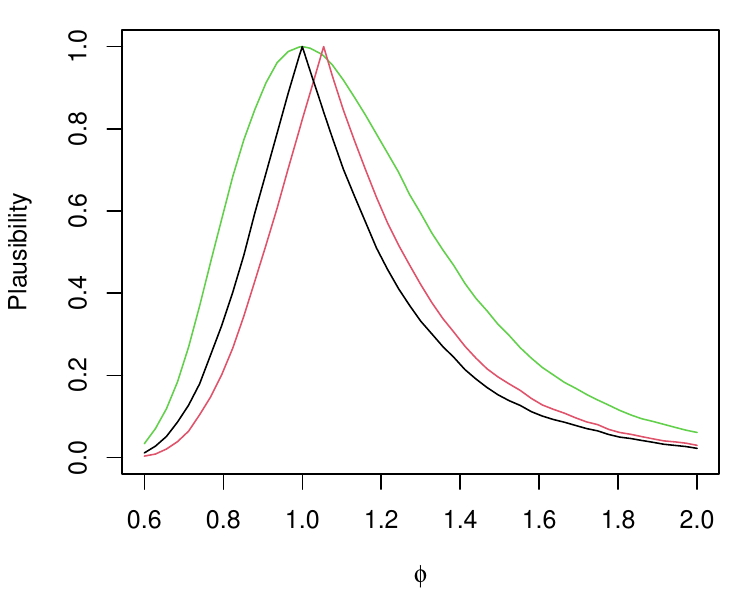}}}
\end{center}
\caption{Panel~(a) shows the joint contour function for $\Theta$, the mean and standard deviation of the normal model; same as in Figure~\ref{fig:normal.marg.compare}(a).  Panel~(b) shows the marginal contour for $\Phi=\text{sd}$ based on the naive extension principle (green) and the more efficient marginalization profile (black) and marginal (red) strategies.}
\label{fig:normal.marg.var}
\end{figure}

\begin{ex}[Fieller--Creasy]
\label{ex:fc}
The so-called Fieller--Creasy problem \citep{fieller1954, creasy1954}, in its simplest form, starts with an independent pair of observables, $(Y_i \mid \Theta_i=\theta_i) \sim \nm(\theta_i, 1)$, for $i=1,2$.  It's not important that the variances are equal (to 1), e.g., the $Y_i$'s could be averages based on different sample sizes, say; all that's effectively being assumed here is that the variances are known. This is a seemingly trivial problem: inference about an unknown normal mean vector with known variances is one of the few statistical problems that's actually {\em solved}.  What's unique about this example, and what makes it surprisingly challenging, is that interest is in the ratio $\Phi = \Theta_1 / \Theta_2$.  This example's fame---or infamy---originated from the two distinct fiducial solutions (one by Fieller and one by Creasy) which both appeared justified based on Fisher's reasoning; in the end, Fisher sided with Fieller's approach which, by the way, produces confidence intervals that attain the exact nominal coverage.  This example is also one of the simplest among those in the class of problems for which there exists no set estimator that has finite length (almost surely) and positive coverage probability \citep{gleser.hwang.1987}.  This implies that the usual strategies for constructing set estimators---such as ``estimate $\pm$ standard error'' or using quantiles of a marginal Bayesian or fiducial posterior distribution---that produce almost surely finite-length intervals, simply aren't going to work.  For these and perhaps other reasons, the late Sir D.~R.~Cox listed this as one of his ``challenge problems'' \citet{fraser.reid.lin.2018}.  More generally, even in relatively simple model, solutions can be quite challenging when the quantity of interest involves a ratio of the model parameters; see, e.g., the gamma mean problem in Example~\ref{ex:gamma.mean} below.

As this is just a marginal inference problem within a fairly standard model, it's worth to see how the proposed IM framework can handle this.  If we write $\Lambda = \Theta_2$, then we find that $\Theta_1 = \Phi \Lambda$ and we have a complete reparametrization.  The log-likelihood function, in this new parametrization, is quite simple:
\[ \ell_y(\phi, \lambda) = -\tfrac12(y_1 - \phi\lambda)^2 - \tfrac12 (y_2 - \lambda)^2. \]
This isn't an ``ideal factorization'' case, hence no P-sufficiency or S-ancillarity to guide us; we do find that $Y_2$ is $\Lambda$-oriented, but that doesn't help for inference on $\Phi$.  So I'll proceed here with the general recommendation to rely on the relative profile likelihood; for now, let's assume that the prior information available is vacuous, It's straightforward to identify the global maximum likelihood estimators of $\Phi$ and $\Lambda$, and the corresponding global maximum value of the likelihood function is a constant in both data and parameters, so can be ignored.  Concerning the profile likelihood, with $\phi$ fixed, it's not difficult to show that the maximum is attained at 
\[ \hat\lambda_y(\phi) = \frac{y_1 + \phi y_2}{1 + \phi^2}. \]
Then the log relative profile likelihood (ignoring additive constants) is 
\[ \log R(y,\phi) = -\frac12 \frac{(y_1 - \phi y_2)^2}{1 + \phi^2}, \quad \phi \in \RR. \]
The key observation is that $\log R(Y,\phi)$ is a pivot when $\Phi=\phi$, so the marginal IM contour for $\Phi$ is easy to get.  In fact, if $\Phi=\phi$, then $-2\log R(Y,\phi) \sim \chisq(1)$, so the contour has basically a closed-form expression:
\[ \pi_y(\phi) = 1 - G_1\{ -2\log R(y,\phi) \}, \quad \phi \in \RR, \]
where $G_d$ is the $\chisq(d)$ distribution function.  As in the general case, this determines a full marginal IM that can be used for reliable uncertainty quantification for $\Phi$.  It's also exactly the marginal IM solution produced in \citet{immarg} which, by the way, returns marginal IM plausibility regions that agree with Fieller's exact confidence intervals.  As this is special kind of problem with some unusual features, it's worth considering a quick illustration.  I'll focus here on a relatively ``weird'' case and, since this is easy to compute, leave experiments with other data sets to the reader.  Following \citet{schweder.hjort.2013}, suppose that the observed data is $y=(1.33, 0.33)$.  Since $y_2$ is relatively close to 0, we ought to be worried about instability affecting our inferences about the ratio $\Phi$.  Indeed, Figure~\ref{fig:fc} shows a plot of the marginal IM's contour function for $\Phi$ in this case, and the immediate observation is that, unlike in all the other examples in this paper, {\em the tails of the contour don't vanish as $\phi \to \pm \infty$}.  This implies that the marginal IM's plausibility region for $\Phi$ is unbounded, which is a necessary condition if it's to have non-zero coverage probability.  For other cases, where $y_2$ isn't too close to 0, the marginal IM contour looks similar to, e.g., the black line in Figure~\ref{fig:normal.marg.var}(b).  

That there's a class of examples for which the standard frequentist and Bayesian solutions fail is, to me at least, quite striking.  The reader might be wondering why the statistical community has swept this issue under the rug as opposed to confronting it head on.  The reason, I think, is that in this case at least it's easy to pinpoint the cause of the problematic behavior and explain it away.  Indeed, all the trouble identified by Gleser--Hwang is caused by the fact that there's nothing that prevents $\Lambda$ from being close to or equal to 0, and it's this ``divide by (near) 0'' that creates a singularity that leads to the non-existence of finite-length confidence sets.  So, all one has to do is say ``I don't think $\Lambda$ is close to 0'' and they've given themselves license to ignore the aforementioned issues.  My claim, however, is that if one genuinely doesn't believe $\Lambda$ is close to 0, then that (partial prior) information should be incorporated into the analysis from the beginning, both for the sake of transparency and for the opportunity to improve efficiency.  The problem, of course, is that the mainstream schools of statistical thought have no way to accommodate partial prior information of this sort.  Here, however, it's at least conceptually straightforward to incorporate partial prior information if available.  For illustration, suppose the partial prior can be described by a possibility contour 
\[ q_\Phi(\phi) = 1 \wedge 5|\phi|^{-1}, \quad \phi \in \RR. \]
This is the same style of prior I used in the log-odds ratio illustration of Example~\ref{ex:log.odds}, just now it encode the prior belief that ``$\E|\Phi| \leq 5$,'' or, more colloquially, ``I don't expect $\Phi$ to be too large.'' Note that this prior is vacuous about $\Lambda$, only mildly informative about $\Phi$.  Following the basic Choquet integral formula \eqref{eq:choquet.theta}, my partial prior marginal IM contour for $\Phi$ takes the form 
\[ \pi_y(\phi) = \int_0^1 \Bigl[ \sup_\lambda \sup_{\varphi: q_\Phi(\varphi) > s} \prob_{Y|\varphi, \lambda}\{ R_q(Y,\varphi) \leq R_q(y,\phi) \} \Bigr] \, ds, \]
where $R_q$ is the $q_\Phi$-regularized relative profile likelihood and the inside $\prob_{Y|\varphi, \lambda}$ probability is evaluated via Monte Carlo over a grid of $(\varphi, \lambda)$ values.  Intuitively, since the prior discounts large values of $\Phi$, we can expect that the final marginal IM contour has thinner---potentially vanishing---tails.  The red line in Figure~\ref{fig:fc} shows this contour and, indeed, the tails are vanishing and more efficient marginal inference obtains if prior information of the form ``I don't expect $\Phi$ to be large'' is incorporated into the analysis at the start.  Note, also, that increasing the efficiency in this way doesn't make the IM solution susceptible to the risk identified in Gleser--Hwang.  The reason is that the definition of {\em validity} also takes the partial prior information into account: the prior can be leveraged to improve efficiency because validity itself is relative to the prior.  Unlike Bayes, since no prior is required for the IM solution, I don't run a risk of making unjustified assumptions---I'm free to make those assumptions, and enjoy the efficiency gains, only when I can justify them.
\end{ex}

\begin{figure}[t]
\begin{center}
\scalebox{0.65}{\includegraphics{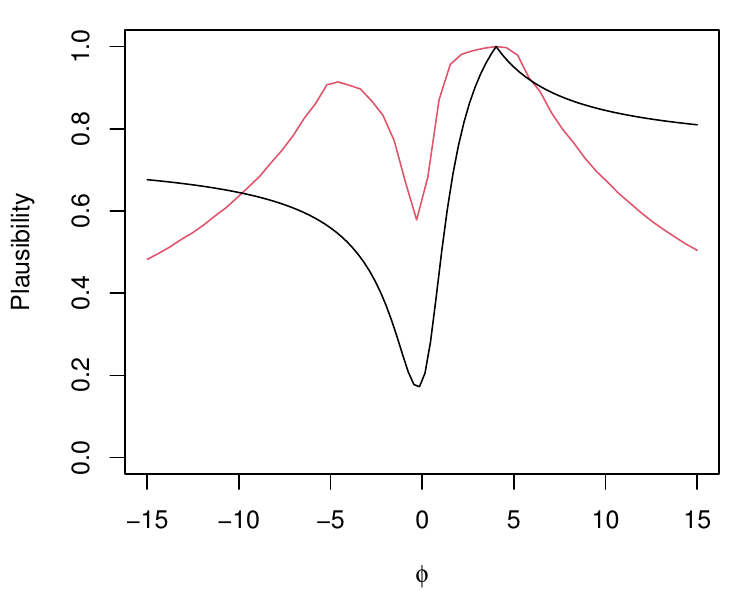}}
\end{center}
\caption{Marginal IM contours for the mean ratio $\Phi$ in the Fieller--Creasy illustration described in Example~\ref{ex:fc}, based on data $y=(1.33, 0.33)$.  Black line is for the vacuous prior case, red line is for the partial prior case.}
\label{fig:fc}
\end{figure}



\subsection{Two challenging practical examples}
\label{SS:big.examples}

The examples above are mostly illustrative in nature, shedding light on how the marginal IM construction works.  Here I want to consider two practically important and non-trivial examples---namely, the Behrens--Fisher and gamma mean problems---which I see as {\em challenge examples} for any framework designed for efficient marginal inference. 

\begin{ex}[Behrens--Fisher]
\label{ex:bf}
\citet[][Sec.~3]{fisher1935a} refers to the following inference problem.  Let $Y$ consist of mutually independent pairs $(Y_{ki})$, where 
\[ Y_{k1},\ldots,Y_{kn_k} \iid \nm(\mu_k, \sigma_k^2), \quad k=1,2, \]
with $\theta=(\mu_1, \mu_2, \sigma_1^2, \sigma_2^2)$ unknown: the goal is inference on the difference $\phi = \mu_2 - \mu_1$.  Fisher's solution makes reference to a relevant result by Walter Behrens, and hence the problem became known as the {\em Behrens--Fisher problem}.  This example has attracted a lot of attention over the years---and still does---because unlike many other problems involving inference on parameters of a normal model, this one doesn't admit a pivot that provides an exact solution.  The issue, of course, is the two variances corresponding to the two normal populations are not assumed to have any relationship; if they were equal or proportional, then the difficulties inherent in the Behrens--Fisher problem would disappear.  For more details on the history and the various proposed solutions to the Behrens--Fisher problem, see, e.g., \citet{kimcohen1998}.  Here I'll construct a strongly valid and efficient marginal IM for the unknown $\Phi$.  

In this case, the relative profile likelihood $R(y,\phi)$ doesn't have a closed-form expression, but it's not difficult to evaluate numerically.  However, it turns out that the distribution of $R(Y,\phi)$, as a function of the random variable $Y$ with $\phi$ as the true mean difference, only depends on the variance ratio, $\lambda=\sigma_1^2/\sigma_2^2$.  So, for a given data set $y$, it's easy to evaluate (via Monte Carlo) the function 
\begin{equation}
\label{eq:bf.lambda}
\phi \mapsto \prob_{Y|\phi,\lambda}\{ R(Y,\phi) \leq R(y,\phi) \}, \quad \text{for any $\lambda$}. 
\end{equation}
This suggest a marginal IM for $\Phi$ with contour function that equals the pointwise maximum of \eqref{eq:bf.lambda} over $\lambda$, i.e., 
\begin{equation}
\label{eq:bf.contour}
\pi_y(\phi) = \sup_{\lambda > 0} \prob_{Y|\phi,\lambda}\{ R(Y,\phi) \leq R(y,\phi) \}, \quad \phi \in \RR, 
\end{equation}
where, numerically, the supremum is replaced by a max over a grid of $\lambda$ values; in my experiments (see below), I've seen very little variation across the different values of $\lambda$ on this grid.  This marginal IM is strongly valid, which implies that the marginal $100(1-\alpha)$\% plausibility interval for $\Phi$ has guaranteed coverage probability at least $1-\alpha$ uniformly across samples sizes and nuisance parameter values.  Most of the other methods available in the literature can only offer asymptotic coverage guarantees.  

As an illustration, consider the often-used data in \citet[][p.~83]{lehmann1975} on travel times to work via two different routes.  This is the same example used by \citet{kimcohen1998} and others.  The relevant summary statistics---sample sizes, sample means, and sample standard deviations---are as follows:
\begin{align*}
n_1 & = 5 & \text{mean}(y_1) & = 7.580 & \text{sd}(y_1) & = 2.237 \\
n_2 & = 11 & \text{mean}(y_2) & = 6.136 & \text{sd}(y_2) & = 0.073.
\end{align*}
Note the relatively wide discrepancy between the two standard deviations; this would make it difficult to justify treating this using a simpler model formulation where the two normal variances are assumed to be the same.  Figure~\ref{fig:bf}(a) plots a couple different things: the gray lines correspond to the functions \eqref{eq:bf.lambda} for 75 different values of $\lambda$ over the range 0.001 to 100, and the black line is the marginal IM contour \eqref{eq:bf.contour}, the pointwise maximum of the gray curves, with the vertical lines marking the endpoints of the marginal IM's 95\% plausibility interval for $\Phi$.  Note that there's very little variation in the gray curves indexed by different values of the nuisance parameter.  For comparison, the 95\% confidence intervals for $\Phi$ produced by several different methods are presented in Table~\ref{tab:bf}. The solution by \citet{hsu1938} and \citet{scheffe1970} in the top row has coverage guarantees but tends to be conservative; the marginal IM solution presented in \citet{immarg} gives interval estimates that agree with the Hsu--Scheff\'e intervals.  The methods in the 2nd through 4th rows of the table offer only approximate coverage probability guarantees.  The last row is the marginal IM solution presented here, and note that my interval is the shortest of all those presented.  

It's worth to explore the performance of the proposed marginal IM solution compared to other methods across data sets.  For brevity, I'll reproduce one part of the extensive simulation study carried out in \citet{fraser.wong.sun.2009}.  I've chosen one of the most difficult, unbalanced cases with sample sizes $n_1=2$ and $n_2=20$.  In these simulations, the true values are $\mu_1=2$, $\mu_2=0$, $\sigma_1^2=1$, and $\sigma_2^2=2$; note that the true mean difference is $\phi=-2$.  I ran 10000 simulations and Figure~\ref{fig:bf}(b) plots the (estimated) distribution function of the random variable $\pi_Y(\phi)$, a function of the simulated data.  It's clear that this distribution function closely follows the diagonal line, indicating that $\pi_Y(\phi)$ is exactly (or at least approximately) uniformly distributed.  Therefore, even in this difficult unbalanced setting, the coverage probability of the marginal IM solution matches the nominal level exactly (up to Monte Carlo error), no sign of conservatism.  For comparison, the coverage probability for the methods compared in \citet{fraser.wong.sun.2009} are presented in Table~\ref{tab:bf2}.  With the exception of Jeffreys's default-prior Bayes solution (which agrees with Fisher's fiducial solution) and the 3rd order accurate solution, all the existing methods fall drastically short of the 90\% coverage target.  The profile-based marginal IM solution proposed here, on the other hand, hits the target on the nose.  
\end{ex}

\begin{table}[t]
\begin{center}
\begin{tabular}{ccc}
\hline
Method & Lower limit & Upper limit \\
\hline
\citet{hsu1938}, \citet{scheffe1970} & $-3.314$ & 0.427 \\
\citet{fisher1935a}, \citet{jeffreys1940} & $-3.308$ & 0.421 \\
\citet{welch1938} & $-3.293$ & 0.406 \\
\citet{welch1947}, \citet{aspin1948} & $-3.273$ & 0.386 \\
{\em Profile marginal IM} & $-3.106$ & 0.227 \\
\hline 
\end{tabular}
\end{center}
\caption{Lower and upper limits of various 95\% confidence intervals for the difference of means $\Phi$ based on the Lehmann's data.  The values in this table (except for the last row) are taken from \citet[][Table~2]{kimcohen1998}.}
\label{tab:bf}
\end{table}

\begin{figure}[t]
\begin{center}
\subfigure[Marginal IM contour for $\Phi$]{\scalebox{0.6}{\includegraphics{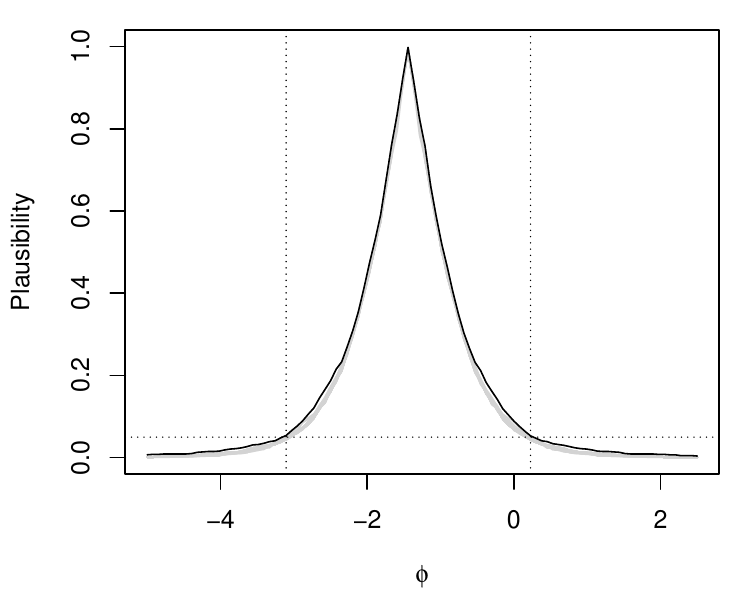}}}
\subfigure[Estimated CDF of $\pi_Y(\phi)$]{\scalebox{0.6}{\includegraphics{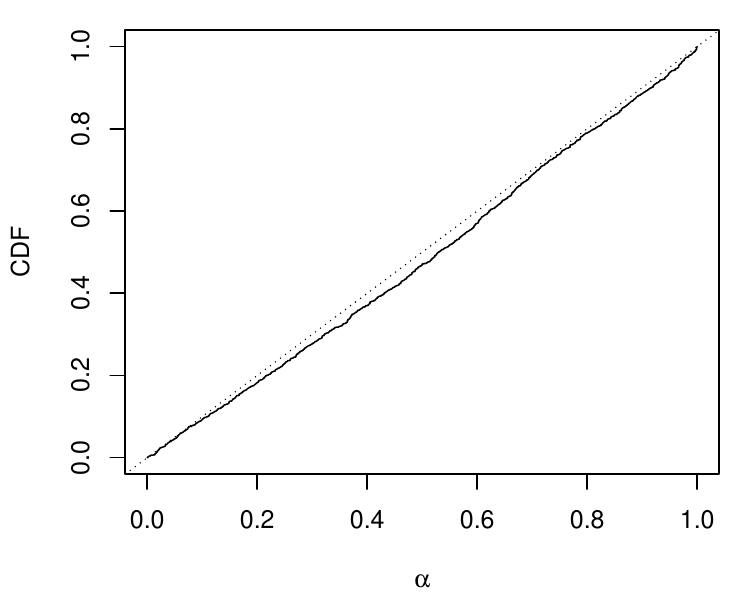}}}
\end{center}
\caption{Panel~(a) shows several nuisance parameter-dependent tentative contours (gray) and the actual marginal IM contour (black) for the mean difference $\Phi$ in the Behrens--Fisher problem based on Lehmann's travel time data.  Panel~(b) shows the estimated distribution function of the random variable $\pi_Y(\phi)$ as a function of simulated data as described in the text, where $\phi$ is the true value of the mean difference.}
\label{fig:bf}
\end{figure}

\begin{table}[t]
\begin{center}
{\small
\begin{tabular}{cccccc}
\hline
\citet{jeffreys1940} & \citet{ghosh.kim.2001} & \citet{welch1947} & 1st order & 3rd order & IM \\
\hline
0.9296 & 0.7873 & 0.8362 & 0.7399 & 0.8617 & 0.9082 \\
\hline 
\end{tabular}
}
\end{center}
\caption{Coverage probability of various 90\% confidence intervals for $\Phi$ in a very unbalanced version of the Behrens--Fisher problem, with $n_1=2$ and $n_2=20$; all but the last entry is taken from \citet[][Table~1a]{fraser.wong.sun.2009}.  Here ``1st order'' and ``3rd order'' correspond to the likelihood ratio-based approximations of Fraser et al.}
\label{tab:bf2}
\end{table}

\begin{ex}[Gamma mean]
\label{ex:gamma.mean}
Consider the two-parameter gamma model with unknown shape parameter $\Theta_1 > 0$ and unknown scale parameter $\Theta_2 > 0$.  Example~10 in Part~II presents a joint IM for simultaneous inference on the pair $\Theta = (\Theta_1, \Theta_2)$.  From this, of course, one can obtain (naive) marginal inference on any feature $\Phi=f(\Theta)$ of $\Theta$.  The focus here in the present example is on the construction of an efficient marginal IM specifically for inference on the mean $\Phi = \Theta_1\Theta_2$, the product.  This is a challenging marginal inference problem that has attracted the attention of a number of researchers, including \citet{grice.bain.1980}, \citet{shiue.bain.1990}, \citet{wong1993}, \citet{fraser.reid.wong.1997}, and \citet{bhaumik.etal.gamma}.  Here I construct a strongly valid marginal IM for $\Phi$ using the general machinery described above.  

Let $Y=(Y_1,\ldots,Y_n)$ denote an iid sample from a gamma distribution with the following parametrization: the shape parameter is $\Lambda$ and the scale parameter is $\Phi/\Lambda$, so that the mean of the distribution is $\Phi$.  Unfortunately, there is no closed-form expression for the relative profile likelihood $R(y,\phi)$, but it can be readily evaluated numerically.  The more significant challenge is that the distribution of the relative profile likelihood depends on both the interest and nuisance parameters, so some effort is required to evaluate the corresponding Choquet integral.  Indeed, in the case of vacuous prior information, the marginal IM contour for $\Phi$ is given by 
\[ \pi_y(\phi) = \sup_{\lambda > 0} \prob_{Y|\phi,\lambda}\{ R(Y,\phi) \leq R(y,\phi) \}, \quad \phi > 0, \]
and the computational obstacle is evaluating the supremum over $\lambda$.  Since the relative profile likelihood is, by Wilks's theorem, an approximate pivot in this example, one would expect that the right-hand side's dependence on $\lambda$ is relatively mild, so this can be accurately approximated by maximizing over a relatively coarse grid of $\lambda$ values.  

For illustration, consider the real data presented in Example~3 of \citet{fraser.reid.wong.1997}, which consists of the survival time (in weeks) for $n=20$ rats exposed to a certain amount of radiation.  
Figure~\ref{fig:gamma.mean} displays a plot of the mapping 
\begin{equation}
\label{eq:gamma.pl.tmp}
\phi \mapsto \prob_{Y|\phi,\lambda}\{ R(Y,\phi) \leq R(y,\phi)\}, 
\end{equation}
for a range of different values of the shape parameter $\lambda$, along with the corresponding marginal IM's contour based on optimizing over $\lambda$.  There are a few key observations worth making here.  First, since gamma is an exponential family model, the maximum likelihood estimator of $\Phi$ is the sample mean, $\bar y \approx 113.5$, which is where the plausibility contours peak.  Second, note that there is effectively no change in the curves \eqref{eq:gamma.pl.tmp} as $\lambda$ varies in the grid $\{0.1, 0.5, 1, 5, 10, 50, 100\}$, which suggests that simple approximations of the marginal IM contour, e.g., by using \eqref{eq:gamma.pl.tmp} with $\lambda$ fixed at its maximum likelihood estimator, ought to be reasonably accurate.  Third, the vertical bars correspond to the 95\% confidence intervals based on three other methods as presented in \citet{fraser.reid.wong.1997}: the red line is based on the first-order normal approximation of the sampling distribution of the maximum likelihood estimator; the green line is based on the more sophisticated method in \citet{shiue.bain.1990}; and the blue line is based on the third-order approximation derived in \citet{fraser.reid.wong.1997}.  The marginal IM's interval is quite different from those bounded by the red and green lines, which isn't too surprising given that the latter are based on relatively crude approximations.  That the marginal IM's interval closely matches that bounded by the blue lines also isn't surprising because the IM solution is precisely that which the third-order method is aiming to approximate.  Beyond the confidence interval comparisons, the marginal IM for $\Phi$ is provably valid and efficient, so it can reliably answer any relevant question concerning the mean $\Phi$.  

To dig deeper into the efficiency of the proposed IM solution, we reproduce the simulation study in \citet[][Ex.~2]{fraser.reid.wong.1997}.  In particular, I generate 10000 samples of size $n=10$ from a gamma distribution with shape 2 and mean $\phi=1$.  Panel~(b) of Figure~\ref{fig:gamma.mean} shows the estimated distribution function of the random variable $\pi_Y(\phi)$, and it's indistinguishable from uniform, as the general theory indicates.  Table~\ref{tab:gamma.mean} shows p-value percentages for a number of different methods available in the literature: including the textbook solution based on the first-order approximate normality of the maximum likelihood estimator, the methods proposed in \citet{shiue.bain.1990} and \citet{wong1993}, and the two third-order signed likelihood root approximations presented in \citet{fraser.reid.wong.1997}.  With the exception of the textbook first-order solution, which is terrible, all the methods perform similarly.  The marginal IM solution is a bit more conservative than some of the others---see the ``$< 0.5\%$'' to ``$< 2.5\%$'' bins---which is perhaps to be expected given that it has the exact validity constraint and the sample size is small.
\end{ex}

\begin{figure}[t]
\begin{center}
\subfigure[Marginal IM contour for $\Phi$]{\scalebox{0.6}{\includegraphics{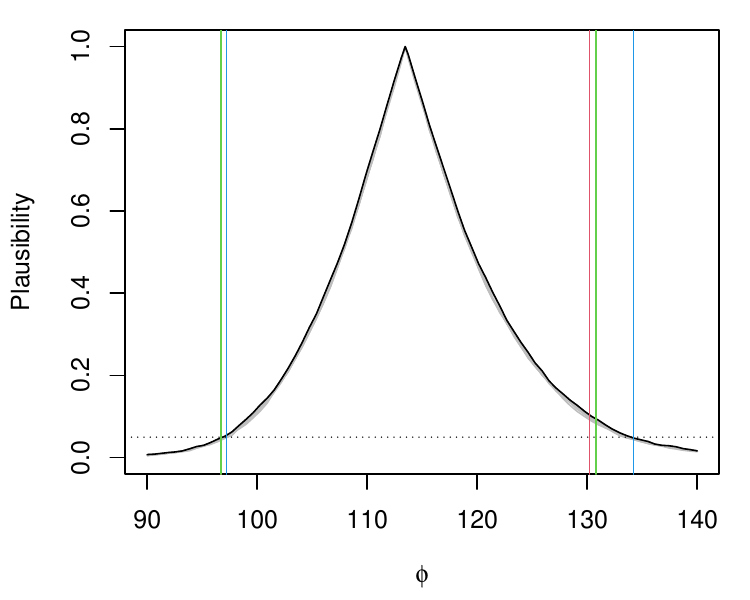}}}
\subfigure[Estimated CDF of $\pi_Y(\phi)$]{\scalebox{0.6}{\includegraphics{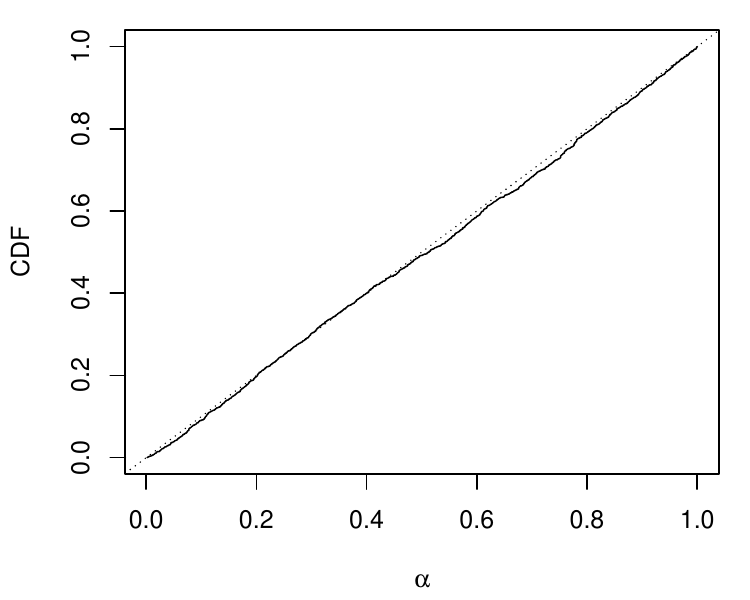}}}
\end{center}
\caption{Panel~(a): Marginal IM contour for the gamma mean as in Example~\ref{ex:gamma.mean}. The gray curves correspond to the function \eqref{eq:gamma.pl.tmp} for particular values of the gamma shape parameter $\lambda$; the black line represents the marginal IM contour for $\Phi$ based on a point-wise maximum of the gray curves. The colored vertical lines mark confidence intervals based on various other methods, as explained in the text. Panel~(b): estimated distribution function of the random variable $\pi_Y(\phi)$ where $\phi$ is the true value of the mean.}
\label{fig:gamma.mean}
\end{figure}

\begin{table}[t]
\begin{center}
{\small
\begin{tabular}{ccccccc}
\hline 
Method & $< 0.5\%$ & $< 1.25\%$ & $< 2.5\%$ & $>97.5\%$ & $> 98.75\%$ & $>99.5\%$ \\
\hline 
1st order & 5.00 & 5.73 & 8.53 & 1.52 & 0.74 & 0.35 \\
\citet{shiue.bain.1990} & 0.28 & 0.70 & 1.19 & 1.20 & 0.67 & 0.30 \\
\citet{wong1993} & 0.29 & 0.69 & 2.18 & 2.05 & 0.95 & 0.41 \\
3rd order (a) & 0.37 & 0.90 & 2.30 & 2.41 & 1.13 & 0.47 \\
3rd order (b) & 0.37 & 0.91 & 2.30 & 2.41 & 1.13 & 0.47 \\
{\em Profile marginal IM} & 0.22 & 0.65 & 1.57 & 2.71 & 1.31 & 0.59 \\
\hline 
\end{tabular}
}
\end{center}
\caption{Results for the gamma mean simulation in Example~\ref{ex:gamma.mean}.  The table shows the percentage of $p$-values in the stated bins (across 10000 simulations) for the various methods which relates to the coverage probability of the corresponding confidence intervals for the mean $\Phi$.  All the values except for the last row are taken from Table~2 in \citet{fraser.reid.wong.1997}. The standard errors across the board are all less than 0.16\%.}
\label{tab:gamma.mean}
\end{table}


\subsection{Words of caution}
\label{SS:caution}

It was demonstrated above that in both ideal and less-than-ideal factorization cases, the profiling strategy effectively eliminated nuisance parameters and led to valid and efficient marginal IMs for the interest parameter.  As hinted at previously, however, there are examples outside the ideal-factorization cases in which profiling will fail to give an efficient marginal IM solution, so blind application of profiling is dangerous.  Here I'll look at two such examples.  Fortunately, what makes the profiling strategy fail can be spotted in the problem setup, and other dimension reduction strategies can be used.  

\begin{ex}[Mean vector length]
\label{ex:length}
Consider the classical problem in which $Y$ is a $n$-dimensional normal random vector with unknown mean vector $\Theta$ and known covariance matrix, say, the identity matrix $I$.  Inference on the mean itself is the same regardless of the dimension, but suppose that the quantity of interest is $\Phi = f(\Theta) = \|\Theta\|$, the Euclidean length of the mean vector.  As mentioned above, the mean vector itself can be reconstructed by introducing the dual nuisance parameter $\Lambda = \Theta/\|\Theta\|$ that represents the unit vector in the direction of $\Theta$.  This turns out to be a non-trivial problem, as first pointed out by \citet{stein1956, stein1959}; it was also recently listed in \citet{fraser.reid.lin.2018} as one of the ``challenge problems'' put forward by the late Sir D.~R.~Cox.  

The likelihood function is relatively straightforward in this case:
\[ p_{Y|\theta}(y) = \exp\{ -\tfrac12 \|y - \theta\|^2 \}, \quad \theta \in \TT = \RR^n. \]
In terms of $(\phi, \lambda)$, the log-likelihood can be re-expressed as 
\[ \log p_{Y|\phi,\lambda}(y) = -\tfrac12 \|y - \phi\lambda\|^2 = -\tfrac12 \|y\|^2 + (\phi \|y\|) \lambda^\top (y / \|y\|) - \tfrac12 \phi^2. \]
Note that the likelihood depends on $y$ only through the pair $\{U(y), V(y)\}$, where $U(y) = \|y\|$ and $V(y) = y / \|y\|$.  It's also relatively well-known \citep[e.g.,][Sections~3.5.4 and 9.3.2]{mardia.jupp.book} that the joint distribution can be factored as 
\[ p_{Y|\theta}(y) = p_{U|\phi}(u) \, p_{V|u,\phi,\lambda}(v), \]
where $p_{U|\phi}$ is a non-central chi-square density and $p_{V|u,\phi,\lambda}$ is a von Mises--Fisher distribution.  The specific details of the von Mises--Fisher distribution aren't relevant here; the key point is that the conditional distribution of $V$, given $U=u$, depends on both $\phi$ and $\lambda$.  Therefore, this isn't one of the ``ideal factorization'' cases discussed above.  It is, however, a ``less-than-ideal factorization,'' so there's a choice to be made: do we follow the general profile likelihood strategy, which avoids throwing away relevant information in $V$, or go for the simplicity of the marginal likelihood based on $U$?  Both IM solutions are valid, so efficiency considerations will tip the scale in one direction or the other.  

The profile-based IM solution will have a contour that's maximized at the global maximum likelihood estimator $\hat\phi_y$ of $\Phi$.  It's well-known that that the maximum likelihood estimator of $\Theta$ is the data vector $y$ so, by invariance of maximum likelihood estimators, it immediately follows that $\hat\phi_y = u = \|y\|$, the length of the data vector $y$.  Indeed, the profile likelihood function is 
\[ \sup_{\lambda} p_{Y|\phi,\lambda}(y) = \exp\{-\tfrac12( \|y\| - \phi )^2\}, \quad \phi > 0, \]
so above claim is easily confirmed without applying the invariance principle. Although it's natural to estimate the mean vector length by the data vector length, this is actually a pretty lousy estimator: it has non-negligible upward bias.  Intuitively, the reason why the maximum likelihood estimator of $\Phi$ has such poor properties is this case is that there are too many nuisance parameters, in fact, $O(n)$ many.  If the peak of the IM contour for $\Phi$ is attained at a severely biased estimator, then it must be rather wide in order to achieve validity; see Figure~\ref{fig:length}.  This suggests a deficiency in the profile likelihood-based IM solution, so let's consider the marginal likelihood-based solution.  

Neither the profile nor the marginal likelihood-based IM contour functions have a closed-form expression, but they can readily evaluated numerically via Monte Carlo.  To get a feeling of why the marginal likelihood-based solution might be superior, note that the corresponding contour will be maximized the maximum marginal likelihood estimator, i.e., $\arg\max_\phi p_{U|\phi}(u)$.  There's no closed-form expression for this either, but a decent approximation is the method of moments estimator: $(U^2 - n)^{1/2}$.  Note the automatic correction for the upward bias of $U$ mentioned above.  Since the peak of the marginal likelihood-based IM contour for $\Phi$ tends to be closer to the true value, more efficiency is expected compared to the profile likelihood-based solution.  

To see the above claims in action, I simulated a Gaussian random vector of dimension $n=5$ with mean vector such that the true value of $\Phi$ is 3.  Note that the marginal likelihood-based IM contour has peak closer to the true value of $\Phi$ than the profile likelihood-based IM contour, and it's also narrower; note that the 95\% plausibility interval based on the former solution is significantly narrower than that based on the latter. 
\end{ex}

\begin{figure}[t]
\begin{center}
\scalebox{0.65}{\includegraphics{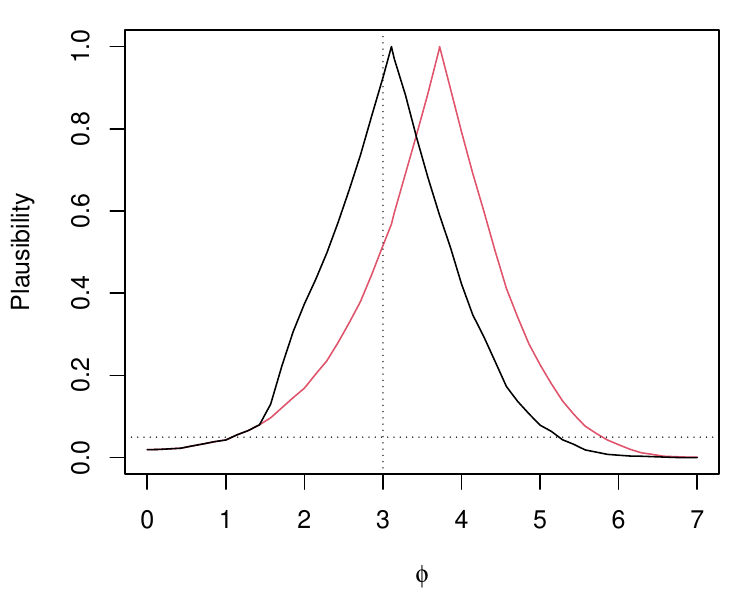}}
\end{center}
\caption{Marginal IM contours for the mean vector length $\Phi$ in Example~\ref{ex:length} based on the marginal likelihood (black) and the profile likelihood (red).}
\label{fig:length}
\end{figure}

\begin{ex}[Neyman--Scott]
\label{ex:ns}
In the early days of statistical theory, it was generally believed that maximum likelihood estimators were consistent. This famous example put forward in \citet{neyman.scott.1948} busted this myth by showing that the maximum likelihood estimator could be inconsistent.  The problem setup is as follows.  Suppose that data $Y$ consists of independent pairs $Y_i = (Y_{i1}, Y_{i2})$, where 
\[ Y_{i1}, Y_{i2} \iid \nm(\mu_i, \sigma^2), \quad i=1,\ldots,n. \]
The quantity of interest $\Phi$ is the uncertain value of the variance $\sigma^2$, and the nuisance parameter $\Lambda$ is the unknown vector of means $(\mu_1, \ldots, \mu_n)$.  It's a standard exercise in a stat theory course to show that the maximum likelihood estimator is 
\[ \hat\phi_Y = \frac1n \sum_{i=1}^n (Y_{i1} - Y_{i2})^2, \]
and that its in-probability limit is $\frac12 \Phi \neq \Phi$.  Of course, this asymptotic bias is easy to correct for, but that's not the point: this is a sign of some inadequacy in the naive likelihood-based solution.  Therefore, we should know here to proceed with caution in our IM construction.  In particular, like in Example~\ref{ex:length} above, we can expect that a profile likelihood-based solution, which returns a contour function maximized at the flawed $\hat\phi_Y$ above, so a marginal likelihood-based approach might be preferred.  

As demonstrated clearly in \citet[][Sec.~2.4]{boos.stefanski.2013}, and perhaps elsewhere, the change-of-variables $y \mapsto \{U(y), V(y)\}$, with 
\[ U_i = 2^{-1/2} (Y_{i1}-Y_{i2}) \quad \text{and} \quad V_i = 2^{-1/2}(Y_{i1} + Y_{i2}), \quad i=1,\ldots,n, \]
has the following properties:
\begin{itemize}
\item the marginal distribution of $U=(U_1,\ldots,U_n)$ doesn't depend on $(\mu_1,\ldots,\mu_n)$;
\vspace{-2mm}
\item $(U_1,\ldots,U_n)$, given $\Phi=\phi$, are iid $\nm(0,\phi)$;
\vspace{-2mm}
\item and $U$ and $V$ are independent. 
\end{itemize}
Since the marginal distribution of $V$ depends on both $\Phi$ and $\Lambda$, we're in the $\Phi$-oriented case under the ``less-than-ideal factorization'' above.  Instead of profiling, I suggest taking the marginal likelihood for $\Phi$ based on $U$ alone.  The maximum marginal likelihood estimator is easily shown to be $2\hat\phi_Y$, which is consistent.  From here the story is very similar to that in Example~\ref{ex:length} above, so I'll skip the details.  But the point, again, is that, when there are a growing number of nuisance parameters and the ordinary maximum likelihood estimator is inconsistent, a valid and efficient marginal IM should for $\Phi$ be constructed based on the marginal likelihood.  It's counter-intuitive, at least to me, that efficient marginal inference in these cases requires that some information about $\Phi$ in $V$ be intentionally ignored.  
\end{ex}

\section{Predictive possibilistic IMs}
\label{S:prediction}

\subsection{Setup}

Prediction of future responses is a fundamentally important problem in statistics and machine learning.  Here I'll focus primarily on the simple but practical case of prediction with respect to an assumed parametric model.  Specifically, let $(Y, Z \mid \Theta=\theta) \sim \prob_{Y,Z|\theta}$ be the statistical for the pair $(Y,Z) \in \YY \times \ZZ$ consisting of an observable $Y$ and a to-be-predicted $Z$, depending on an unknown $\Theta$.  Note that $Y$ and $Z$ might not have the same form, e.g., $Y$ and $Z$ might have different marginal distributions due to dependence on distinct covariate values (hidden in the notation), $Z$ might be a function, such as a maximum, of several independent $Y$-like realization, or, more generally, $Y$ might be the sufficient statistic based on a sample of $Z$-like random variables.  The particular form of the relationship between $Y$ and $Z$ isn't important; all that matters is the dependence on a common parameter $\Theta$.  A valid IM construction in this case was presented in \citet{impred}, but the approach here is more general in several respects.  Partial prior information could be available about $\Theta$ in the form of an upper probability distribution with contour $q(\theta) = q_\Theta(\theta)$.  This problem roughly fits into the setup described in the previous section, where $\Phi=Z$ is of primary interest and $\Lambda=\Theta$ is of only secondary importance.  From this perspective, it's clear that prediction is just an extreme version of the marginal inference problem where the model parameter itself is a nuisance.  

Recall that marginalization can be carried out in at least two ways: one based on applying the basic extension principle to the joint IM's output and the other taking the priorities of the analysis into special consideration in the IM construction.  For the prediction problem here it turns out that there are even more ways one could imagine proceeding.  In particular, I can see three options for constructing a {\em predictive IM} for inference on/prediction of $Z$:
\begin{enumerate}
\item First construction an IM for $\Theta$, given $Y=y$, as described above, then combine this in an appropriate way with a possibilistic representation of the model that relates $Z$ to $\Theta$ and then marginalize to $Z$ via the extension principle.  This is similar to Bayes's rule combines first by multiplying the model density and the posterior density, and then integrates to get the predictive density. 
\vspace{-2mm}
\item Construct a full IM for the pair $(Z, \Theta)$, given $Y=y$, as described above, and then marginalize to $Z$ via the extension principle. 
\vspace{-2mm}
\item Use the fact that $\Theta$ is a nuisance parameter in order to reduce dimension before constructing a marginal IM for $Z$.  
\end{enumerate} 
These three options are in decreasing order of simplicity and increasing order of efficiency.  That is, Option~1 is relatively simple but less efficient, while Option~3 is less simple but more efficient.  Option~2 isn't particularly appealing because, usually, either simplicity or efficiency is the priority but this one achieves neither.  So I'll only briefly describe Option~2 and focus my attention on Options~1 and 3; see Section~\ref{SS:pred.three} below. 

Before describing the details of these options, it'll help to be clear about what the objectives are.  Like in the previous sections, my goal here is to construct a pair $(\lPi_y, \uPi_y)$ of lower and upper probabilities on $\ZZ$, with a consonance structure so that they're fully determined by a contour function $\pi_y(z)$.  Analogous to \eqref{eq:valid}, I'll require that this predictive IM have a {\em strong prediction validity} property, i.e., 
\begin{equation}
\label{eq:strong.pred.validity}
\uprob_{Y,Z,\Theta}\{ \pi_Y(Z) \leq \alpha \} \leq \alpha, \quad \alpha \in [0,1]. 
\end{equation}
Since the event ``$\pi_Y(Z) \leq \alpha$'' doesn't directly depend on $\Theta$, the $\uprob_{Y,Z,\Theta}$-probability on the left-hand side above boils down to a supremum over all marginal distributions for $(Y,Z)$ induced by precise priors $\prior$ for $\Theta$ in the prior credal set $\credal$.  That is, the left-hand side of the above display can be rewritten as 
\[ \sup_{\prior \in \credal} \int_\TT \prob_{Y,Z|\theta}\{ \pi_Y(Z) \leq \alpha \} \, \prior(d\theta). \]
In the vacuous-prior case, for example, the property in \eqref{eq:strong.pred.validity} reduces to 
\[ \sup_\theta \prob_{Y,Z|\theta}\{ \pi_Y(Z) \leq \alpha \} \leq \alpha, \quad \alpha \in [0,1], \]
which is probably a more familiar looking condition to the reader. 

The setup above and the details below focus on the very special---albeit practically relevant---case where prediction is with respect to a specified parametric model.  It many applications, however, the goal is to carry out prediction without assuming a particular distributional form for the observables.  It turns out that this more general prediction problem can also be handled within the proposed framework, but I'll postpone discussion of this till Section~\ref{S:beyond} below.

\subsection{Three valid IM constructions}
\label{SS:pred.three}

\subsubsection*{Option 2}

I'll start with a brief explanation of Option~2.  This isn't an ideal solution, in my opinion, because it's neither the simplest nor the most efficient.  But the construction here will help with the development of Option~3 below, so it's worth briefly mentioning this case.

Recall that Option~2 suggests first constructing an IM for the pair $(Z,\Theta)$ and then marginalizing over $\Theta$ using the extension principle, leaving a marginal IM for $Z$.  Let $p_{Y,Z|\theta}(y,z)$ denote the joint density/mass function of $(Y,Z)$, given $\Theta=\theta$. For the first step, following the general framework, the ``relative likelihood'' function is
\[ R_q(y, z, \theta) = \frac{p_{Y,Z|\theta}(y,z) \, q_\Theta(\theta)}{\sup_{x,\vartheta} \{p_{Y,Z|\vartheta}(y,x) \, q_\Theta(\vartheta)\}}, \quad (y,z,\theta) \in \YY \times \ZZ \times \TT. \]
In general, there's no further simplification that can be made to this relative likelihood function.  So to proceed with the IM construction for $(Z, \Theta)$, I'd simply plug this formula into the general recipe described above, which leads to the an IM contour function 
\[ \pi_y(z, \theta) = \uprob_{Y,Z,\Theta}\{ R_q(Y,Z, \Theta) \leq R_q(y, z, \theta)\}, \quad (z,\theta) \in \ZZ \times \TT. \]
This joint IM for $(Z, \Theta)$ is strongly valid by the general theory in Part~II; moreover, the marginal predictive IM for $Z$, with contour 
\[ \pi_y(z) = \sup_{\theta \in \TT} \pi_y(z, \theta), \quad z \in \ZZ, \]
would also satisfy the strong prediction validity as described in the previous subsection as a consequence of the validity-preserving property of the extension principle.  I'll have more to say about the prediction validity property for Options~1 and 3 below.  

In certain special cases, however, some simplifications can be made and dimension can be reduced for the sake of efficiency.  These are the vacuous- and complete-prior cases.
\begin{itemize}
\item Recall that, with a vacuous prior, i.e., $q_\Theta(\theta) \equiv 1$, there's an opportunity to reduce dimension by {\em fixing $\theta$} in the relative likelihood.  In this case, the IM contour is 
\[ \pi_y(z, \theta) = \prob_{Y,Z|\theta}\{ R(Y,Z, \theta) \leq R(y, z, \theta)\}, \quad (z,\theta) \in \ZZ \times \TT, \]
and efficiency is gained by not integrating/optimizing over $\theta$. 
\vspace{-2mm}
\item With a complete prior, where $q$ now denotes the prior probability density/mass function, the dimension reduction is achieved by {\em fixing $y$}.  That is, the numerator in the expression for $\eta$ can be rewritten as 
\[ p_{Y,Z|\theta}(y,z) \, q_\Theta(\theta) \propto p_{Z|y,\theta}(z) \, q_{\Theta|y}(\theta), \]
where ``$\propto$'' means as a function of $(z, \theta)$, and $q_{\Theta|y}(\theta)$ is the Bayesian posterior density/mass function for $\Theta$, given $Y=y$; the proportionality ``constant'' is the marginal density of $Y$, with respect to prior $q_\theta$, which only depends on $y$.  That proportionality constant cancels in the ratio that defines $\eta$, after some further simplifications, the IM for $(Z, \Theta)$ has contour function 
\[ \pi_y(z, \theta) = \prob_{Z, \Theta|y}\{ p_{Z|y,\Theta}(Z) \, q_{\Theta|y}(\Theta) \leq p_{Z|y,\theta}(z) \, q_{\Theta|y}(\theta)\}, \]
where the probability is with respect to the conditional distribution of $(Z, \Theta)$, given $Y=y$, which has density/mass function $p_{Z|y,\theta}(z) \, q_{\Theta|y}(\theta)$.  
\end{itemize}

\subsubsection*{Option 3}

Next I'll describe Option~3 because it's more in line with the development in Section~\ref{S:nuisance} above than is Option~1.  Recall that Option~3 is based on the understanding that only prediction of $Z$ is relevant and, therefore, $\Theta$ is a nuisance parameter.  Then the marginal inference strategy described in the previous section can be applied, which suggests the relative profile likelihood 
\[ R_q(y,z) = \frac{\sup_\theta \{p_{Y,Z|\theta}(y,z) \, q_\Theta(\theta)\}}{\sup_{x,\theta} \{p_{Y,Z|\theta}(y,x) \, q_\Theta(\theta)\}}, \quad (y,z) \in \YY \times \ZZ. \] 
Of course, if the prior information for $\Theta$ were vacuous, then $q_\Theta(\theta) \equiv 1$ and that term can be dropped from the above expression; the case of a complete prior will be considered below.  In any case, the construction of a {\em predictive} IM for $Z$, given $Y=y$, now proceeds by defining the contour function 
\[ \pi_y(z) = \uprob_{Y, Z, \Theta}\{ R_q(Y,Z) \leq R_q(y, z)\}, \quad z \in \ZZ. \]
When the partial prior information is encoded in a possibility measure, the above Choquet integral can be simplified as before, i.e., 
\[ \pi_y(z) = \int_0^1 \Bigl[ \sup_{\theta: q_\Theta(\theta) > s} \prob_{Y,Z|\theta}\{ R_q(Y,Z) \leq R_q(y,z)\} \Bigr] \, ds, \quad z \in \ZZ. \]
In the vacuous prior case, with $q(\theta) \equiv 1$, this simplifies even further:
\[ \pi_y(z) = \sup_\theta \prob_{Y,Z|\theta}\{ R(Y,Z) \leq R(y,z)\}, \quad z \in \ZZ. \]
Depending on the structure of the problem, it may happen that $R(Y,Z)$ is a pivot (see Example~\ref{ex:normal.mean.pred}), in which case the IM computation becomes relatively straightforward.  For example, in the vacuous-prior case, if $R(Y,Z)$ is a pivot, then the supremum in the above display can be dropped and computation of the IM contour function is easy.  

That prediction validity \eqref{eq:strong.pred.validity} holds for the IM constructed above is easy to verify.  Indeed, 
\begin{align*}
\uprob_{Y,Z,\Theta}\{\pi_Y(Z) \leq \alpha\} & = \sup_{\prior} \prob_{Y,Z|\prior}\Bigl[ \sup_{\prior'} \prob_{Y',Z'|\prior'}\{R_q(Y',Z') \leq R_q(Y,Z)\} \leq \alpha \Bigr] \\
& \leq \sup_{\prior} \prob_{Y,Z|\prior}\bigl[ \prob_{Y',Z'|\prior}\{R_q(Y',Z') \leq R_q(Y,Z) \} \leq \alpha \bigr].
\end{align*}
The $\prior$-specific probability on the right-hand side is upper-bounded by $\alpha$ by standard arguments, for each $\prior$.  Then the supremum must also be bounded by $\alpha$, proving the strong prediction validity claim. 

Next, consider the case of a precise prior for $\Theta$, where $q_\Theta$ is the prior density/mass function.  As before, dimension can be reduced---and efficiency gained---by {\em fixing $y$}.  Towards this, note that the numerator of relative profile likelihood at the start of this subsection can be rewritten as 
\begin{align*}
\sup_\theta \{ p_{Y,Z|\theta}(y,z) \, q_\Theta(\theta) \} & = \sup_\theta \{ p_{Z|y}(z) \, q_{\Theta|y,z}(\theta) \, p_Y(y) \} \\
& = p_{Z|y}(z) \, p_Y(y) \, \sup_\theta q_{\Theta|y,z}(\theta), 
\end{align*}
where $p_{Z|y}(z)$ is the posterior predictive distribution of $Z$, given $Y=y$, $p_Y(y)$ is the marginal distribution of $Y$ under the Bayes model with prior $q_\Theta$, and $q_{\Theta|y,z}(\theta)$ is the posterior distribution of $\Theta$, given $(Y,Z) = (y,z)$.  The marginal density term, $p_Y(y)$, cancels in the ratio that defines $R$, which leaves just 
\[ R_q(y,z) = \frac{\sup_\theta q_{\Theta|y,z}(\theta)}{\sup_{x,\theta} \{p_{Z|y}(x) \, q_{\Theta|y,x}(\theta)\}} \, p_{Z|y}(z). \]
In the above display, the denominator of the leading term only depends on $y$ and, therefore, can be treated as a constant.  The numerator technically depends on $z$, but that dependence is rather weak.  In fact, in some cases (Example~\ref{ex:normal.mean.pred}), that term doesn't actually depend on $\tilde y$, so it too can be treated as a constant.  More generally, it'll be typically be the case that $y$ is more informative than $z$---recall that $y$ here represents the observed data set (or sufficient statistic) so it would have far more influence on the posterior than the single realization $z$.  Putting all this together, dropping proportionality constants where possible, the relative profile likelihood can be re-expressed as
\[ R_q(y, z) = \bigl\{ \sup_\theta q_{\Theta|y,z}(\theta) \bigr\} \, p_{Z|y}(z) \propto p_{Z|y}(z). \]
The ``$\propto$'' above is generally not exact---the term in curly brackets might depend mildly on $z$---but there's no clear reason not to just ignore that term and work with the predictive density itself on the right-hand side; more on this choice below.  Then the complete-prior predictive IM for $Z$, given $Y=y$, has contour function 
\[ \pi_y(z) = \prob_{Z|y}\{ p_{Z|y}(Z) \leq p_{Z|y}(z)\}, \quad z \in \ZZ, \]
where $\prob_{Z|y}$ denotes the posterior predictive distribution of $Z$, given $Y=y$, with respect to the Bayes model with (precise) prior $q$.  This is the probability-to-possibility transform of the Bayesian predictive distribution and, therefore, enjoys certain optimality properties.  For example, the $100(1-\alpha)$\% prediction regions derived from $\pi_y(z)$ above are optimal in the sense that they have the smallest Lebesgue measure of all sets having posterior predictive probability at least $1-\alpha$.

\subsubsection*{Option 1}

Option~1 is very different from the previous two options described above.  This is ideally suited for a case where, for example, the user is handed a strongly valid IM for $\Theta$, given $Y=y$, and the task is to update this to a predictive IM for $Z$, given $Y=y$.  This is different from the previous cases in that, with Options~2--3, the user knew that prediction of $Z$ was the primary goal and could focus directly on that task.  Here, it's as if the user first wanted inference on $\Theta$ and then later was asked to use that same IM for predictive inference on $Z$.  To keep the details relatively simple, here I'll assume that $Y$ and $Z$ are conditionally independent, given $\Theta$.  Lots of applications satisfy the independence assumption, so this is not a severe restriction; and results similar to those below ought to be possible even in certain cases where conditional independence fails.  

To set the scene, let $\pi_y(\theta)$ be the contour function of a strongly valid IM for $\Theta$, given $Y=y$.  Moreover, let $p_\theta(z)$ denotes the density of $Z$, given $\Theta=\theta$, and consider a possibilistic representation, say, $f_\theta(z)$ thereof:
\[ f_\theta(z) = \prob_{Z|\theta}\{ p_{Z|\theta}(Z) \leq p_{Z|\theta}(z)\}, \quad z \in \ZZ. \]
Then the strategy is to suitably combine $f_\theta(z)$ and $\pi_y(\theta)$ into a joint possibility distribution for $(Z,\Theta)$, given $Y=y$, and then marginalize out $\Theta$.  Various combination strategies are discussed in, e.g., \citet{destercke.etal.2009}, \citet{troffaes.etal.2013}, and \citet[][Sec.~3.5.4]{hose2022thesis} but, in the present context, they'd all take the form 
\begin{equation}
\label{eq:K}
\pi_y(z) = \sup_\theta \mathcal{K}\{ f_\theta(z), \pi_y(\theta) \}, 
\end{equation}
for a suitable function $\mathcal{K}: [0,1]^2 \to [0,1]$.  Intuitively, $f_\theta(z)$ encodes a conditional distribution of $Z$, given $\Theta=\theta$, and $\pi_y(\theta)$ encodes a conditional distribution of $\Theta$, given $Y=y$, so the role that $\mathcal{K}$ plays is like a possibilistic analogue of the probabilistic multiplication operation that converts this pair into a sort of joint distribution for $(Z,\Theta)$, given $Y=y$.  Then the supremum over $\theta$ on the outside corresponds to marginalization over $\Theta$, via the extension principle, to get a possibilistic conditional distribution for $Z$, given $Y=y$, which will play the role of a predictive IM.  

For statistical and historical reasons, I prefer the combination rule that's motivated by Fisher's classical strategy for combining p-values in significance testing contexts \citep[][Sec.~21.1]{fisher1973b}.  This boils down to the choice of $\mathcal{K}$ as 
\[ \mathcal{K}(u,v) = uv\{1 - \log(uv)\}, \quad (u,v) \in [0,1]^2. \]
The connection between this formula and Fisher's p-value combination rule is as follows.  In the context of significance testing, let $U$ and $V$ denote independent p-values, so that $U, V \iid \unif(0,1)$ under the null hypothesis.  Then $-2(\log U + \log V)$ has a chi-square distribution with 4 degrees of freedom, so the p-value for the combined test, based on the product of p-values rule, is given by 
\[ \prob(UV \leq uv) = \prob\{\underbrace{-2(\log U + \log V)}_{\text{$\sim \chisq(4)$}} \geq -2(\log u + \log v)\} \]
where $(u,v) \in [0,1]^2$ here denote the observed p-values.  Fisher would've stopped here and suggested compared the observed $-2(\log u + \log v)$ to the critical value of a $\chisq(4)$ distribution.  Jost,\footnote{\url{http://www.loujost.com/StatisticsandPhysics/SignificanceLevels/CombiningPValues.htm}, accessed October 31st, 2022} however, has shown that the chi-square probability in the above display can be evaluated in closed-form, and the expression is $\mathcal{K}(u,v)$.  

Suppose that the IM for $\Theta$, given $Y=y$, is strongly valid with respect to the vacuous prior.  In that case, both $f_\theta(Z)$ and $\pi_Y(\theta)$ are independent and stochastically no smaller than $\unif(0,1)$ under $\prob_{Y,Z|\theta}$.  Then strong prediction validity \eqref{eq:strong.pred.validity}, relative to the vacuous prior for $\Theta$, follows immediately from the Fisher/p-value connection described in the previous paragraph.  To see this, first note:
\begin{align*}
\sup_\theta \prob_{Y,Z|\theta}\{ \pi_Y(Z) \leq \alpha \} & = \sup_\theta \prob_{Y,Z|\theta}\bigl\{ \sup_\vartheta \mathcal{K}(f_\vartheta(Z), \pi_Y(\vartheta)) \leq \alpha \} \\
& \leq \sup_\theta \prob_{Y,Z|\theta}\bigl\{ \mathcal{K}(f_\theta(Z), \pi_Y(\theta)) \leq \alpha \}.
\end{align*}
Now, by the above properties of $f_\theta(Z)$ and $\pi_Y(\theta)$, it follows that $\mathcal{K}\{f_\theta(Z), \pi_Y(\theta)\}$ is stochastically no smaller than $\unif(0,1)$, uniformly in $\theta$.  Therefore, the right-hand side of the above display is upper-bounded by $\alpha$, hence strong prediction validity.

\subsection{Summary}

To summarize, I presented three different options above for construction of a predictive IM for $Z$, given $Y=y$, under fairly general models, i.e., no independence or identically distributed assumptions, only that $(Y,Z)$ are related to the same model parameter $\Theta$; extension beyond the parametric model case considered here will be discussed briefly in Section~\ref{S:semi}.  Of the three option, my recommendation is Option~3 as its one and only goal is strongly valid and efficient prediction.  It achieves this by following the Principle of Minimal Complexity---reducing the dimension as much as possible before carrying out the Choquet integration.  The other two constructions don't fully commit to the prediction task, they hold on to the option of making inference on $\Theta$ too, which is an added constraint that limits the data analysts' ability to reduce dimension.  Therefore, as I mentioned above, there will generally be an efficiency loss compared to Option~3 and the other two options.  This can be clearly seen in the results of Example~\ref{ex:normal.mean.pred} below.  

Option~1 is unique since it's designed specifically for cases where an IM for inference about $\Theta$, given $Y=y$, has been constructed, and then prediction is required as an afterthought.  In this case, the assumption is that the data analyst doesn't have access to the data that went into the construction of an IM for $\Theta$, he only has the IM output.  This constraint limits his ability to construct an efficient predictive IM in that he's unable to carry out the preliminary dimension reduction steps that lead to efficiency.  Moreover, it's unclear at this point whether validity can be achieved through a combination strategy like that presented above except under special conditions, e.g., conditional independence and vacuous prior assumptions.  Nevertheless, the use of various strategies to combine different valid IMs into a single valid IM is technically interesting and practically useful.  For example, in meta-analysis, one can imagine there being IM output produced and published independently by different research groups.  Then the goal might be to combine these various IMs for inference about the common parameter, or to predict the results of a new follow-up study.  So, I think this IM combination idea warrants further investigation.

\subsection{Examples}

\begin{ex}[Normal]
\label{ex:normal.mean.pred}
For a simple illustration, suppose that $(Y \mid \Theta=\theta) \sim \nm(\theta, n^{-1} \sigma^2)$ and $(Z \mid \Theta=\theta) \sim \nm(\theta, \sigma^2)$.  In this case, $\sigma > 0$ is taken to be fixed; analogous results can be obtained when $\sigma$ is unknown but the details are more involved.  So then $Y$ is just the minimal sufficient statistic for the normal mean model based on $n$ iid samples.  I'll also consider the vacuous prior case for the sake of comparison.  

The objective of this example is to illustrate and compare the different options for constructing a strongly valid predictive IM for $Z$.  I'll flesh out each of these constructions below in turn.  With a slight abuse of notation, I'll write $p_{Y|\theta}(y)$ and $p_{Z|\theta}(z)$ for the densities of $(Y \mid \Theta=\theta)$ and $(Z \mid \Theta=\theta)$. 
\begin{itemize}
\item[2.] For this option, I first construct a joint IM for $(Z,\Theta)$ and then marginalize out $\Theta$.  In this case, the relative likelihood function takes the form 
\[ R(y,z,\theta) = \frac{p_{Y|\theta}(y) \, p_{Z|\theta}(z)}{\sup_{x,\vartheta} p_{Y|\vartheta}(y) \, p_{Z|\vartheta}(x)} = \exp\Bigl\{ - \frac{n(y-\theta)^2 + (z-\theta)^2}{2\sigma^2} \Bigr\}. \]
Since $n(Y-\Theta)^2 + (Z - \Theta)^2$ is a pivot, it's easy to get the joint IM:
\begin{align*}
\pi_y(z,\theta) & = \prob_{Y,Z|\theta}\{ R(Y,Z,\theta) \leq R(y,z,\theta)\} \\
& = 1 - {\tt pchisq}\Bigl(\frac{n(y-\theta)^2 + (z - \theta)^2}{\sigma^2}, \, {\tt df} = 2\Bigr).
\end{align*}
Applying the extension principle to marginalize over $\Theta$ leads to
\[ \pi_y(z) = \sup_\theta \pi_y(z,\theta) = 1 - {\tt pchisq}\Bigl(\frac{(z - y)^2}{\sigma^2(1 + n^{-1})}, \, {\tt df} = 2\Bigr). \]
\item[3.] For this option, the strategy is to marginalizing before calculating the IM contour.  This amounts using a profile relative likelihood, which in this case is given by 
\[ R(y,z) = \frac{\sup_\theta p_{Y|\theta}(y) \, p_{Z|\theta}(z)}{\sup_{x,\theta} p_{Y|\theta}(y) \, p_{Z|\theta}(x)} = \exp\Bigl\{-\frac{(z - y)^2}{2\sigma^2(1 + n^{-1})} \Bigr\}. \]
Since $Z - Y$ is a pivot, I can easily get the contour 
\begin{align*}
\pi_y(z) & = \sup_\theta \prob_{Y,Z|\theta}\{ R(Y,Z) \leq R(y,z)\} \\
& = 1 - {\tt pchisq}\Bigl(\frac{(z - y)^2}{\sigma^2(1 + n^{-1})}, \, {\tt df} = 1\Bigr).
\end{align*}
\item[1.] Option~1 starts with possibilistic representations of ``$(Z \mid \Theta)$'' and ``$(\Theta \mid Y)$,'' and combines these into a sort of ``joint IM'' for $(Z,\Theta)$, given $Y$, and then marginalizing via the extension principle.  Here $Y$ and $Z$ are conditionally independent, so the combination strategy described above is appropriate.  The two possibilistic representations I'll take as the starting point are
\begin{align*}
f_\theta(z) & = 1 - {\tt pchisq}\Bigl( \frac{(z-\theta)^2}{\sigma^2}, {\tt df} = 1 \Bigr) \\
\pi_y(\theta) & = 1 - {\tt pchisq}\Bigl( \frac{n(\theta - y)^2}{\sigma^2}, {\tt df} = 1 \Bigr).
\end{align*}
Unfortunately, the combination and marginalization can't be done in closed-form, but it's not too difficult to carry out these steps numerically.  
\end{itemize} 
The expressions are sort of messy, so it'll help to be able to visualize the results of the three different constructions.  Figure~\ref{fig:normal.pred.comp} shows the three IM plausibility contour functions for $Z$, given $Y=y$, based on $y=0$, $n=5$, and $\sigma=1$.  In this case, we see that Options~1 and 2 are very similar, with Option~2 appearing to be slightly more efficient than Option~1, but the solution based on Option~3 is by far the most efficient.  That Option~3 is more efficient than Option~2 can actually be seen from the formulas above: the former has ``${\tt df}=1$'' while the latter has ``${\tt df} = 2$'', which explains Option~3's sharp peak compared to Option~2's rounded peak.  Note that the $100(1-\alpha)$\% prediction plausibility interval derived from the Option~3 solution agrees exactly with the standard prediction interval presented in textbooks, which is optimal in all the usual senses.  

Incidentally, a standard result in the Bayesian literature is that, with a default/flat prior for $\Theta$, the posterior predictive distribution is $(Z \mid Y=y) \sim \nm(y, \sigma^2(1 + n^{-1}))$.  This is the maximal inner probabilistic approximation of the predictive possibility measure derived in Option~3 above.  In other words, the probability-to-possibility transform of the default Bayes predictive distribution agrees exactly with the Option~3 solution.  
\end{ex}

\begin{figure}[t]
\begin{center}
\scalebox{0.65}{\includegraphics{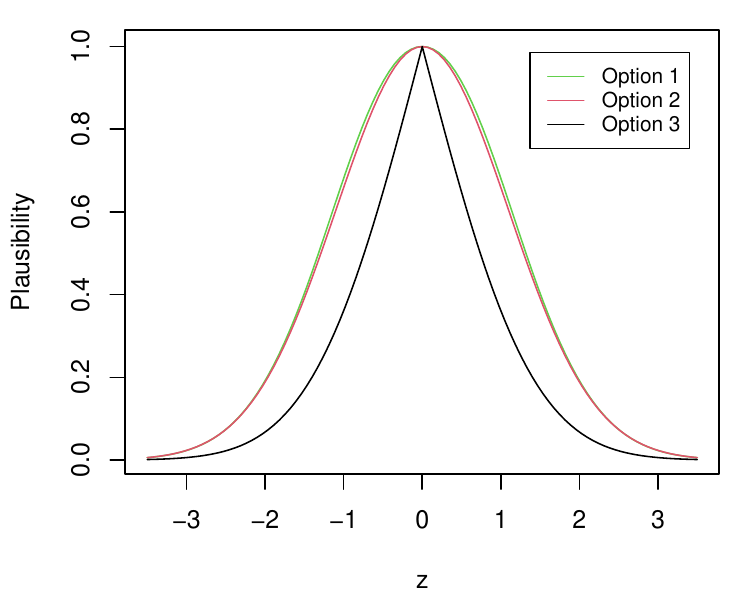}}
\end{center}
\caption{Plots of the predictive IM plausibility contour functions for the three different construction options described in the main text, based on $y=0$, $n=5$, and $\sigma=1$.}
\label{fig:normal.pred.comp}
\end{figure}

\begin{ex}[Multinomial]
\label{ex:mult.pred}
Inference and prediction in the multinomial model is a fundamentally important problem; see, also, Example~\ref{ex:mult.inference} below.  Let $(Y \mid \Theta=\theta)$ have a multinomial distribution, $\mult_K(n,\theta)$, where $K$ is the cardinality of the support, i.e., the number of categories, $n$ is the sample size, and $\theta=(\theta_1,\ldots,\theta_K)$ is a probability vector taking values the probability simplex in $\RR^K$.  A realization of $Y$ is just a frequency table $(Y_1,\ldots,Y_K)$ counting the instances of the $K$ categories in the sample, with the sum of those table entries equal to $n$.  The mass function of $Y$ is given by
\[ p_{Y|\theta}(y) \propto \prod_{k=1}^K \theta_k^{y_k}. \]
Let $(Z \mid \Theta=\theta) \sim \mult_K(1, \theta)$ denote a single independent realization from the same multinomial model.  Technically, $Z$ is a $K$-vector of 0's with a single entry equal to 1, but here I'll treat $Z$ as the position of the 1.  The goal is to construct a predictive IM for $Z \in \{1,2,\ldots,K\}$, given $Y=y$.  

I'll focus here on just the Option~3 construction, again with the vacuous prior for illustration.  The relative profile likelihood function in the context is 
\[ R(y, z) = \frac{\sup_\theta \theta_z^{y_z + 1} \prod_{k \neq z} \theta_k^{y_k}}{\max_\zeta \sup_\theta \theta_\zeta^{y_\zeta + 1} \prod_{k \neq \zeta} \theta_k^{y_k}}. \]
This expression looks a lot messier than it really is, since there are closed-form expressions for the optimization problems in both the numerator and the denominator, though these aren't worth displaying here.\footnote{Note that the optimization problem in the denominator amounts to a sort of entropy minimization.  Entropy is maximized by a uniform distribution, so minimization amounts to taking $\zeta$ to be the category with the largest probability, i.e., a rich-get-richer rule.}  Then the predictive IM has contour 
\[ \pi_y(z) = \sup_\theta \prob_{Y,Z|\theta}\{ R(Y,Z) \leq R(y,z)\}, \quad z \in \{1,2,\ldots,K\}. \]
For illustration, consider the application in \citet{goodman1965} and \citet{denoeux2006} with $n=220$ psychiatric patients and $K=4$ categories corresponding to four diagnoses: neurotic, depressed, schizophrenic, or having a personality disorder.  The observed counts are $y=(91, 49, 37, 43)$, so there's a clear tendency towards the first category, corresponding to neurosis.  Figure~\ref{fig:mult.pred} shows a plot of the IM's predictive plausibility contour.  As expected, the plausibility contour mode is the first category, but the other three categories have non-negligible plausibility too.  In fact, all reasonable $100(1-\alpha)$\% prediction sets for $Z$ contain all four categories. The same conclusion is drawn from other approaches, e.g., \citet{denoeux2006} develops a belief function fo prediction and, for small $\alpha$, the smallest set that his method assigns at least $1-\alpha$ belief to is all four categories.   
\end{ex}

\begin{figure}[t]
\begin{center}
\scalebox{0.65}{\includegraphics{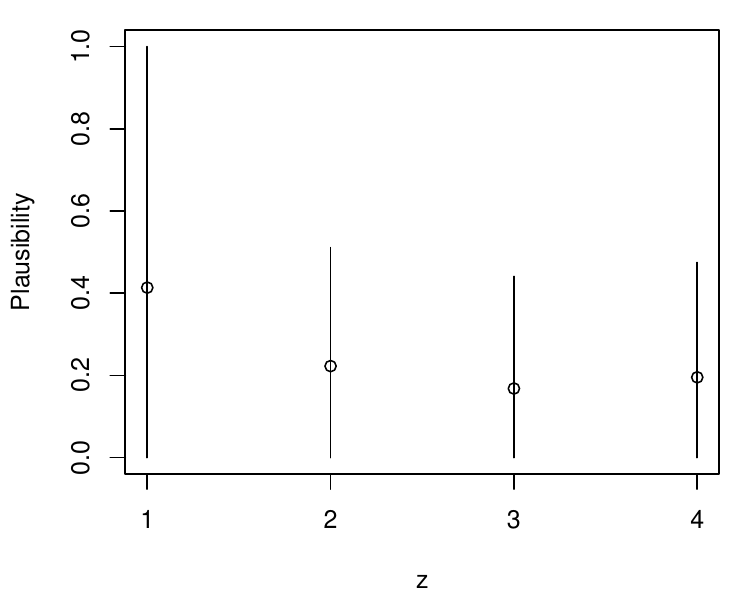}}
\end{center}
\caption{Plot of the predictive IM plausibility contour, $\pi_y(z)$, for the real-data illustration in Example~\ref{ex:mult.pred} involving $K=4$ categories. Circles correspond to the maximum likelihood estimator of $\theta=(\theta_1,\ldots,\theta_4)$.}
\label{fig:mult.pred}
\end{figure}

\begin{ex}[Gamma]
\label{ex:gamma.pred}
Consider a gamma model like in Example~\ref{ex:gamma.mean}, where $Y=(Y_1,\ldots,Y_n)$ denotes an iid sample of size $n$ from a gamma distribution with an unknown $\Theta=(\Theta_1,\Theta_2)$ consisting of the model's shape and scale parameter, respectively.  The focus here, however, is on prediction of a future observable, say, $Z$.  One case, of course, is where $Z=Y_{n+1}$ is a subsequent observation from the same gamma model, but that's not the only possibility.  Another practically relevant case, not uncommon in reliability applications \citep[e.g.,][]{hamada.etal.2004, wang.hannig.iyer.2012}, is that where $Z=\max\{Y_{n+1},\ldots,Y_{n+k}\}$ is the maximum of $k$-many future observations, with $k$ a given non-negative integer.  Note that $Z$ is independent of $Y$.  For this example here, I'll focus on ``Option~1'' where the IM plausibility contour $\pi_y(\theta)$ as presented in Example~10 of Part~II is combined with information about the $\theta$- and $k$-dependent distribution of $Z$ through the relationship \eqref{eq:K}, with the Fisher-based function ${\cal K}$ recommended above.  This is applicable since $Y$ and $Z$ are independent and I'm assuming, as before, that the prior information about $\Theta$ is vacuous.  

It's easy to check that the probability density function for $Z$, given $\Theta=\theta$ and a fixed $k$, is given by 
\[ p_\theta^{(k)}(z) = k \, p_\theta(z) \, \{1 - P_\theta(z)\}^{k-1}, \quad z > 0, \]
where $p_\theta$ and $P_\theta$ are the density and distribution functions of $\gam(\theta_1,\theta_2)$, with $\theta=(\theta_1, \theta_2)$.  Then the possibilistic representation of the distribution of $Z$ is 
\[ f_\theta(z) = \prob_{Z|\theta}\{ p_\theta^{(k)}(Z) \leq p_\theta^{(k)}(z) \}, \quad z > 0, \]
which, for given $\theta$ and $k$, can easily be evaluated based on Monte Carlo.  Plugging this and the joint contour $\pi_y(\theta)$ as presented in Example~10 of Part~II into the formula \eqref{eq:K} gives a (strongly valid) predictive IM with contour $\pi_y(z)$ for $Z$, given $Y=y$, with $k=3$.  For the same data as in Example~\ref{ex:gamma.mean}, Figure~\ref{fig:gamma2.pred} shows both the joint contour and this predictive plausibility contour. The sample mean is about 113 and the sample standard deviation is about 36, so one can't rule out the possibility that $Z$ is considerably larger than the maximum in the sample, which is 165.  So the fact that the IM's predictive contour has a long tail and the corresponding 90\% prediction interval stretches has upper bound near 244 is not unexpected.  
\end{ex}

\begin{figure}[t]
\begin{center}
\subfigure[Joint contour for $\Theta=(\text{shape, scale})$]{\scalebox{0.6}{\includegraphics{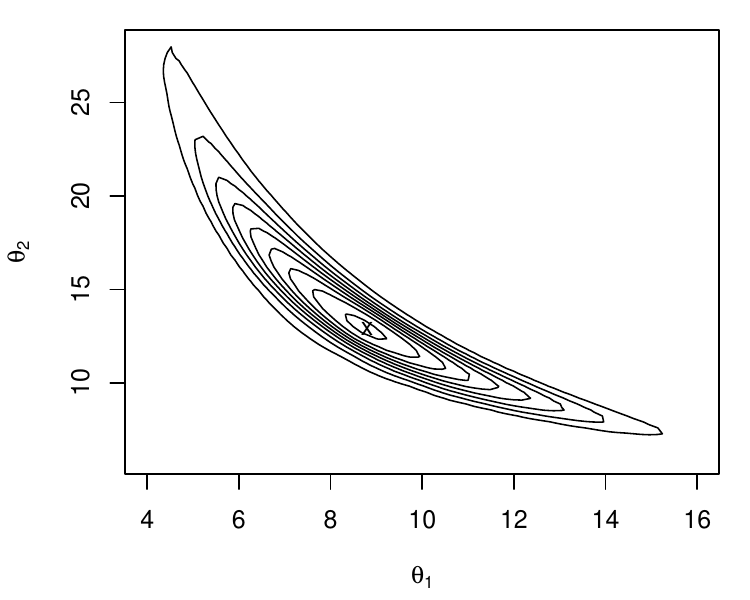}}}
\subfigure[Predictive contour for $Z$]{\scalebox{0.6}{\includegraphics{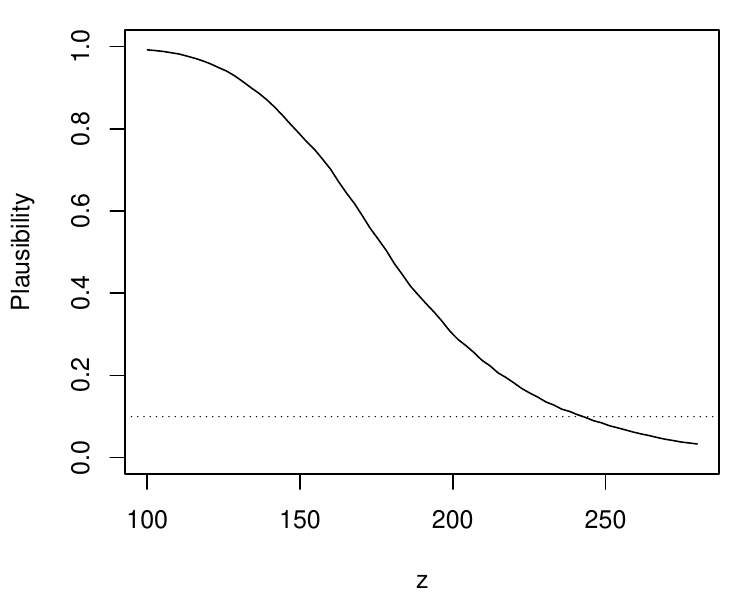}}}
\end{center}
\caption{Panel~(a) shows the joint IM contour for $\Theta$ based on the same rat survival time data analyzed in Example~\ref{ex:gamma.mean}.  Panel~(b) shows the IM's predictive contour for $Z$, the maximum of $k=3$ many future gamma observations.}
\label{fig:gamma2.pred}
\end{figure}

\section{Non-parametric possibilistic IMs}
\label{S:semi}

\subsection{Setup}

To set the notation and terminology, suppose that the distribution $\prob_{Y|\theta}$ of the data $Y$ is indexed by $\theta \in \TT$, where $\theta$ is an infinite-dimensional index such as the density.  The point is that I'm just using $\theta$ to label the distribution $\prob_{Y|\theta}$, so it's not a parametric model in the usual sense.  To make this point clear, I'll directly write $\theta$ for the density/mass function of $\prob_{Y|\theta}$ in some of the expressions below.  The {\em non-parametric} problem assumes that the distribution of $Y$ is unknown, which means there's an uncertain variable $\Theta \in \TT$ about which inference is to be drawn.  Just like in the previous sections of this paper, it'll often be the case that only some (finite-dimensional) feature $\Phi$ of $\Theta$ is of interest; that is, $\Phi = f(\Theta)$ for some functional $f: \TT \to \FF$.  


In addition, there might also be (partial) prior information about $\Theta$ or directly about $\Phi$ to be included in the model formulation.  If $\Theta$ denotes the density function, then prior information might come in the form of lending more credibility to densities that are smoother.  
An advantage, however, of the partial prior IM framework is that one need not say anything about the density $\Theta$ directly, available prior information about the relevant feature $\Phi$ can be incorporated directly.  Indeed, all one needs is a contour $q_\Phi(\phi)$ and they're ready to construct a valid, partial prior-dependent IM for $\Phi$.

\subsection{Valid IM construction}
\label{SS:np.construction}

None of the theory presented in the previous sections of the paper required that the unknowns were finite-dimensional, so there's nothing new here.  For the non-parametric case, if partial prior information for $\Theta$ is available and encoded in the contour $q_\Theta$, then I'll first construct the plausibility order 
\[ R_q(y, \theta) = \frac{\theta(y) \, q_\Theta(\theta)}{\sup_{\vartheta \in \TT} \{ \vartheta(y) \, q_\Theta(\vartheta)\}}, \quad (y,\theta) \in \YY \times \TT, \]
and then define the IM for $\Theta$---or, equivalently, for the distribution of $Y$ itself---as 
\[ \pi_y(\theta) = \uprob_{Y,\Theta}\{ R_q(Y,\Theta) \leq R_q(y,\theta)\}, \quad \theta \in \TT, \]
where $\uprob_{Y,\Theta}$ is the upper joint distribution of $(Y,\Theta)$ based on the non-parametric model and the partial prior information.  Obviously, since $\Theta$ is such a complex object and the partial prior is completely general, there's no opportunity for dimension reduction and efficiency gain---we simply get what we get.  But strong validity of this IM for $\Theta$ holds by the general theory in Part~II.  In case prior information about $\Theta$ is vacuous, then there's an opportunity to reduce dimension as before, 
\begin{equation}
\label{eq:pi.np.vac}
\pi_y(\theta) = \prob_{Y|\theta}\{ R(Y,\theta) \leq R(y, \theta)\}, \quad \theta \in \TT, 
\end{equation}
where, now, $R(y,\theta) = \theta(y) / \sup_{\vartheta \in \TT} \vartheta(y)$ is the non-parametric relative likelihood.  A complete prior can be handled in this non-parametric case too, as described in Part~II, i.e., rather than fixing the value of $\theta$ as in the above display, I fix the value of $y$.  The point is that, in the vacuous- or complete-prior cases, there's some opportunity for dimension reduction: the vacuous prior collapses the Choquet integral in the $\Theta$ dimension whereas the complete prior does so in the $Y$ dimension.  

Next, suppose the goal is inference on a finite- and probably relatively low-dimensional feature $\Phi = f(\Theta)$, e.g., a moment or quantile.  That the quantity of interest is low-dimensional creates an opportunity for efficiency gain compared to starting with the IM for $\Theta$ and marginalizing to $\Phi$ from there via the extension principle.  Following the general framework presented in Section~\ref{S:nuisance}, the construction of an IM for $\Phi$, based on partial prior information encoded in the contour $q_\Phi$, starts with a likelihood-driven plausibility ordering 
\[ R_q(y,\phi) = \frac{\{\sup_{\theta: f(\theta)=\phi} \theta(y)\} \, q_\Phi(\phi)}{\sup_\theta \{ \theta(y) \, q_\Phi(f(\theta))\}}, \quad (y,\phi) \in \YY \times \FF. \]
In what follows, it's important to remember that $\theta(y)$ represents $\prob_{Y|\theta}(\{y\})$, the distribution determined by $\theta$ evaluated at the data $y$.  So the optimization over $\theta$ in the above display is closely related to the so-called {\em empirical likelihood} framework described by, e.g., \citet{owen1988, owen1990, owen.book}, \citet{qin.lawless.1994}, \citet{tang.leng.2010}, and others.  Then the IM construction proceeds exactly as before, producing a contour 
\[ \pi_y(\phi) = \uprob_{Y,\Theta}\{ R_q(Y,f(\Theta)) \leq R_q(y,\phi)\}, \quad \phi \in \FF. \]
Strong validity holds by the general theory, so the critical question is how to carry out the necessary computations.  For this, let me consider the special case of vacuous prior information, so that the plausibility ordering takes the form 
\begin{equation}
\label{eq:rpel}
R(y,\phi) = \frac{\sup_{\theta: f(\theta)=\phi} \theta(y)}{\sup_\theta \theta(y)}, \quad (y,\phi) \in \YY \times \FF. 
\end{equation}
This is exactly the empirical likelihood ratio statistic discussed and analyzed extensively in \citet{owen.book}.  In that case, the IM contour function is 
\begin{equation}
\label{eq:pi.rpel}
\pi_y(\phi) = \sup_{\theta: f(\theta)=\phi} \prob_{Y|\theta}\{ R(Y,\phi) \leq R(y,\phi)\}, \quad \phi \in \FF, 
\end{equation}
which is simply the p-value function for the empirical likelihood ratio test of $H_0: \Phi=\phi$ based on the statistic \eqref{eq:rpel}.  Under regularity conditions \citep[e.g.,][Theorem~3.6]{owen.book}, when $Y=(Y_1,\ldots,Y_n)$ is an iid sample and $n \to \infty$, a version of Wilks's theorem applies to $-2$ times the log empirical likelihood ratio, that is, it has a limiting chi-square distribution that's independent of $\theta$.  Therefore, the strongly valid, vacuous-prior IM for $\Phi$ can be well approximated by a chi-square tail probability fairly generally.  Other kinds of approximations are possible as well; see Section~\ref{SS:approx}.  My point is that the very same general framework developed above and illustrated in simple, low-dimensional parametric inference problems can be readily applied to non-parametric problems.  The catch, however, is that computation is obviously a much more serious challenge in these high-complexity settings. I imagine that there will be cases in which it's much more efficient---computationally and/or statistically---to construct an approximate pivot by some other means than via the relative empirical likelihood.  The downside to a non-likelihood-based construction is that it's no longer clear how partial prior information can be incorporated in a principled way.  More on this in Section~\ref{S:beyond}. 

I should also mention briefly about how the ideas developed above would apply to the case of prediction without parametric model assumptions, as I hinted at in Section~\ref{S:prediction}.  Suppose, for simplicity, that $Y=(Y_1,\ldots,Y_n)$ and $Z=Y_{n+1}$ consist of iid samples from a distribution $\prob_{Y|\Theta}$, where $\Theta$ is uncertain, and the goal is prediction of $Z$.  Following the prediction strategy above, and assuming prior information about $\Theta$ is vacuous, I'd take the plausibility ordering to be 
\[ R(y,z) = \frac{\sup_{\theta \in \TT} \{\theta(y) \, \theta(z)\}}{\sup_{x \in \YY,\theta \in \TT} \{\theta(y) \, \theta(x)\}}, \quad (y,z) \in \YY^n \times \YY. \]
From here, it's straightforward to write down an expression for the predictive IM for $Z$, given $Y=y$, that's free of parametric model assumptions.  The problem, however, is that it's not clear how the predictive IM can be computed.  If it turned out that $R(Y,Z)$ were a pivot, with a known distribution independent of $\Theta$, then the predictive IM computations would be immediate.  Since it's not clear if/when this pivotal structure holds, one might consider an alternative strategy that sacrifices some of the efficiency-related benefits of working exclusively with likelihoods for the benefit of a pivotal structure that aids in computation; see Section~\ref{S:beyond} for more discussion on this.

\subsection{Approximations}
\label{SS:approx}

In the previous subsection I already suggested the option to approximate the IM's plausibility contour for $\Phi$, in the vacuous prior case, using the Wilks-like limiting distribution of the empirical likelihood ratio statistic.  One can also make certain adjustments (one-to-one transformations) to the plausibility ordering, e.g., Bartlett correction \citep[e.g.,][Sec.~3.3]{owen.book} that would improve the accuracy of of this limiting approximation.  But asymptotic approximations aren't the only approximation games in town and here I want to mention two such alternatives.  This discussion will be relatively high-level, focusing on the main ideas rather than details.  Also, for ease of connecting this new framework to the existing literature, I'll focus here exclusively on the vacuous-prior case; extending these ideas to handle the case with partial prior information is a huge open question that will have to be resolved elsewhere. 

The first of these alternatives is a relatively obvious one: just use the {\em bootstrap} \citep[e.g.,][]{efron1979, efrontibshirani1993, owen1988}.  There are, however, some subtleties in getting this to work properly, as I explain next.  To start, for simplicity, let $Y=(Y_1,\ldots,Y_n)$ consist of $n$ many iid observations from $\prob_{Y|\Psi}$ depending on the unknown (infinite-dimensional) $\Theta$.  Suppose the goal is inference on $\Phi=f(\Theta)$ and let $\hat\phi_Y$ denote the maximum likelihood estimator of $\Phi$, i.e., 
\[ \hat\phi_Y = f(\hat\theta_Y), \]
where $\hat\theta_Y$ is the (non-parametric) maximum likelihood estimator of $\Theta$.  For a bootstrap sample size $B$, let $y^b$ denote a random sample of size $n$, with replacement, drawn from the observed data values $y=(y_1,\ldots,y_n)$.  Then a bootstrap approximation of the IM contour $\pi_y(\phi)$ for $\Phi$, given $Y=y$, in \eqref{eq:pi.rpel} is 
\begin{equation}
\label{eq:pi.boot}
\pi_y(\phi) \approx \pi_y^\text{boot}(\phi) := \frac1B \sum_{b=1}^B 1\{ R(y^b, \hat\phi_y) \leq R(y,\phi)\}, \quad \phi \in \FF, 
\end{equation}
where $R$ is as given in \eqref{eq:rpel}.  Two remarks deserve to be made here.  
\begin{itemize}
\item First, note that $\hat\phi_y$ remains fixed as $b=1,\ldots,B$.  The reason is that the $y^b$'s represent samples from the ``population'' $y$ and the ``true $\Phi$'' corresponding to the ``population $y$'' is $\hat\phi_y$.  As the bootstrap story goes, if the $y$ sample is sufficiently informative, say, as $n \to \infty$, then the ``population $y$'' approximates $\prob_{Y|\Theta}$ and, therefore, the $R(y^b,\hat\phi_y)$'s are approximately representative samples of $R(Y,f(\Theta))$, so the approximation in \eqref{eq:pi.boot} should be relatively accurate. 
\vspace{-2mm}
\item Second, note that there's no ``$\sup_{\theta: f(\theta)=\phi}$'' in \eqref{eq:pi.boot} like there is in \eqref{eq:pi.rpel}.  The reason is that, technically, strong validity only requires that $\pi_Y(\Phi)$ be stochastically bounded at the ``true $\Phi$'' or, in this case, at $f(\Theta)$.  This control at the ``true value'' can't be achieved exactly with our less-than-fully-informative finite samples, so the supremum is a conservative adjustment to make up for this shortcoming.  In the Utopian bootstrap world, where $n \to \infty$ and samples are fully informative, the true values are recovered and there's no need for a supremum.  It's no different than the situation described above where the relative profile empirical likelihood is an asymptotic pivot and, therefore, the supremum over $\theta$ such that $f(\theta)=\phi$ drops out.  The point is that there's no finite-sample strong validity guarantees for the bootstrap approximation in \eqref{eq:pi.boot}, only asymptotically approximation strong validity as $n \to \infty$; see, e.g., \citet[][Theorem~1]{cella.martin.imrisk}.
\end{itemize} 

Another interesting but very different kind of approximation is that based on the so-called {\em universal inference} framework developed in \citet{wasserman.universal}.  The developments here are based on (profile) relative likelihood functions, but the previous make use of a split (profile) relative likelihood function, which I explain below.  Suppose, for simplicity, that $Y$ consists of $n$ many iid observations like above.  Then split this collection into two chunks, denoted by $Y^{(1)}$ and $Y^{(2)}$, where, for concreteness, $Y^{(1)}$ corresponds to the first $\lceil n/2 \rceil$ many observations and $Y^{(2)}$ the rest.  Let $\hat\theta_{Y^{(2)}}$ denote the maximum likelihood estimator of $\Theta$ based on the second chunk of data, i.e., 
\[ \hat\theta_{Y^{(2)}} = \arg\max_{\theta \in \TT} \theta(Y^{(2)}). \]
Now, for the quantity of interest $\Phi = f(\Theta)$, define the split-dependent plausibility ordering given by 
\[ R_\text{split}(y, \phi) = \frac{\sup_{\theta: f(\theta)=\phi} \theta(y^{(1)})}{\hat\theta_{y^{(2)}}(y^{(1)})}, \quad (y,\phi) \in \YY \times \FF. \]
This is just a ratio of $y^{(1)}$-data likelihoods, the numerator profiled over those $\psi$ satisfying $f(\theta)=\phi$ and the denominator evaluated at $\hat\theta_{y^{(2)}}$.  If both chunks of data are ``similarly informative,'' then one would expect $\hat\theta_{y^{(1)}} \approx \hat\theta_{y^{(2)}}$, in which case $R_\text{split}$ is just the relative profile likelihood based on $y^{(1)}$.  But note that $R_\text{split}(y,\phi)$ is not bounded above by 1; when this not-bounded-by-1 feature might be an issue, see below, it's easy enough to just truncate it at 1.  Applying the general IM construction above to this split-based plausibility order gives
\[ \pi_y(\phi) = \sup_{\theta: f(\theta)=\phi} \prob_{Y|\theta}\{ R_\text{split}(Y,\phi) \leq R_\text{split}(y,\phi)\}. \]
This is no more straightforward to compute than the original, no-split contour, but it's easy to approximate.  Indeed, by Markov's inequality, 
\begin{align*}
\pi_y(\phi) & = \sup_{\theta: f(\theta)=\phi} \prob_{Y|\theta}\{ R_\text{split}(Y,\phi) \leq R_\text{split}(y,\phi)\} \\
& = \sup_{\theta: f(\theta)=\phi} \prob_{Y|\theta}\{ R_\text{split}(Y,\phi)^{-1} \geq R_\text{split}(y,\phi)^{-1}\} \\
& \leq 1 \wedge \Bigl[ R_\text{split}(y,\phi) \, \sup_{\theta: f(\theta)=\phi} \E_{Y|\theta}\{ R_\text{split}(Y,\phi)^{-1}\} \Bigr] \\
& \leq 1 \wedge R_\text{split}(y,\phi),
\end{align*}
where I've made use of the key result in Equation~(6) of \citet{wasserman.universal}, a simple consequence of the law of iterated expectation, which states that the $\prob_{Y|\theta}$-expected value of $R_\text{split}(Y,\phi)^{-1}$ is no more than 1, uniformly in $\theta$ with $f(\theta)=\phi$.  Then the approximation I propose is to take 
\[ \pi_y^\text{split}(\phi) = 1 \wedge R_\text{split}(y,\phi), \]
which is relatively easy to compute---it only requires evaluating the split relative profile likelihood function, no probability calculations necessary.  Moreover, it follows from Theorem~3 in \citet{wasserman.universal} that the contour in the above display defines a strongly valid IM for $\Phi$, given $Y=y$.  This might appear too good to be true, but there's a price for the apparent simplicity: the data-splitting strategy generally results in a loss of efficiency.  This could be an acceptable trade-off in complex non-parametric problems where there might not be any other options for constructing a valid IM.


\subsection{Examples}

\begin{ex}[Multinomial]
\label{ex:mult.inference}
Reconsider the multinomial model from Example~\ref{ex:mult.pred}, with $\Theta=(\Theta_1,\ldots,\Theta_K)$ the $K$-dimensional vector in the probability simplex $\TT$, so that $\Theta_k$ denotes the probability of class $k$, for $k=1,\ldots,K$.  Since every discrete distribution on $\{1,\ldots,K\}$ can be described by such a $\Theta$ vector, I refer to this as the ``discrete non-parametric'' model.  It's for this reason that the multinomial model, while relatively simple, is of fundamental importance.  Many of the more general non-parametric developments, such as Bayesian non-parametrics via the Dirichlet process \citep[e.g.,][]{ferguson1973}, can be seen as extensions of the multinomial model.  

More specifically, let $(Y \mid \Theta=\theta) \sim \mult_K(n, \theta)$, so that $Y=(Y_1,\ldots,Y_K)$ where $Y_k$ is the sample frequency count for category $k$, with $k=1,\ldots,K$.  As before, this determines a likelihood function which, using the notation of this section, is given by $\theta(y) \propto \prod_{k=1}^K \theta_k^{y_k}$, for $\theta \in \TT$.  Assuming vacuous prior information for $\Theta$, just for simplicity, the plausibility ordering is determined by the relative likelihood alone, 
\[ R(y, \theta) = \prod_{k=1}^K \Bigl( \frac{n \theta_k}{y_k} \Bigr)^{y_k}, \]
where I've plugged in the likelihood function maximizer, which is available in closed form.  From here it's straightforward to evaluate the contour function of the (strongly valid) IM for $\Theta$, via Monte Carlo, using the formula \eqref{eq:pi.np.vac}.  

This is the same multinomial problem considered in a recent discussion paper \citep{gong.jasa.mult, gong.jasa.rejoinder} published in the {\em Journal of the American Statistical Association}.  They're approaching the problem from the perspective of Dempster--Shafer inference, so a comparison with the proposed solution here makes sense.  Jacob et al.~compared their solution to another with an IM-like flavor \citep{lawrence.etal.mult}.  Jacob, et al.~consider two real-data examples, both involving scientifically relevant questions concerning the multinomial parameter $\Theta$, and I'll reanalyze one of them here.  

In the $K=4$ case, \citet[][p.~368]{rao.linear.book} considered a so-called linkage model wherein the parameter $\Theta=(\Theta_1,\ldots,\Theta_4)$ has the following low-dimensional structure:
\[ \vartheta(\omega) = \Bigl( \frac12 + \frac{\omega}{4}, \frac{1-\omega}{4}, \frac{1-\omega}{4}, \frac{\omega}{4} \Bigr), \quad \omega \in (0,1). \]
In Rao's example, the categories represent four different phenotypes in animals, so the above model represents a simple(r) relationship between the proportions of these phenotypes in the animal population in question.  The data analyzed in \citet{gong.jasa.mult} and in \citet{lawrence.etal.mult} has observed cell counts $y=(25, 3, 4, 7)$.  My focus here is on the question of whether the above linkage model is {\em plausible} given the observed data.  So I have in mind the following assertion/hypothesis about the uncertain $\Theta$:
\[ A = \{\text{$\Theta = \vartheta(\omega)$ for some $\omega \in (0,1)$}\}, \]
and the goal is to calculate $\uPi_y(A)$, the IM's upper probability assigned to the above assertion; if this quantity is small, then I might be willing to reject the claim that the low-dimensional structure imposed by the linkage model is present in this data example.  By the IM's consonance structure, this amounts to solving a simple, one-dimensional optimization problem
\[ \uPi_y(A) = \sup_{\omega \in (0,1)} \pi_y\{ \vartheta(\omega) \}. \]
Figure~\ref{fig:mult.link} shows a plot\footnote{Note that this is different from marginalization via the extension principle. In marginalization, every value of the full parameter determines a value of the reduced parameter. Here, however, not all $\Theta \in \TT$ correspond to a $\vartheta(\omega)$ for some $\omega$.  This explains why the curve in Figure~\ref{fig:mult.link} doesn't reach the value 1 as we'd expect in the case of marginalization, e.g., Figure~\ref{fig:or}.} of $\omega \mapsto \pi_y\{\vartheta(\omega)\}$ and, clearly, the supremum is attained at $\omega \approx 0.61$ and the plausibility there is $\approx 0.97$. Here, obviously the plausibility is large, so the data shows effectively no signs of disagreement with Rao's linkage model that assumes $\Theta$ is of the form $\vartheta(\omega)$ for some $\omega \in (0,1)$.  This conclusion isn't surprising, given that Rao---who's a pretty smart guy---already took the linkage model as given for his analysis of these data.  From here, one can take the reduced model and carry out the IM construction to make inference on the uncertain linkage parameter $\Omega \in (0,1)$ directly.
\end{ex}

\begin{figure}[t]
\begin{center}
\scalebox{0.65}{\includegraphics{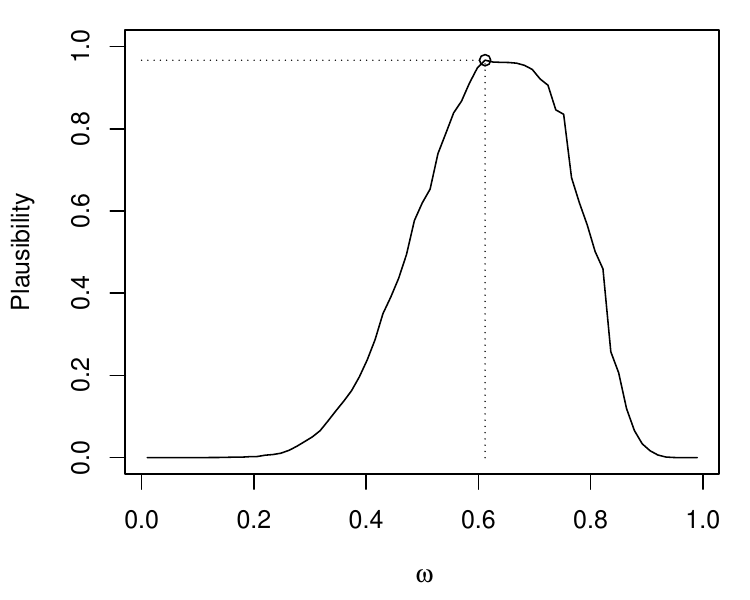}}
\end{center}
\caption{Plot of the function $\omega \mapsto \pi_y\{ \vartheta(\omega) \}$ in the multinomial/genetic linkage model in Example~\ref{ex:mult.inference}. The maximum value is $\uPi_y(A)$, the plausibility that the unknown $\Theta$ satisfies the linkage model constraint indexed by $\omega$.}
\label{fig:mult.link}
\end{figure}

\begin{ex}[Non-parametric quantile]
\label{ex:quantile}
Suppose the goal is inference on, say, the $r^\text{th}$ quantile $\Phi=\Phi_r$ of a completely unknown distribution, $\prob_{Y|\Theta}$, i.e., where $\Phi$ is such that 
\[ \prob_{Y|\Theta}(Y \leq \Phi) = r, \quad r \in (0,1). \]
Let $Y=(Y_1,\ldots,Y_n)$ be an iid sample of size $n$ from this unknown distribution.  Following the general framework as described above, I want $R(y,\phi)$ to be the empirical likelihood ratio statistic in \eqref{eq:rpel} which, in this case \citep[e.g.,][Theorem~5]{wasserman1990b}, is 
\[ R(y,\phi) = \Bigl\{ \frac{r}{u(y,\phi)} \Bigl\}^{u(y,\phi)} \Bigl\{ \frac{1-r}{n-u(y,\phi)} \Bigr\}^{n-u(y,\phi)}, \]
where
\[ u(y,\phi) = \begin{cases} |\{i: y_i \leq \phi\}| & \text{if $\phi < \hat\phi_y$} \\ nr & \text{if $\phi=\hat\phi_y$} \\ |\{i: y_i < \phi\}| & \text{if $\phi > \hat\phi_y$}, \end{cases} \]
and $\hat\phi_y$ the $r^\text{th}$ quantile of the sample $y$.  This is assuming, for computational simplicity, that the prior information about $\Phi$ is vacuous.  In this case, it's straightforward to get a bootstrap approximation, $\pi_y^\text{boot}(\phi)$, of the marginal IM's plausibility contour for $\Phi$ as in \eqref{eq:pi.boot}.  I simulated $n=25$ observations from a $\gam(3,1)$ distribution and that bootstrap-based approximate plausibility contour is plotted in Figure~\ref{fig:np.quantile}(a).  The stair-step pattern is a result of $R$ only depending on certain sample counts rather than the numerical values.  The true quantile in this case is $\approx 3.6$, which is right near the peak of the plausibility contour, as desired.  

Recall that this is only an approximate IM, so there's no guarantee that strong validity holds exactly.  It's not difficult, however, to get empirical confirmation that validity does hold, at least approximately.  Figure~\ref{fig:np.quantile}(b) shows a plot of the distribution function 
\begin{equation}
\label{eq:np.quantile.cdf}
\alpha \mapsto \prob_{Y|\theta}\{ \pi_Y^\text{boot}(f(\theta)) \leq \alpha \}, \quad \alpha \in [0,1], 
\end{equation}
and the fact that this curve falls below the diagonal line is an indication that strong validity holds.  This is only for one distribution, namely, $\gam(3,1)$ but the other experiments I conducted lead to the same conclusion, namely, that strong validity holds for the bootstrap-based approximation, even for relatively small $n$. 
\end{ex}

\begin{figure}[t]
\begin{center}
\subfigure[Plausibility contour]{\scalebox{0.6}{\includegraphics{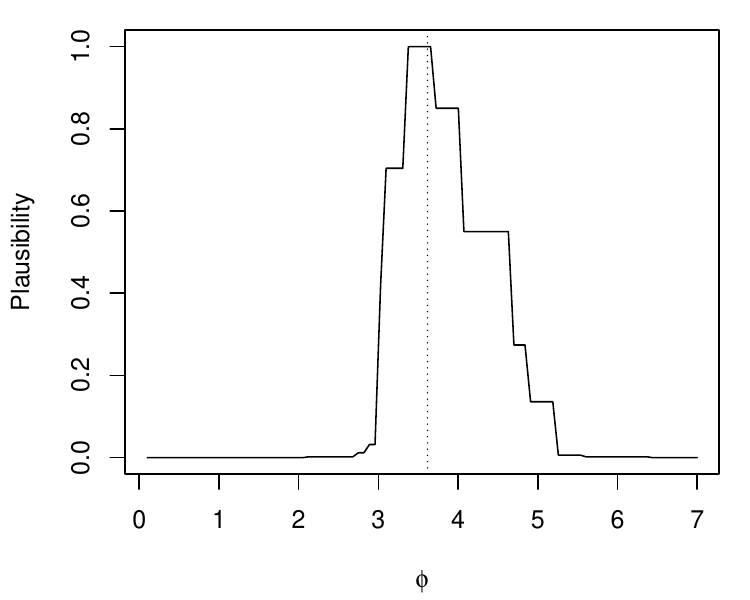}}}
\subfigure[Distribution function \eqref{eq:np.quantile.cdf}]{\scalebox{0.6}{\includegraphics{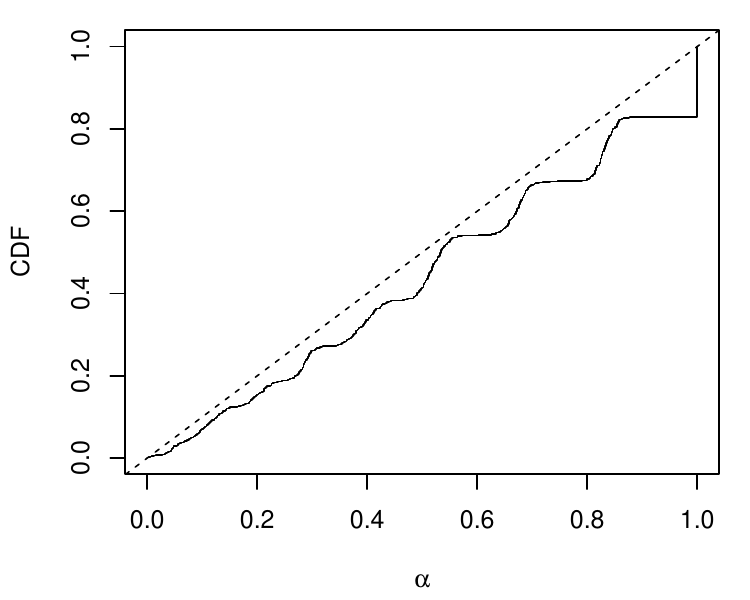}}}
\end{center}
\caption{Plots pertaining to Example~\ref{ex:quantile}.  Panel~(a) is the bootstrap approximation \eqref{eq:pi.boot} of the marginal IM's plausibility contour for $\Phi$ based on a sample of size $n=25$ from a $\gam(3,1)$ distribution; vertical line is the true quantile, $\approx 3.6$.  Panel~(b) shows a Monte Carlo approximation of the distribution function \eqref{eq:np.quantile.cdf}, for the same gamma model, and that the curve falls below the diagonal line confirms strong validity.}
\label{fig:np.quantile}
\end{figure}

\begin{ex}[Non-parametric mean]
\label{ex:mean}
Arguably the most fundamental problems in statistics is inference on the mean of a population based on random sampling.  Here I let $\Phi$ denote that unknown mean but I assume nothing more about the underlying distribution, $\Theta$, other than that its tails are such that it admits a finite mean.  For this case, assuming vacuous prior information about $\Phi$, just for simplicity, I can follow the IM construction in \eqref{eq:pi.rpel}, with $R(y,\phi)$ being the empirical likelihood ratio statistic for the mean, which is fleshed out in detail in, e.g., \citet[][Ch.~2.9]{owen.book}.  For the computation of $R(y,\phi)$, I used the function {\tt el.test} in the R package {\tt emplik} \citep{R:emplik}.  Using the bootstrap approximation suggested above, I found the plausibility contour $\pi_y^\text{boot}(\phi)$ for based on a real data set consisting of $n=29$ observations on the density of the Earth relative to water taken by Cavendish back in 1798 \citep[][Table~8]{stigler1977}.  The peak of the contour is, of course, at the sample mean $\bar y = 5.48$, and the circle marks the plausibility contour $\pi_y^\text{boot}(\phi^\star)$ at the ``true value'' of $\Phi$, which is $\phi^\star = 5.517$.  The horizontal line at $\alpha=0.05$ determines the upper-$\alpha$ level set, so it's clear that the true value is contained in the $100(1-\alpha)$\% plausibility region based on the analysis here. 
\end{ex}

\begin{figure}[t]
\begin{center}
\scalebox{0.65}{\includegraphics{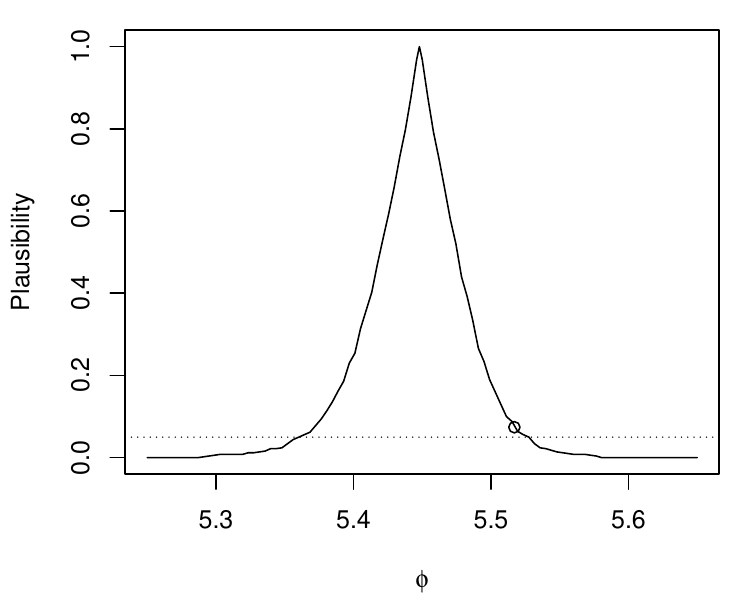}}
\end{center}
\caption{Plot of the IM plausibility contour for the mean $\Phi$, based on the bootstrap approximation \eqref{eq:pi.boot}, in the nonparametric case considered in Example~\ref{ex:mean}.  The data used here in this illustration are those measurements of the density of the Earth relative to water, taken by Cavendish back in 1798 \citep[][Table~8]{stigler1977}. The circle marks the plausibility evaluated at the ``true value'' of $\Phi$, which is 5.517.}
\label{fig:np.mean}
\end{figure}

\section{Possibilistic IMs without a likelihood}
\label{S:beyond}

\subsection{Setup}

In all of the cases discussed previously in the paper and almost all of the previous literature on IMs, the focus was on cases involving parametric models that connect the observable data to the unknown quantities of interest, e.g., model parameters and/or future observations.  But there is a wide class of classical and modern problems that don't fit this mold.  Perhaps the simplest of these problems is inference on an unknown quantile.  More specifically, suppose that $Y=(Y_1,\ldots,Y_n)$ is an iid sample from a distribution $\prob$ supported on $\RR$ and, for a given $q \in (0,1)$, the quantity of interest is the $q$-quantile of $\prob$, defined by the equation $\prob(Y_1 \leq \phi) = q$.  It's straightforward to estimate $\theta$, via the corresponding sample quantile, and, if desired, (approximate) confidence intervals are available.  But what about ``probabilistic inference'' in the sense that I'm concerned with here?  Do I first have to infer the whole (infinite-dimensional) $\prob$ and then marginalize to the scalar $\phi$?  \citet{wasserman.quote} describes this gap poignantly:
\begin{quote}
{\em The idea that statistical problems do not have to be solved as one coherent whole is anathema to Bayesians but is liberating for frequentists. To estimate a quantile, an honest Bayesian needs to put a prior on the space of all distributions and then find the marginal posterior. The frequentist need not care about the rest of the distribution and can focus on much simpler tasks.} 
\end{quote}
My claim is that it's the frequentists' implicit imprecision that's liberating---they leave unspecified (via vacuous models) those things that aren't relevant to their analysis.  My modest goal here is to suggest a framework that would allow probabilistic inference without the anathema, without the Bayesians' requirement to have a (precise) model for everything.  The full force of this will be developed in a subsequent, follow-up paper. 

Issues similar to those in the quantile problem arise is more modern problems.  Machine learning applications often start with a loss function $(y,\phi) \mapsto \ell_\phi(y)$ mapping data and decision rule pairs $(y,\phi) \in \YY \times \FF$ to a loss incurred by applying rule $\phi$ when data is $y$.  Then the data analyst's task boils down to estimation of and inference on the rule that minimizes the risk, or expected loss, i.e., $\Phi = \arg\min_{\phi \in \FF} \prob \ell_\phi$.  Alternatively, it might be that $\Phi$ is defined as the solution to a so-called estimating equation \citep[e.g.,][]{godambe1960, huber1981}, i.e., $\prob z_\Phi = 0$ for a given (possibly vector-valued) function $(y,\phi) \mapsto z_\phi(y)$. In any case, since a parametric statistical model isn't required to define the inferential target, the risk minimizer, estimating equation solution, etc., Manski's law \citep[][p.~1]{manski.book} dictates that the data analyst make as few model assumptions as possible.  Lots of problems, including the quantile example above, fit in this inference-on-risk-minimizers setting, so it's important to address this gap between frequentist, Bayesian, and other probabilistic inference frameworks.  Our work on {\em generalized posterior distributions} \citep[e.g.,][and the references therein]{syring.martin.mcid, syring.martin.scaling, gibbs.general, martin.syring.chapter2022, wu.martin.mcid} shows how to construct ``posterior distributions'' without a likelihood function, thus providing a generalization of Bayesian inference that's particularly suited for cases where the quantity of interest is a risk minimizer.  

The key point at the heart of Wasserman's remark and of what I'm suggesting here is that, while it's possible to connect the observable data $Y$ to quantity of interest $\Phi$ by thinking in terms of a non-parametric model (e.g., with an empirical likelihood as described in Section~\ref{S:semi}), this might not be the most statistically, computationally, or conceptually efficient solution.  As an alternative, one might consider defining a plausibility order for $\Phi$ in terms of a generic mapping $(y,\phi) \mapsto \rho(y,\phi)$ that makes no reference to a (empirical, marginal, or profile) likelihood.  Of course, this lacks the principles developed in Part~II for the case when a model/likelihood is available, but efficient marginal inference in non-parametric problems will likely require bending the rules a bit.  

Since it's currently not clear how one can incorporate partial prior information about $\Phi$ into these no-likelihood applications---that's an important open problem---I'll assume in what follows that the prior information is vacuous.  In that case, all we have to use is the mapping $\rho$, but the principles detailed in Part~II and applied above offer some guidance.  If $\Phi = f(\prob)$ is a relevant feature, some functional applied to the uncertain distribution $\prob$ for data $Y$, then I suggest constructing a (marginal) IM for $\Phi$ with contour 
\begin{equation}
\label{eq:nolik.contour}
\pi_y(\phi) = \sup_{\prob: f(\prob)=\phi} \prob\{ \rho(Y,\phi) \leq \rho(y,\phi)\}, \quad \phi \in \FF. 
\end{equation}
All the desirable properties of the IMs constructed in this manner carry over to this likelihood-free case.  In particular, the plausibility regions $\{\phi: \pi_y(\phi) > \alpha\}$ are exact $100(1-\alpha)$\% confidence regions.  If one could choose $\rho$ such that $\rho(Y,\phi)$ is a pivot under $\prob$ with $f(\prob) = \phi$, then this would be easy to implement via Monte Carlo.  While this can be done in certain applications \citep{cella.isipta23}, there are currently no broadly general strategies available for constructing pivots.  Below I'll highlight, in two practically relevant scenarios, where alternative strategies can be applied to make the above computation manageable, and with little or no sacrifice in validity.  These surveys are meant to just to give an idea of what's possible, further investigations are needed.

\subsection{Inference on risk minimizers}

As described above, for a given loss function $(y,\phi) \mapsto \ell_\phi(y)$, suppose that $\Phi$ is defined as the minimizer of the corresponding risk (expected loss) function, i.e., 
\[ \Phi = \arg\min_{\phi \in \FF} r(\phi), \quad \text{where} \quad r(\phi) = \prob \ell_\phi. \]
This is an unknown/uncertain quantity because $\prob$ itself is unknown/uncertain.  Since the goal is direct inference on $\Phi$, I don't want to introduce an indirectly-relevant ``model parameter'' $\Theta$ so that I can form a likelihood as in the previous sections.  Fortunately, the definition of $\Phi$ as a risk minimizer is enough structure to suggest a plausibility ordering $\rho$ and a corresponding marginal IM for $\Phi$. 

Let $Y^n=(Y_1,\ldots,Y_n)$ consist of an iid sample from $\prob$; note that these could be independent--dependent variable pairs with joint distribution $\prob$, but I'll not make this explicit in the notation.  For the observed $y^n$, the corresponding empirical risk is 
\[ r_{y^n}(\phi) = \frac1n \sum_{i=1}^n \ell_\phi(y_i), \]
and a natural estimate of $\Phi$ is obtained by minimizing the empirical risk:
\[ \hat\phi_{y^n} = \arg\min_\phi r_{y^n}(\phi). \]
Analogous to the relative likelihood plausibility ordering, I propose the following:
\[ \rho(y^n, \phi) = \exp[ -\{ r_{y^n}(\phi) - r_{y^n}(\hat\phi_{y^n}) \} ] \in [0,1]. \]
(The exponential form isn't necessary, that's just to make it resemble the relative likelihood.)  From here, one can define a marginal IM for $\Phi$ with contour as in \eqref{eq:nolik.contour}.  This is exactly the IM solution presented in \citet{cella.martin.imrisk}, and they proposed a bootstrap approximation analogous to that in Section~\ref{SS:approx} above to carry out the necessary computations.  With a bootstrap approximation, there's virtually no hope of having an exact validity result, but they proved an asymptotic validity theorem and demonstrated the IM's strong finite-sample performance in simulations.

\subsection{Prediction} 

Let's revisit the prediction problem discussed in a few places above.  Assume that $Y^n=(Y_1,\ldots,Y_n)$ consists of iid observations from a common distribution $\prob$, and that the goal is to predict the next observation $Y_{n+1}$.  In fact, I can be even more general and assume that the $Y$-process is exchangeable and that $\prob$ is the full joint distribution for the process.  This is an extreme case of marginal inference, where the entirety of the highly-complex $\prob$ is a nuisance parameter to be eliminated.  It may not be realistic/attractive to introduce a density function and a corresponding likelihood as suggested in Section~\ref{SS:np.construction} above, so here I'll avoid the use of likelihood.  A popular prediction method in the literature these days is {\em conformal prediction} \citep[e.g.,][]{vovk.shafer.book1, shafer.vovk.2008}.  A close connection between conformal prediction and IMs has already been demonstrated in \citet{imconformal, imconformal.supervised}, and what I present below offers some new perspectives.  

Let $\rho: \YY^n \times \YY$ be a mapping with two inputs: one is a data set and the other is a candidate value for the next observation.  Without loss of generality, I'll assume that $\rho(y^n, y_{n+1})$ is a measure of ``conformity'' of the candidate value $y_{n+1}$ with the data set $y^n$; that is, larger values correspond to $y_{n+1}$ that's consistent with the values in $y^n$.  For example, if $\hat y_{y^n}$ is a point prediction of the next observation, then the plausibility order $\rho$ can be defined as 
\[ \rho(y^n, y_{n+1}) = -d(\hat y_{y^n}, y_{n+1}), \]
where $d \geq 0$ is any suitable measure of distance between two points in $\YY$.  The only other constraint is that $\rho$ be symmetric in its first argument, i.e., the data set $y^n$ can be shuffled arbitrarily without affecting the value of $\rho$.  Having specified this plausibility order, the predictive IM contour for $Y_{n+1}$ is 
\[ \pi_{y^n}(y_{n+1}) = \sup_\prob \prob\{ \rho(Y^n, Y_{n+1}) \leq \rho(y^n, y_{n+1}) \}, \quad y_{n+1} \in \YY, \]
where the supremum is over all exchangeable joint distributions for the full $Y$-process.  Note that the probability calculation above is with respect to the joint distribution of $(Y^n, Y_{n+1})$ under $\prob$.  Even though this looks a little different than the setup above, all the same strong (prediction) validity properties hold, in particular, 
\begin{equation}
\label{eq:conformal.validity}
\sup_\prob \prob\{ \pi_{Y^n}(Y_{n+1}) \leq \alpha \} \leq \alpha, \quad \text{all $\alpha \in [0,1]$, all $n$}. 
\end{equation}
The problem, of course, is that the supremum makes evaluation of the IM contour unattainable.  A key point, however, is that the supremum also makes the resulting IM unnecessarily conservative.  By applying the {\em Principle of Minimal Complexity} from Part~II, it's possible to reduce the dimension of the aforementioned Choquet integral, which makes the computation simpler and the IM more efficient.  As explained in Part~II and applied in a few places above, the implementation of the {\em Principle} often boils down to conditioning on things that can be meaningfully conditioned on.  In this case, the structure of the problem makes it possible to condition on the {\em set of values} $\{y_1,\ldots,y_n,y_{n+1}\}$ while leaving their arrangement unspecified.  This set is, of course, a minimal sufficient statistic, so conditioning on this feature will eliminate the dependence on the unknown $\prob$, so the supremum drops out completely.  That is, the new predictive IM contour---based on conditioning and the aforementioned {\em Principle}---is given by
\begin{align*}
\pi_{y^n}(y_{n+1}) & = \sup_\prob \prob\bigl[ \rho(Y^n, Y_{n+1}) \leq \rho(y^n, y_{n+1}) \mid \{y_1,\ldots,y_n,y_{n+1}\} \bigr] \\
& = \frac{1}{(n+1)!} \sum_\sigma 1\{ \rho(y^{\sigma(1:n)}, y_{\sigma(n+1)}) \leq \rho(y^n, y_{n+1})\},
\end{align*} 
where the sum is over all $(n+1)!$ many permutations, $\sigma$, of the integers $1,\ldots,n,n+1$.  Finally, since $\rho$ is symmetric in its first argument, the right-hand side above can be further simplfied:
\[ \pi_{y^n}(y_{n+1}) = \frac{1}{n+1} \sum_{i=1}^{n+1} 1\{ \rho(y_{-i}^{n+1}, y_i) \leq \rho(y^n, y_{n+1})\}, \quad y_{n+1} \in \YY. \]
The reader will surely recognize the right-hand side above as the so-called ``transducer'' or ``p-value'' output produced by the inductive conformal prediction algorithm.  In particular, a result establishing what is equivalent to the prediction validity property in \eqref{eq:conformal.validity} can be found in Corollary~2.9 of \citet{vovk.shafer.book1}.  The derivation above, which makes use of conditioning and sufficiency of the empirical distribution more closely resembles that in \citet{faulkenberry1973} and, more recently, \citet{hoff2023}.  That one can arrive at the very powerful conformal prediction methodology through a (generalized---in the sense of allowing generic orderings $\rho$) IM-driven line of reasoning is quite remarkable.

\section{Conclusion}
\label{S:discuss}

This paper, Part~III of the series, considered the problem of efficient marginal inference on an interest parameter through a suitable elimination of the underlying nuisance parameters.  Depending on the problem at hand, this can be relatively straightforward or quite difficult (at least computationally).  For inference on parameters in a posited statistical model, if there's an ``ideal factorization,'' then valid and efficient marginal inference is almost immediate, through a general relative profile likelihood-based formulation.  Outside the ``ideal'' class of problems, the same proposal still works and is shown to very strong solutions in some challenging applications, namely, the gamma mean and Behrens--Fisher problem.  In fact, based on the results presented in Example~\ref{ex:bf} above, my conjecture is that the proposed IM solution is the best available among in the sense of being exactly valid and also empirically efficient.  There are certain cases where the profile relative likelihood strategy is inefficient, in particular, when there's a large number of nuisance parameters; but this risk can be anticipated, and other marginalization strategies can be applied, as I showed in Examples~\ref{ex:length}--\ref{ex:ns}.  

Prediction problems can be viewed as extreme cases of marginal inference, where all of the model parameters are nuisance and to be eliminated.  Here the same relative profile likelihood-based construction is possible, leading to what I called a predictive IM that is provably valid and, among other things, can be used to construct prediction regions for features of future observables.  For instance, in Example~\ref{ex:gamma.pred}, I showed how to construct an valid predictive IM for the maximum of the next $k$ realizations in a sequence of gamma observables.  This same problem has been investigated in, e.g., \citet{hamada.etal.2004},  \citet{wang.hannig.iyer.2012}, and \citet{impred}, but none of these proposals come equipped with exact prediction coverage guarantees.  

The first part of this paper, and most of the previous literature on IMs, focused on the case of a finite-dimensional parametric model for the observable data.  Section~\ref{S:semi} lays the groundwork for a new, (empirical) likelihood-based approach for marginal inference on certain features of the non-parametric model, indexed by an infinite-dimensional unknown.  In this case, as expected, computation is a more serious challenge, and there I put forward some first thoughts on efficient approximations via, say, bootstrap.  A few examples of this were presented, in particular, non-parametric inference on a mean and on a quantile.  There are other ``non-parametric'' problems in which it may be preferable to proceed without thinking in terms of an infinite-dimensional unknown, e.g., in machine learning problems where the quantity of interest is defined as a risk minimizer.  Section~\ref{S:beyond} briefly describes how an approach similar to what was developed in the first part of the paper can be applied even in this seemingly-very-different context.  In fact, for prediction, I showed how the powerful and now widely-used conformal prediction algorithm can be derived from (a slightly broader perspective on) this general IM framework.  

I'll conclude this discussion with a brief mention of some directions for future investigation.  These and/or other things will be addressed in subsequent parts of this series.  
\begin{itemize}
\item An important problem that's closely related to the elimination of nuisance parameters is {\em model assessment} and, in turn, the task of {\em model selection}.  The point is that, if the model and, as usual, the model parameters are both unknown, then there's really an uncertain pair $(\Gamma, \Theta_\Gamma)$, where $\Gamma \in \mathbb{G}$ is the uncertain model index and $\Theta_\Gamma \in \TT_\Gamma$ is the uncertain, model-specific parameter.  When it comes to model assessment, the entirety of $\Theta_\Gamma$ is a nuisance parameter and the goal is marginal inference on $\Gamma$.  From this perspective, it's only natural to consider the same relative profile likelihood-based IM construction presented here.  The result would be a strongly valid IM on the model space $\mathbb{G}$, offering provably reliable possibilistic uncertainty quantification about the model, something no other frameworks are able to offer.  This can also accommodate partial prior information, e.g., to encourage simplicity/sparsity/parsimony/etc. 
\item In the context considered in Section~\ref{S:beyond}, when there's no likelihood directly in consideration, it's no longer clear how to incorporate partial prior information.  For sure, it's not so simple as normalizing the likelihood times prior to get a relative likelihood function.  One of the challenges is in defining the upper joint distribution ``$\uprob_{Y,\Phi}$'' that would be used to carry out the Choquet integration.  While there might still be a work-around, I can also see that this difficulty is to be expected and perhaps insurmountable: the Choquet integral requires specification of an imprecise probability, which in turn requires a {\em probabilistic} link between data and parameters, hence a sort of model, likelihood, etc.  
\item Finally, computation of the marginal IM clearly is feasible in the examples presented here in this paper.  There's also lots of other similar examples where the same (naive) Monte Carlo-driven strategies can be put to work.  For problems that involve a lot of nuisance parameters, however, this might be quite expensive.  One option would be to give up some information/efficiency about the interest parameter in exchange for computational benefits, e.g., to work with a marginal- instead of profile-based relative likelihood in a ``not-so-ideal factorization'' case.  Another option is to develop some new and less-naive strategies for Monte Carlo-based optimization using, say, stochastic gradient decent.  Some initial work on this was presented in \citet{syring.martin.isipta21}, but I think more can be done. 
\end{itemize}

\section*{Acknowledgments}

Thanks to Leonardo Cella for helpful discussions and comments on an earlier draft.  This work is partially supported by the U.S.~National Science Foundation, SES--2051225.

\bibliographystyle{apalike}
\bibliography{/Users/rgmarti3/Dropbox/Research/mybib}

\end{document}